\documentclass[12pt,a4paper]{amsart} 
\usepackage[margin=2cm]{geometry}

\usepackage{microtype}
\usepackage[cal=euler]{mathalfa}
\usepackage{amsmath,amsthm,xspace}
\usepackage{tikz-cd}

\numberwithin{equation}{section}

\newcommand{\nn}{\nonumber}

\newcommand{\eps}{\epsilon}

\newcommand{\IR}{\mathds{R}}
\newcommand{\IC}{\mathds{C}}
\newcommand{\IF}{\mathds{F}}
\newcommand{\IZ}{\mathds{Z}}
\newcommand{\IN}{\mathds{N}}

\newcommand{\IP}{\mathds{P}}

\newcommand{\sgn}{\mbox{sgn}}

\newcommand{\cD}{\mathcal{D}}

\newcommand{\cC}{\mathcal{C}}
\newcommand{\cS}{\mathcal{S}}

\newcommand{\cM}{\mathcal{M}}

\newcommand{\cO}{\mathcal{O}}

\newcommand{\cX}{\mathcal{X}}

\newcommand{\cT}{\mathcal{T}}

\newcommand{\cI}{\mathcal{I}}

\newcommand{\de}{\mathrm{d}}
\newcommand{\De}{\mathrm{D}}
\newcommand{\I}{\mathrm{i}}

\newcommand{\Ext}{\mathrm{Ext}}
\newcommand{\Exp}{\mathrm{Exp}}
\newcommand{\Log}{\mathrm{Log}}
\newcommand{\Hom}{\mathrm{Hom}}
\newcommand{\coh}{\mathrm{coh}}

\def\bea{\begin{eqnarray}}
\def\eea{\end{eqnarray}}
\def\be{\begin{equation}}
\def\ee{\end{equation}}
\def\ba{\begin{align}}
\def\ea{\end{align}}
\def\bse{\begin{subequations}}
\def\ese{\end{subequations}}

\def\OmS{\Omega_{\rm S}}

\def\({\left(}
\def\){\right)}
\def\[{\left[}
\def\]{\right]}
\def\<{\left\langle}
\def\>{\right\rangle}

\def\tr{{\rm tr}}

\def\bOm{\overline{\Omega}}
\def\Omstar{\Omega_\star}

\RequirePackage{amsmath,amssymb}
\RequirePackage{etoolbox,bbm,xspace}
\RequirePackage{version,xparse}

\ifx\notheorems\undefined
\theoremstyle{plain}
\newtheorem{theorem}{Theorem}[section]

\newtheorem{lemma}[theorem]{Lemma}

\newtheorem{conjecture}[theorem]{Conjecture}
\theoremstyle{definition}
\newtheorem{remark}[theorem]{Remark}

\newtheorem{example}[theorem]{Example}

\theoremstyle{remark}

\fi

\def\mat#1{\ensuremath{#1}\xspace}
\def\dmat#1#2{\gdef#1{\mat{#2}}}
\def\csdmat#1#2{\csdef{#1}{\mat{#2}}}
\def\redef#1{\csletcs{t@#1}{#1}\csdmat{#1}{\csuse{t@#1}}}
\def\opr#1#2{\dmat#1{\operatorname{#2}}}

\def\oper#1{\csdmat{#1}{\operatorname{#1}}}

\makeatletter
\def\forlist#1#2{\@xp\forcsvlist\@xp{\@xp#1\@xp}\@xp{#2}}
\makeatother

\forcsvlist\oper{
Vect,Mod,Rep,Ab,Grp,Com,Set,Top,Sch,Coh,Qcoh,Ob,Ring,Alg,
Map,Hom,End,Aut,Ext,Cat,Sh,Ind,RHom,Fun,Var,rad,
Ker,Im,deg,Tr,det,rk,diag,im,coker,tr,cone,add,
sk,cosk,colim,Res,Gal,Coker,dlog,
GL,SL,PGL,Gr,Lib,Conv,Sol,Tor,
Max,Spec,Ass,Ann,Quot,Pic,Sym,MHS,MHM,
td,ch,Supp,supp,vol,trdeg,lc,NS,Loc,Iso,per,gr,Der,Lie,
Log,Exp,Pow,ord,lcm,VB,Hilb,cl,SU,SO,Ad,ad,id}
\opr\Id{id}
\opr\lHom{\lb{Hom}}
\opr\lAut{\lb{Aut}}
\opr\lIso{\lb{Isom}}
\opr\hgt{ht}
\opr\chr{char}
\opr\udim{\mathbf{dim}}
\opr\cls{\mathsf{cl}}
\opr\mmod{mod}
\opr\Et{\mathsf{Et}}

\def\mto{\mapsto}
\def\xto{\xrightarrow}

\def\emb{\hookrightarrow}

\dmat\es{\varnothing}

\def\sbs{\subset}

\redef{iff}
\dmat\imp{\implies}
\dmat\xx{\times}
\dmat\dd{\partial}

\def\iso{\simeq}
\def\ms{\backslash}
\def\bi{\mathbf i}
\def\bop{\bigoplus}
\def\ilim{\varprojlim} 
\def\ts{\otimes}

\def\lb{\underline}
\def\ub{\overline}

\def\oh{{\frac12}}
\dmat\Gm{\mathbb G_{\mathrm m}}
\dmat\one{\mathbbm1}
\def\ilim{\varprojlim} 

\def\what{\widehat}
\def\tl{\tilde}
\def\wtl{\widetilde}
\def\inv{^{-1}}
\def\dual{^{\vee}}
\dmat\pt{\mathbf{pt}}

\dmat\bul{\bullet}
\newdimen\bulh
\def\GIT{/\!\!/}

\def\defbb#1{\csdmat{b#1}{\mathbb{#1}}}
\def\defcal#1{\csdmat{c#1}{\mathcal{#1}}}
\def\deffrak#1{\csdmat{f#1}{\mathfrak{#1}}}
\def\defbf#1{\csdmat{bf#1}{\mathbf{#1}}}
\def\defsf#1{\csdmat{s#1}{\mathsf{#1}}}
\def\defall#1{\defbb{#1}\defcal{#1}\deffrak{#1}\defbf{#1}\defsf{#1}}
\forcsvlist\defall{A,B,C,D,E,F,G,H,I,J,K,L,M,N,O,P,Q,R,S,T,U,V,W,X,Y,Z}
\forcsvlist\deffrak{a,b,c,d,e,f,g,h,j,k,l,m,n,o,p,q,r,s,t,u,v,w,x,y,z}
\def\defbl#1{\csdmat{b#1}{\mathbf{#1}}} 
\forcsvlist\defbl{a,b,c,d,g,h,i,j,k,l,n,o,p,q,r,s,t,u,v,w,x,y,z}
\dmat\bk{\Bbbk}

\dmat\al\alpha
\dmat\ga\gamma
\dmat\de\delta
\dmat\la\lambda
\dmat\ta\tau
\dmat\hi\chi
\dmat\eps\varepsilon
\dmat\vi\phi
\dmat\te\theta
\dmat\ksi\xi
\dmat\si\sigma
\dmat\om\omega
\dmat\Ga\Gamma
\dmat\La\Lambda
\dmat\De\Delta
\dmat\Om\Omega
\dmat\Si{\Sigma}
\dmat\ze\zeta
\dmat\Te{\Theta}
\dmat\Na{\nabla}
\dmat\kp{\kappa}
\forcsvlist\redef{Phi,Psi,pi,mu,nu}

\def\set#1{\mat{\left\{#1\right\}}}
\def\sets#1#2{\left\{\left.#1\ \right\vert#2\right\}}
\def\ang#1{\mat{\left\langle#1\right\rangle}}

\def\rbr#1{\left(#1\right)}

\def\n#1{\left\lvert#1\right\rvert}

\def\ov#1#2{{\substack{#1\\#2}}} 
\def\pser#1{[\![#1]\!]} 
\def\lser#1{(\;\!\!\!(#1)\!\!\!\;)} 
\def\pmat#1{\begin{pmatrix}#1\end{pmatrix}}

\def\smat#1{\left(\begin{smallmatrix}#1\end{smallmatrix}\right)}

\edef\col{\mathchar\the\mathcode`:} 
\patchcmd\colon{:}{\col}{}{} 
\makeatletter
\begingroup
\catcode`:=\active 
\gdef:{\@ifnextchar={\col}\colon}
\endgroup

\AtBeginDocument{\mathcode`:="8000}

\def\qq#1{\lq#1\rq}
\def\tpdf{\texorpdfstring} 

\def\ie{i.e.\@\xspace} 
\def\cf{cf.\@\xspace} 
\def\eg{e.g.\@\xspace} 

\def\wrt{with respect to\xspace}

\def\defcite#1#2{
	\DeclareDocumentCommand{#1}{o}
		{\IfValueTF{##1}{\cite[##1]{#2}}{\cite{#2}}\xspace}
}

\DeclareDocumentCommand{\idef}{mo}{%
	\ifmmode#1\else\textit{#1}\fi%
	\IfValueTF{#2}{\index{#2}}{\index{#1}}%
}

\makeatletter
\def\leqnomode{\tagsleft@true\let\veqno\@@leqno}
\def\reqnomode{\tagsleft@false\let\veqno\@@eqno}
\makeatother

\newenvironment{ctikzcd}
	{\center\tikzcd}{\endtikzcd\endcenter}
\newenvironment{ctikz}
	{\center\tikzpicture}{\endtikzpicture\endcenter}

\theoremstyle{plain}

\DeclareDocumentCommand\sec{o}{
	\IfValueTF{#1}{\subsection{#1}}
			{\subsection{}\hspace{-1ex}}}

\defcite\MR{Mozgovoy:2008fd}

\newcommand\setpic[2]{
	\tikzset{#1/.pic={code={#2}}}}

\tikzset{ 
	crc1/.style={black,fill=yellow,radius=1cm},
	crc2/.style={black,fill=red,radius=1cm},
}

\tikzset{
crc3/.style={circle,draw=black,fill,inner sep=2pt},
crc4/.style={crc3,fill=white},
rcrc/.style={circle,fill=red,inner sep=2pt},
ar1/.style={->,blue,shorten >=2.5pt,shorten <=2.5pt,>=latex},
lbl/.style args={#1}{label=left:{$#1$}},
lbr/.style args={#1}{label=right:{$#1$}},
lbu/.style args={#1}{label=above:{$#1$}},
}

\setpic{triang}{
\begin{scope}
\draw[ar1](0,0)node[rcrc]{}--(1,2)node[rcrc]{};
\draw[ar1](1,2)--(2,0)node[rcrc]{};
\draw[ar1](2,0)->(0,0);
\draw[ar1](2,0)->(3,2)node[rcrc]{};
\draw[ar1](3,2)--(1,2);
\draw(0,1.4)node[crc3]{}--(1,.6)node[crc4]{}
--(2,1.4)node[crc3]{}--(3,.6)node[crc4]{};
\draw (2,1.4)node[crc3]{}--(2,2.6)node[crc4]{};
\draw (1,-.6)node[crc3]{}--(1,.6)node[crc4]{};
\end{scope}
}

\oper{wt}
\opr\NCDT{NC}

\def\Omstar{\Om_*}
\def\defbbI#1{\csdmat{I#1}{\mathbb{#1}}}
\forcsvlist\defbbI{T,R,C,F,Z,N,Q,H,O,P}

\oper{crit}
\oper{Arg}
\oper{Re}
\def\oprod{\prod^\curvearrowright}
\def\hOm{\cA}
\def\f{\mathsf f}
\dmat\Z{\mathsf Z}
\def\rf{\mathrm{ref}}
\oper{coh}
\oper{per}
\oper{Ch}
\oper{sgn}
\def\NC{\mathrm{NC}}
\defcite\MN{morrison_motivicb}
\def\NC{\mathrm{NC}}
\def\re{\mathrm {re}}
\def\im{\mathrm {im}}
\def\I{I}
\def\J{J}

\def\bibalias#1#2{
\csdef{b@#1}{\csname b@#2\endcsname}}

\def\dP{\mathrm{dP}}
\def\PdP{\mathrm{PdP}}

\usepackage[colorlinks,allcolors=blue]{hyperref}

\title{Attractor invariants, brane tilings and crystals}

\author{Sergey Mozgovoy}
\address{School of Mathematics, Trinity College Dublin, Dublin 2, Ireland
\newline\indent
Hamilton Mathematics Institute, Trinity College Dublin, Dublin 2, Ireland}
\email{mozgovoy@maths.tcd.ie}

\author{Boris Pioline}
\address{Laboratoire de Physique Th\'eorique et Hautes Energies, 
Sorbonne Universit\'e and CNRS UMR 7589, 
Campus Pierre et Marie Curie, 4 place Jussieu, F-75005, Paris, France}
\email{pioline@lpthe.jussieu.fr}

\bibalias{Beaujard:2020sgs}{beaujard_vafa}
\bibalias{Behrend:2009dc}{behrend_motivic}
\bibalias{MR2373143}{bridgeland_stability}
\bibalias{Morrison:2011bc}{morrison_motivicb}
\bibalias{Morrison:2011rz}{morrison_motivica}
\bibalias{Mozgovoy:2011ps}{mozgovoy_motivicb}
\bibalias{Mozgovoy:2011zh}{mozgovoy_wall-crossing}
\bibalias{1043.17010}{reineke_harder-narasimhan}
\bibalias{Szendroi:2007nu}{szendroi_non-commutative}
\bibalias{Young:2008hn}{young_generating}
\def\AA{\hat\bA}
\def\cbop{\what\bop}
\def\tcX{\widetilde{\mathcal X}}

\begin{document}
\begin{abstract}
Supersymmetric D-brane bound states on a Calabi-Yau threefold $\cX$ are counted 
by generalized Donaldsdon-Thomas invariants $\Omega_\Z(\gamma)$, depending on a 
Chern character (or electromagnetic charge) $\gamma\in H^*(\cX)$ and a stability condition (or central charge)  $\Z$.
Attractor invariants $\Omega_*(\gamma)$ are special instances of DT invariants, where
$\Z$ is the attractor stability condition $\Z_\gamma$ (a generic perturbation of self-stability),
from which DT invariants for any other stability condition can be deduced. 
While difficult to compute in general, these invariants become tractable
when $\cX$ is a crepant resolution of  a singular toric Calabi-Yau threefold 
associated to a brane tiling, and hence to a quiver with potential.
We survey some known results and conjectures about framed and unframed refined DT invariants in this context, and compute attractor invariants explicitly for a variety of toric Calabi-Yau threefolds,
in particular when $\cX$ is the total space of the
canonical bundle of a smooth projective surface,
or when $\cX$ is a crepant resolution of $\IC^3/G$.
We check that in all these cases, $\Omega_*(\gamma)=0$  unless $\gamma$ is the dimension vector of a simple representation or belongs to  the kernel of the skew-symmetrized Euler form.
Based on computations in small dimensions, we predict the values of all
attractor invariants, thus potentially solving the problem of counting DT invariants of these threefolds in all stability chambers.
We also compute the non-commutative refined DT invariants and verify that they agree with the counting of molten crystals in the unrefined limit.
\end{abstract}

\maketitle
\tableofcontents

\section{Introduction and summary}

Elucidating the microscopic origin of the entropy of black holes is a key objective of any putative theory of quantum gravity. String theory reached this milestone about 25 years ago, with the first quantitative description of the micro-states of BPS black holes (those black holes preserving a fraction of the supersymmetry of the vacuum), as supersymmetric bound states of D-branes wrapped around
various cycles of the internal Calabi-Yau (CY) three-fold $\cX$ \cite{Strominger:1996sh}. The indices counting (with sign) BPS bound states of arbitrary D-brane charge $\gamma\in H^*(\cX)$ were soon determined exactly for $\cX=T^6$, $\cX=K3\times T^2$ or orbifolds thereof preserving at least 16 supercharges (see e.g.\ \cite{Sen:2007qy} for a review). Despite some early successes, determining the exact BPS indices for a genuine CY3-fold with $\SU(3)$ holonomy, hence preserving 8 supercharges, is still an open problem, except for very special charges $\gamma$. Our goal in this article is to review recent
progress in the case of non-compact threefolds, and provide some support for recent conjectures on the BPS indices for local Fano surfaces in the so-called \qq{attractor} or \qq{self-stability} chamber \cite{Beaujard:2020sgs}. 

\medskip

Mathematically,  this physics problem amounts to computing the generalized Donaldson-Thomas  (DT) invariants $\Omega_{\Z}(\gamma)$ for arbitrary $\gamma\in H^*(\cX)$ 
and central charge 
$\Z\in \cS(\cX)$ in the space of Bridgeland stability conditions 
(or rather, its poorly understood subspace  $\cM(\cX)\subset \cS(\cX)$ spanned by string theory moduli) 
\cite{Douglas:2000gi,MR1403918,MR2373143,Aspinwall:2004jr}.
Informally, $\Omega_\Z(\gamma)$ is the weighted Euler characteristic of the moduli space of semistable objects with Chern character
$\gamma$ in the bounded derived category of coherent sheaves $D^b(\coh\cX)$. It is notoriously hard to calculate in general, except when $\gamma$ describes a rank-one coherent sheaf on $\cX$ (a D6-brane
in physics parlance) and $\Z$ is the large volume central charge, in which case $\Omega_\Z(\gamma)$ reduces to the usual Donaldson-Thomas invariant \cite{thomas1998holomorphic}, computable
from the Gromov-Witten invariants of $\cX$ \cite{gw-dt}. Furthermore, physicists are often 
not content knowing $\Omega_\Z(\gamma)$, but would also like to know its refined (or motivic) version $\Omega_\Z(\gamma,y)$, which counts  BPS states weighted by their angular momentum in 3 dimensions \cite{Dimofte:2009bv,Gaiotto:2010be}. The latter is known for the Hilbert scheme of $n$ points (corresponding to bound states of a single D6-brane with $n$ D0-branes) on an arbitrary $\cX$ \cite{Behrend:2009dc}, but hard to 
define for general compact CY3-folds and D-brane charges \cite{ks}.

\medskip

For a non-compact CY3-fold,
the derived category $D^b_c(\coh\cX)$ of (compactly supported) coherent sheaves 
becomes more tractable,
being equivalent to the bounded derived category of quiver representations  $D^b(Q,W)$
for a certain quiver with potential $(Q,W)$, the vertices correspond to a tilting sequence 
of line bundles on $\cX$
(or in physics parlance, a set of \qq{fractional branes}).
When $\cX$ is toric (i.e.\ has a $(\IC^\times)^3$-action), $(Q,W)$ can be more directly
obtained from a brane tiling \cite{Franco:2005rj,Hanany:2006nm,Mozgovoy:2008fd}.    
It is convenient to associate to  $\Z$ a vector of stability parameters $\theta$, known
in physics as Fayet-Iliopoulos (FI) parameters.
The  DT invariants $\Om_\te(\gamma)=\Omega_\Z(\gamma)$,
along with  their refined version $\Om_\te(\ga,y)=\Om_\Z(\ga,y)$,
can in principle be computed using 
representation theory techniques,
but the presence of a potential $W$ (called superpotential
by physicists) greatly complicates the matter. 
The framing induced by the presence of a non-compact D-brane (e.g.\ a D6-brane wrapped on $\cX$)
can be used to overcome these difficulties
\cite{Szendroi:2007nu}. 

\medskip

In general, the generalized DT invariant $\Omega_\te(\ga)$ or its refined version depend sensitively on the stability parameter $\theta$,
their jumps across walls of marginal stability in $\cS(\cX)$ 
being governed by a universal wall-crossing formula \cite{ks}. A notable exception is the case of small crepant resolutions (sometimes known as \qq{local curves}),
where the (unframed) quiver is symmetric
and the skew-symmetrized Euler form $\ang{-,-}$ 
(defined in \eqref{defskEuler} below, and also known as the Dirac-Schwinger-Zwanziger pairing)
 vanishes, so that DT invariants
are independent of $\theta$ (indeed, they count D2-D0 brane bound states,
which are generally free from wall-crossing and equal to Gopakumar-Vafa invariants). The framed indices counting D6-D2-D0 branes do depend on the stability parameters $(\te,\te_\infty)$, 
where $\te_\infty$ is the FI parameter for the framing vertex, but they vanish in a particular chamber and can be deduced elsewhere by the semi-primitive wall-crossing formula \cite{Denef:2007vg,Dimofte:2009bv, Mozgovoy:2011zh}, a special case of \cite{ks}. In the non-commutative (NC) chamber, obtained by crossing through all the walls associated
to D2-D0 bound states (see e.g.\ Figure \ref{fig:conifoldray} for the resolved conifold), they have a simple combinatorial interpretation as counting \qq{molten crystals}
in a certain `classical crystal' associated to the brane tiling \cite{Okounkov:2003sp,Iqbal:2003ds,Szendroi:2007nu,Mozgovoy:2008fd,Ooguri:2008yb,Aganagic:2010qr,Eager:2011ns}, or equivalently,  dimer configurations on the brane tiling itself \cite{Mozgovoy:2008fd}. In the large volume chamber,
obtained by crossing only those rays with negative D0-brane charge, they instead reproduce the usual DT invariants
computed by the (refined) topological vertex \cite{Aganagic:2003db,Iqbal:2007ii}.  
Motivic DT invariants for small crepant resolutions have been computed for arbitrary charge and stability condition in a series of mathematical works \cite{behrend_motivic,Morrison:2011rz,Mozgovoy:2011ps,Morrison:2011bc,mozgovoy_invariantsa}, as we shall review below.

\medskip

In contrast, for non-compact CY3-folds admitting compact divisors, such as the 
total space of the canonical bundle $K_S$ of a complex projective surface $S$ (sometimes called \qq{local surface}),
the unframed quiver is not symmetric, and the unframed indices $\Omega_\te(\ga)$ counting D4-D2-D0 bound states have a complicated dependence on the stability parameter $\te$.
One mathematically natural stability condition is the \qq{trivial stability}, where 
any quiver representation is considered to be semistable, but the corresponding moduli space is badly singular and its physical interpretation obscure. Another natural stability condition for the 
framed quiver is the NC chamber, where unrefined invariants can be computed by counting 
molten crystal or dimer configurations, just as for small crepant resolutions. Unfortunately, 
computing the refined invariants in this approach appears to be difficult. 

\medskip

Physics instead suggests a stability
condition $\Z_\gamma$ known as the \qq{attractor point}, which corresponds to the string moduli 
reached at the horizon of a spherically symmetric black hole with charge $\gamma$ \cite{Ferrara:1995ih}. Its key feature is that, in a neighborhood of $\Z_\gamma$ in $\cM(\cX)$, only single-centered black holes (and potentially, special multi-centered configurations 
known as \qq{scaling solutions} \cite{Bena:2012hf}) contribute to the index. 
Mathematically, $\Z_\gamma$ corresponds to
a generic perturbation $\te_\ga$ of the self-stability parameter $\ang{-,\ga}$ for quiver representations \cite{bridgeland2016scattering, Alexandrov:2018iao,Beaujard:2020sgs}.
Following \cite{Beaujard:2020sgs}, we define the attractor invariants (or attractor indices) as 
$\Omstar(\ga,y)= \Om_{\te_\ga}(\ga,y)$.
These attractor invariants
are closely related to the notion of initial data in the theory of  
wall-crossing structures and scattering diagrams \cite{kontsevich_wall,gross_canonical,bridgeland2016scattering}, 
which we use to show that $\Om_{\te_\ga}(\ga,y)$ is independent of the choice of  perturbation.
Provided the attractor indices $\Omstar(\ga,y)$ can be determined for all dimension vectors $\ga$, then the DT invariants $\Omega_\te(\ga,y)$, for any stability parameter $\te$, can be computed
by applying the wall-crossing formula repeatedly, or more efficiently by using the attractor flow
tree formula \cite{Denef:2001xn,Manschot:2010xp,Alexandrov:2018iao}.

\medskip

In \cite{Beaujard:2020sgs}, based on the study of D4-D2-D0 indices on $K_S$ and their relation to 
Vafa-Witten invariants of the surface $S$,  it was conjectured that the attractor indices $\Omstar(\ga,y)$
for quivers $(Q,W)$ associated to local Fano surfaces have a very simple structure:
\begin{conjecture}\label{C1}
$\Omstar(\ga,y)$ vanishes unless $\ga$ is associated to a simple representation of $Q$,
or $\ga$ lies in the kernel of the skew-symmetrized Euler form, that is $\ang{\ga,-}=0$;
\end{conjecture}
\begin{conjecture}\label{C2}
If $\ga$ is not associated to a simple representation of $Q$, then $\Omstar(\ga,y)$ vanishes
unless $\ga$ is a multiple of the dimension vector $\delta=(1,1,\dots,1)$ 
associated to a single D0-brane. 
\end{conjecture}
Note that Conjecture \ref{C2} is stronger than Conjecture \ref{C1}, since $\ang{-,-}$ 
has rank $2$ for local surfaces
and therefore has a large kernel, including the special vector $\delta$. Both conjectures were supported by
an analysis of the expected dimension of the moduli space of quiver representations in the attractor chamber, and by the fact  that the indices $\Omega_\te(\ga,y)$ in suitable chambers agree with known results for Vafa-Witten invariants on Fano surfaces. If correct, these conjectures open a path to compute refined invariants in any chamber, in particular framed DT invariants in the NC chamber. 

\medskip

In this work, after reviewing the results quoted above in more detail, we shall provide additional support for Conjecture \ref{C1}, and refine it as follows:

\begin{itemize}
\item[i)] For $\cX=K_S$ with $S$ one of the toric Fano surfaces $\IP^2, \IF_0, \dP_1, \dP_2$, 
we compute the attractor indices 
$\Omega_*(\ga,y)$ rigorously, for low values of $\ga$. Our method is to first 
evaluate DT invariants for trivial stability, by extending the \qq{double dimensional reduction}
method of \cite{Morrison:2011rz}, and then to  apply the Joyce-Reineke formula to reach
the attractor point; we find that  both Conjectures \ref{C1} and \ref{C2} are corroborated. The same
holds for some cases with more than one compact divisor, such as $Y^{3,2}$ and 
a crepant resolution of $\IC^3/\IZ_5$.

\item[ii)] More generally, for almost Fano surfaces such as $\IF_2$ or $\PdP_2$, or for
crepant resolutions of the orbifold $\IC^3/\IZ_r$, where $\bZ_r$ acts by
\begin{equation*}
1\mto\diag(\om,\om^k,\om^{-k-1}),\qquad
\omega=e^{2\pi \bi/r},
\end{equation*}
for some $0\le k<r/2$, we find that  Conjecture \ref{C1} is corroborated, while  Conjecture \ref{C2} fails when the quiver has cycles which do not pass through all the vertices. 

\item[iii)] In all considered cases, we find that the invariant $\Om_\te(n\delta,y)$ 
corresponding to $n$ D0-branes is independent of a stability parameter 
$\theta$ (provided the latter is generic) and of the integer $n\geq 1$, and is given by
\be
\label{OmBBS}
\Omstar(n\delta, y) = \Om_\te(n\delta, y)  
=-(y^3+(i+b-3)y+i y\inv),
\ee
where $i$ and $b$ are the numbers of internal and boundary lattice points in the toric diagram 
of \cX. Comparing with the virtual Poincar\'e polynomial $P(\cX;y)$, defined for any algebraic variety in 
\S\ref{sec:refined unframed}
and computed for any toric CY3-fold in Lemma \ref{lm:Ptoric}, we see that \eqref{OmBBS} 
agrees  with the motivic invariant $(-y)^{-3}P(\cX;y)$.
This statement can be deduced from the results of \cite{Behrend:2009dc} on the Hilbert scheme of $n$ points on $\cX$ (see Remark \ref{remarkBBS}).
For  other dimension vectors $\ga$ such that $\ang{\ga,-}=0$, corresponding to D2-D0 brane bound states wrapped on an exceptional curve, the index turns out to be either $0$ or $-y$. 

\item[iv)]  We verify these results by computing the framed DT invariants in the NC chamber using the 
attractor tree formula (or the Joyce-Reineke wall-crossing formula), and comparing against 
the unrefined DT invariants obtained by  counting molten crystals 
\end{itemize}

While we are not able to prove these conjectures yet, we note that for $\cX=K_{\IP^2}$, the vanishing of attractor invariants
for dimension vectors outside the kernel of $\ang{-,-}$ is closely related to the 
scattering diagram construction in \cite{Bousseau:2019ift}. We also note that 
some of the invariants above have been computed independently using exponential network 
techniques \cite{Banerjee:2018syt,Banerjee:2019apt,Banerjee:2020moh}, and agree
with our results  in the unrefined limit.

\medskip

From the physics viewpoint, we find it remarkable that the full BPS spectrum of D4-D2-D0 bound states
along with their framed analogues can be reconstructed from such simple data.  In particular, it shows that all
such BPS states  behave similarly to multi-centered black hole bound states, even though gravity is decoupled and none of the constituents carries enough entropy to form a black hole. An interesting
question is to further separate the attractor index $\Omstar(\gamma)$ into a single-centered invariant
$\OmS(\gamma)$ (also known as pure Higgs index) from additional contributions from \qq{multi-centered scaling solutions} \cite{Bena:2012hf,Lee:2012sc,Manschot:2011xc}. A preliminary analysis using the Coulomb branch formula (see \cite{Manschot:2014fua} and references therein) suggests that for D0-branes on $\cX$, 
\be
\label{OmSgen}
\OmS(n\delta,y) = \Omstar(n\delta,y) + i(y+1/y) = -y^3 - (b-3) y 
\ee
which depends only on the number $b$ of boundary points in the toric diagram. Note that 
$i=0$ for small crepant resolutions, where $\OmS(\gamma,y)$ and $\Omstar(\gamma,y)$ 
necessarily coincide for any $\gamma$ due to the vanishing of $\ang{-,-}$.

\medskip

Finally, we note that for most dimension vectors, the  motivic DT invariants fail to be invariant under $y\mapsto 1/y$. This failure can be traced to the D0-brane invariants $\Omstar(n\delta,y)$ in \eqref{OmBBS} (and D2-D0 branes on exceptional curves when applicable), and is a consequence
of the fact that Poincar\'e duality does not hold for cohomology with compact support on a non-compact
algebraic variety.
Physically, the relevant index should count $L^2$-normalizable bound states, and therefore
the $L^2$-cohomology of $\cX$ and related quiver moduli spaces \cite{Yi:1997eg,Sethi:1997pa}. 
It is natural to speculate that this count can be obtained by retaining the common terms in $\Omega_\te(\ga,y)$
and $\Omega_\te(\ga,1/y)$ as argued in \cite{Lee:2016dbm,Duan:2020qjy}. If so, it would follow from
\eqref{OmSgen} that the $L^2$-analogue of $\OmS(n\delta,y)$ vanishes, and that the 
$L^2$-analogue of $\Omstar(n\delta,y)$ equals $-i(y+1/y)$, where $i$ is the rank of the gauge group in M-theory compactified on $\cX$,
as suggested in \cite[\S 5.1]{Duan:2020qjy}. The independence on $n$ is of course the basic property which allows to view D0-branes as Kaluza-Klein gravitons in M-theory \cite{Witten:1995ex}.

\medskip

The remainder of this article is organized as follows. In \S\ref{sec_inv} we recall basic definitions and properties of DT invariants for quivers with relations, and the relation between framed and unframed invariants for symmetric quivers. In \S\ref{sec_wc} we extend these relations to non-symmetric quivers using wall-crossing formulae,  introduce the notion of attractor invariants, and present one version of the attractor tree formula. In \S\ref{sec:tilingcrystal} we review the relation between brane tilings and 
quivers with potential for singular CY 3-folds, and the relation between molten crystals and 
non-commutative DT invariants. In \S\ref{sec_small} we survey some known results about motivic DT invariants for small crepant resolutions, i.e.~toric CY 3-folds without compact divisors. Finally, in \S\ref{sec_surface} we present a method for computing motivic DT invariants for local   toric CY 3-folds with compact divisors, and apply it to compute attractor indices in a variety of examples, collecting evidence for the conjectures mentioned earlier.

\medskip

\noindent {\it Acknowledgments.} 
We are grateful to 
Sergey Alexandrov, Sibasish Banerjee, Guillaume Beaujard, Pierrick Bousseau, Pierre Descombes,
Sebastian Franco, Yang-Hui He, Pietro Longhi, Jan Manschot, Markus Reineke, Olivier Schiffmann,
Piljin Yi for useful discussions.

\vspace*{1cm}

\section{Invariants of quivers with potentials \label{sec_inv}}

In this section we review basic facts about representations of quivers with potential, and the
relations between framed and unframed, numerical and motivic Donaldson-Thomas invariants. 

\subsection{Jacobian algebras of quivers with potentials}
\label{sec:Jacobian}
Let $Q=(Q_0,Q_1,s,t)$ be a quiver with the set of vertices $Q_0$, the set of arrows $Q_1$, and with $s: Q_1\to Q_0$ and $t: Q_1\to Q_0$ being the source and target of the arrows.
The paths in $Q$ form a basis of the \idef{path algebra} $\bC Q$, with the composition defined by concatenation of paths.
We define the source and target of 
a path  $p=a_n\dots a_1$ to be $s(a_1)$ and $t(a_n)$, respectively. The path 
$p$  is called a \idef{cycle} if $t(p)=s(p)$.
We define the (cyclic) \idef{derivative} of $p$ \wrt $a\in Q_1$ to be
\begin{equation}
\frac{\dd p}{\dd a}=\sum_{a_i=a}a_{i+1}\dots a_na_1\dots a_{i-1}.
\end{equation}
 A \idef{quiver with potential}
is a pair $(Q,W)$, where $W$ is a linear combination of cycles in $Q$. We denote by $Q_2$
the set of cycles contributing to $W$.
We define the derivative $\frac{\dd W}{\dd a}$ of $W$ \wrt $a\in Q_1$ by linearity.
Define the \idef{Jacobian algebra}
$$J=J(Q,W)=\bC Q/\rbr{\dd W/\dd a\col a\in Q_1}.$$

\begin{example}\label{ex:C3}
Let $(Q,W)$ be the quiver with one vertex, $3$ loops $x,y,z$, and potential $W=xyz-xzy$:
 \begin{center}
\begin{tikzpicture}[inner sep=2mm,scale=2]
  \node (a) at ( 0,0) [circle,draw] {$1$};
    \draw [->] (a) edge[loop above] node{$x$} (a);
        \draw [->] (a) edge[loop left] node{$y$} (a);
            \draw [->] (a) edge[loop right] node{$z$} (a);
\end{tikzpicture}
\end{center}
Then 
\begin{equation}
\frac{\dd W}{\dd z}=xy-yx=[x,y]
\end{equation}
 and similarly for other derivatives.
The path algebra $\bC Q$ is the free algebra $\bC\ang{x,y,z}$ having non-commuting generators $x,y,z$,
while the Jacobian algebra $J(Q,W)$ is the polynomial algebra $\bC[x,y,z]$.
Note that this algebra is a coordinate ring of a CY3-fold $\bC^3$.
Usually our Jacobian algebras will be non-commutative, although they will still correspond to some CY3-fold.
\end{example}

Define a \idef{cut} of $(Q,W)$ to be a subset $I\sbs Q_1$ such that every nonzero term of $W$ contains exactly one arrow from $I$.
Given a cut, we define a \qq{partial} Jacobian algebra
\begin{equation}
J_I=J_I(Q,W)=\bC Q'/(\dd W/\dd a\col a\in I),\qquad Q'=(Q_0,Q_1\ms I).
\end{equation}
There is a natural monomorphism and a natural epimorphism of algebras
\begin{equation}
i:J_I\emb J,\qquad p:J\twoheadrightarrow J_I
\end{equation}
such that $p\, i=1$.

\begin{example}
Let $(Q,W)$ be as in Example \ref{ex:C3} and let $I=\set z$.
Then $J=\bC[x,y,z]$ and $J_I=\bC[x,y]$.
The monomorphism $i:J_I=\bC[x,y]\emb J=\bC[x,y,z]$ corresponds to the projection $\bC^3\to\bC^2$, $(x,y,z)\mto(x,y)$ and the epimorphism $p:J\to J_I$ corresponds to the inclusion $\bC^2\to\bC^3$, $(x,y)\mto(x,y,0)$.
\end{example}

\subsection{Stability conditions\label{sec_stab}}
A $Q$-representation is a tuple $((M_i)_{i\in Q_0},(M_a)_{a\in Q_1})$ such that $M_i$ is a finite-dimensional vector space for $i\in Q_0$ and $M_a:M_i\to M_j$ is a linear map for all arrows $a:i\to j$ in $Q$.
Given a representation $M$ of $Q$, we define its \idef{dimension vector} 
\begin{equation}
\udim M=(\dim M_i)_{i\in Q_0}\in\bN^{Q_0},\qquad \bN=\bZ_{\ge0}.
\end{equation}
Define a 
\idef{stability function} (or a central charge)
to be a linear map
\begin{equation}
\Z:\bZ^{Q_0}\to\bC
\end{equation}
such that, for all canonical basis vectors $e_i\in\bN^{Q_0}$,
\begin{equation}
\Z(e_i)\in\cH=\sets{re^{\pi \bi\vi}}{r>0,\vi\in(0,1]}\qquad
\forall i\in Q_0.
\end{equation}
We will often write $\Z(M)=\Z(\udim M)$ for a representation $M$.
For any $0\ne z\in\bC$, let $\Arg z$ be the unique $\vi\in(-\pi,\pi]$ such that $z\in\bR_{>0}e^{\bi\vi}$.
For any nonzero representation $M$, define its \idef{phase} to be
\begin{equation}
\vi(M)=\frac1\pi\Arg \Z(M)\in(0,1].
\end{equation}
A representation $M$ is called \idef{$\Z$-semistable} if, for any subrepresentation $0\ne N\sbs M$, we have
\begin{equation}
\vi(N)\le \vi(M).
\end{equation}

\begin{remark}
Let us relate the above definition of a stability function to the notion of a stability condition on a triangulated category \cite{bridgeland_stability}.
A \idef{Bridgeland stability condition} on a triangulated category $\cD$  is a pair $z=(\Z,\cA)$, where 
$\cA$ is the heart of a bounded $t$-structure on $\cD$ (the corresponding $t$-structure $\cD^{\le0}\sbs\cD$ is uniquely determined by \cA as the extension-closed subcategory generated by $\cA[n]$ for $n\ge0$) and $\Z:K(\cA)\to\bC$ is a linear map such that $\Z(F)\in\cH$ for all $0\ne F\in\cA$ and $\Z$ satisfies the Harder-Narasimhan (HN) property
\cite{bridgeland_stability}.
In our earlier discussion the required abelian category is the category of representations $\cA=\Rep Q$ 
(and $\cD=D^b(\cA)$).
The required linear map is the composition $K(\cA)\xto\udim\bZ^{Q_0}\xto\sZ\bC$.
\end{remark}

We can formulate the above semistability condition using slopes as follows.
Let us decompose
\begin{equation}
\label{Zthrho}
\Z=-\te+\bi\rho,\qquad \te,\rho:\bZ^{Q_0}\to\bR.
\end{equation}
where $\theta,\rho$ are real-valued linear forms.
As before, we will write $\te(M)=\te(\udim M)$, $\rho(M)=\rho(\udim M)$.
For any non-zero representation $M$, define its \idef{slope}
\begin{equation}
\mu(M)=\frac{\te(M)}{\rho(M)}
=-\frac{\Re \Z(M)}{\Im \Z(M)}
\in(-\infty,+\infty].
\end{equation}
Then
\begin{equation}
\vi(N)\le\vi(M)\iff \mu(N)\le\mu(M)
\end{equation}
and we can use the latter condition to test semistability.
We can change the stability function without changing the stability condition by considering
\begin{equation}
\Z'=\Z-c\rho=-(\te+c\rho)+\bi\rho,\qquad c\in\bR.
\end{equation}
In particular, assume that we have a representation $M$
with $\rho(M)\ne0$.
Then
$\te'=\te-\frac{\te(M)}{\rho(M)}\rho$
satisfies $\te'(M)=0$.
Therefore $M$ is semistable if and only if, for any subrepresentation $N\sbs M$, we have $\te'(N)\le 0$.
The elements $\te\in\Hom(\bZ^{Q_0},\bR)\iso\bR^{Q_0}$ are sometimes called \idef{Fayet-Iliopoulos parameters} or 
\idef{weights} (note that $\bZ^{Q_0}$ can be interpreted as the root lattice of the root system associated with~$Q$).
A representation $M$ is called \idef{$\te$-semistable} if $\te(M)=0$ and, for any subrepresentation $N\sbs M$, we have 
\begin{equation}
\te(N)\le 0.
\end{equation}
We similarly define $\te$-semistable representations of the Jacobian algebra.
Note that the condition of $\te$-semistability becomes trivial for $\te=0$.

\subsection{Moduli spaces of representations}
Given $d\in\bN^{Q_0}$, define the \idef{space of representations}
\begin{equation}
R(Q,d)=\bop_{a:i\to j}\Hom(\bC^{d_i},\bC^{d_j}).
\end{equation}
It is equipped with an action of the group $G_d=\prod_i\GL_{d_i}(\bC)$ (in physics, the complexified gauge group) so that the orbits correspond to isomorphism classes of representations having dimension vector $d$.

\medskip

Let $R(J,d)\sbs R(Q,d)$ be the closed subvariety of representations of the Jacobian algebra $J=J(Q,W)$, i.e.\ of representations satisfying $\partial W/\partial a=0$ for all $a\in Q_1$.
This subvariety is also equipped with an action of the group $G_d$.
We can interpret $R(J,d)$ as a critical locus of a map on $R(Q,d)$ as follows.
For any representation $M\in R(Q,d)$ and any cycle $p=a_n\dots a_1$, define
\begin{equation}
\tr(p|M)=\tr(M_{a_n}\dots M_{a_1})
\end{equation}
and define $\tr(W|M)$ by linearity.

\begin{lemma}[See \eg\ {\cite{segal_a}}]
Consider the map
\begin{equation}
\om_d =\tr(W|-):R(Q,d)\to\bC,\qquad M\mto\tr(W|M).
\end{equation}
Then $R(J,d)=\crit\om_d$.
\end{lemma}

For any weight $\te\in\bR^{Q_0}$ such that $\te\cdot d=0$, let 
\begin{equation}
R^\te(Q,d)\sbs R(Q,d),\qquad
R^\te(J,d)\sbs R(J,d)
\end{equation}
be open subsets of $\te$-semistable representations.
We consider moduli spaces
\begin{equation}
M^\te(Q,d)=R^\te(Q,d)\GIT G_d,\qquad
M^\te(J,d)=R^\te(J,d)\GIT G_d,
\end{equation}
where $\GIT$ denotes the good quotient \cite[\S4.2]{huybrechts_geometry}, 
and moduli stacks
\begin{equation}
\cM^\te(Q,d)=[R^\te(Q,d)/G_d],\qquad
\cM^\te(J,d)=[R^\te(J,d)/G_d].
\end{equation}
For $\te=0$ (trivial stability), we have $R(J,d)=R^0(J,d)$ and we define
\begin{equation}
M(J,d)=M^0(J,d),\qquad
\cM(J,d)=\cM^0(J,d).
\end{equation}
Similarly, we define $R^\Z(J,d)$, $M^\Z(J,d)$ and  $\cM^\Z(J,d)$, for any stability function $\Z$
using $\te$-semistability for $\te$ given by the decomposition \eqref{Zthrho}.

\subsection{NCDT and other numerical framed invariants}
Given a vector $\f\in\bN^{Q_0}$ (which we will call a \idef{framing vector}), define a \idef{framed representation} of $Q$ to be a representation $M$ of $Q$ equipped with an element
\begin{equation}
s\in\bop_{i\in Q_0}\Hom(\bC^{\f_i},M_i)
\iso\Hom\rbr{P,M},\qquad
P=\bop_{i\in Q_0}P_i^{\oplus \f_i}
,
\end{equation}
where $P_i$ is the indecomposable projective representation corresponding to a vertex $i\in Q_0$.
We can interpret a framed representation as a representation of a different quiver as follows.
Define a quiver $Q^\f$, called a \idef{framed quiver}, to be a quiver obtained from $Q$ 
by adding one vertex labelled by $\infty$ and $\f_i$ arrows $\infty\to i$, for all $i\in Q_0$.
Then a framed representation of $Q$ can be identified with a representation $M^\f=(M,M_\infty,s)$ of $Q^\f$ such that $\dim M_\infty=1$. Physically, the framing vertex corresponds to an infinitely heavy source, also known as defect, as we discuss in \S\ref{sec:tilings}. This motivates the label
$\infty$ for the framing vertex. 

\medskip

For any dimension vector $d\in\bN^{Q_0}$, let $d^\f=(d,1)\in\bN^{Q^\f_0}$.
Define the \idef{space of framed representations}
\begin{equation}
R^\f(Q,d)=\bop_{i\in Q_0}\Hom(\bC^{\f_i},\bC^{d_i})\oplus R(Q,d)=R(Q^\f,d^\f).
\end{equation}
Similarly, define $J^\f=J(Q^\f,W)$ and 
\begin{equation}
R^\f(J,d)
=\bop_{i\in Q_0}\Hom(\bC^{\f_i},\bC^{d_i})\oplus R(J,d)
=R(J^\f,d^\f). 
\end{equation}

A framed representation $M^\f$ having dimension vector $d^\f=(d,1)$ is called \idef{\NCDT-stable} (where  \NCDT stands for \qq{non-commutative}) if it is generated by $M_\infty$.
This means that for any subrepresentation $N\sbs M^\f$ with $N_\infty\ne0$, we have $N=M^\f$.
This stability corresponds to the weights
\begin{equation}
\te^\f=(\te,-\te\cdot d),\qquad \te\in\bR^{Q_0}_{<0},
\end{equation}
of the framed quiver, satisfying $\te^\f\cdot(d,1)=0$.
For convenience, let us choose a specific vector in this chamber
\begin{equation}\label{NCDT stability}
\te^\f_{\NCDT}=\rbr{-\de,\de\cdot d},\qquad
\de=(1,\dots,1)\in\bZ^{Q_0}.
\end{equation}

Let 
\begin{equation}
R^{\f,{\NCDT}}(J,d)=R^{\te^\f_{\NCDT}}(J^\f,d^\f)\sbs R^\f(J,d)
\end{equation}
denote the subspace of \NCDT-stable framed representations
and
\begin{equation}
M^{\f,{\NCDT}}(J,d)=R^{\f,{\NCDT}}(J,d)/G_d
\end{equation}
be the corresponding moduli space (a geometric quotient).
Then we define the partition function of \idef{numerical NCDT invariants}, 
or \idef{unrefined NCDT partition function}, as the formal sum
\begin{equation}\label{part fun1}
Z_{\f,\NCDT}(x)
=\sum_{d\in \IN^{Q_0} }
e(M^{\f,{\NCDT}}(J,d),\nu)x^d,
\end{equation}
where $x^d = \prod_{i\in Q_0} x_i^{d_i}$, and $e(X,\nu)=\sum_{k\in\bZ} k \, e(\nu\inv(k))$ is the 
weighted Euler number of $X$, with $\nu:X\to\bZ$ the Behrend
function \cite{behrend_donaldson-thomas}.
We will discuss later more explicit formulas for this expression.
Note that if $X$ is smooth, then $e(X,\nu)=(-1)^{\dim X}e(X)$.
Usually we will consider the basic framings $\f=e_i$ for $i\in Q_0$, and then denote $Z_{\f,\NCDT}$ by $Z_{i,\NCDT}$.

\medskip

More generally, for any weight $\te\in\bR^{Q_0}$,
consider the weight
\begin{equation}
\label{framed te}
\te^\f=(\te,-\te\cdot d)
\end{equation}
of the framed quiver.
It satisfies $\te^\f\cdot (d,1)=0$.
A framed representation $M^\f$ having dimension vector $(d,1)$ is $\te^\f$-semistable if and only if for any unframed subobject $N\sbs M^\f$ we have $\te(N)\le 0$ and for any unframed quotient $M^\f\twoheadrightarrow N$, we have $\te(N)\ge0$.


\medskip

Let us assume that $\te\cdot d'\ne0$ for all $0< d'\le d$, so that the inequalities above are strict.
Under this assumption the action of $G_d$ on $R^{\f,\te}(J,d)=R^{\te^\f}(J^\f,d^\f)$ is free and we define the moduli space
\begin{equation}
M^{\f,\te}(J,d)=R^{\f,\te}(J,d)/G_d
\end{equation}
and the partition function of \idef{framed numerical 
DT invariants} (assuming that $\te\cdot d\ne0$ for all $d\ne0$)
\begin{equation}\label{part fun2}
Z_{\f,\te}(x)
=\sum_{d\in \IN^{Q_0} }
e(M^{\f,\te}(J,d),\nu)x^d.
\end{equation}
As before, we denote $Z_{e_i,\te}$ for the basic framings by $Z_{i,\te}$.

\begin{example}
Let $Q$ be a quiver with one vertex and no loops.
The potential is necessarily trivial.
Consider a framing vector $\f\in\bN$.
For any $d\in\bN$, the space of framed stable representations
$R^{\f,\NCDT}(J,d)=R^{\f,\NCDT}(Q,d)$ is given by the set of epimorphisms $s:\bC^\f\to\bC^d$,
which is empty unless $d\leq \f$.
Taking the quotient by $G_d=\GL_d(\bC)$, we obtain for $d\leq \f$
the Grassmannian
\begin{equation}
M^{\f,\NCDT}(J,d)=\Gr(\f,d).
\end{equation}
Its dimension is equal to $d(\f-d)$
and its Euler number is the binomial coefficient $\binom \f d$.
Therefore
\begin{equation}\label{ncdt vect spaces}
Z_{\f,\NCDT}(x)=\sum_{d\in\bN}(-1)^{\f d+d}\binom \f dx^d
=(1-(-1)^{\f}x)^\f.
\end{equation}
\end{example}

\subsection{Refined unframed invariants}
\label{sec:refined unframed}
In this section we will define refined invariants of the stacks $\cM(J,d)$ of representations of the Jacobian algebra $J=J(Q,W)$.
As we will see in the next section, these invariants can be used to compute framed numerical invariants defined in \eqref{part fun1} and \eqref{part fun2}.

\medskip

For any algebraic variety $X$, 
let $P(X;y)$ denote its virtual Poincar\'e polynomial (sometimes we will denote it by $[X]$, the motivic class of $X$).  It is additive with respect to complements and, for a smooth projective variety~$X$, equals 
\begin{equation}
\label{defP}
P(X;y)=\sum_{i=0}^{\dim X} \dim H^i(X,\bQ) \,(-y)^i.
\end{equation}
In that case the Laurent polynomial $(-y)^{-\dim X}P(X;y)$ is symmetric under $y\to 1/y$, and can be interpreted as a character of the $\SU(2)$-Lefschetz action on $H^*(X,\bC)$; if $X$ is the moduli
space of some D-brane bound state localized in $\IR^3$, then the Lefschetz action realizes the rotations in $\IR^3$. 
If $X$ is non-projective (but has pure Hodge structure), then $P(X;y)$ is defined using cohomology with compact support,
\begin{equation}
P(X;y)
=\sum_{i=0}^{\dim X} \dim H^i_c(X,\bQ) \,(-y)^i,
\end{equation}
and the Laurent polynomial $(-y)^{-\dim X}P(X;y)$ need no longer be symmetric under $y\to 1/y$.
In either case, the specialization of $P(X;y)$ at $y=1$ gives the Euler number.
Note that for the affine plane  $P(\bA^1;y)=
P(\bP^1;y)-1=y^2$
can be identified with the number $q$ of points of $\bA^1$ over a finite field $\bF_q$.
We shall often omit the second argument and use $q$ and $y^2$ interchangeably.

\medskip

For any global stack $[X/G]$ (with an appropriate group~$G$)
we define $P([X/G])=P(X)/P(G)$.
Note that this class usually no longer specializes to Euler numbers but rather
has a pole at $y=\pm 1$, due to the vanishing of 
\begin{equation}
\label{PGLn}
P(\GL_n)=q^{n^2}(q\inv)_n,\qquad (q)_n=(q;q)_n=\prod_{i=1}^n(1-q^i),
\end{equation} 
at $q=y^2=1$.

\medskip

Let us now assume that the weight $\te$ is  such that $\te\cdot d\ne0$ for all $d\in\bN^{Q_0}\ms\set0$.
We have seen that in this case the action of $G_d$ on $R^{\f,\te}(J,d)$ is free.
Therefore we can replace the corresponding stack by the moduli space $M^{\f,\te}(J,d)$ and define the
partition function of \idef{refined framed DT invariants} informally as 
\begin{equation}
\label{partfun2ref}
Z_{\f,\te}^\rf(x)
=\sum_{d\in \IN^{Q_0}} (-y)^{\hi_Q(d,d)-\f\cdot d}
P(M^{\f,\te}(J,d)) \, x^d.
\end{equation}
where $\hi_Q$ is the Euler form defined in \eqref{defEuler}.
The partition function of numerical framed DT invariants 
\eqref{part fun2} is obtained by specialization at $y=1$.
In particular, we define $Z_{\f,\NCDT}^\rf=Z_{\f,-\de}^\rf$,
\cf\eqref{NCDT stability}.

\begin{remark}
\label{rm:exp motive}
The definition \eqref{partfun2ref} is not precise. More rigorously, one needs to consider the exponential motivic class of the map $\om_d:M^{\f,\te}(Q,d)\to\bA^1$ (note that $M^{\f,\te}(J,d)$ is its critical locus) instead of the invariants used above. This technicality can be circumvented when the quiver has a cut $I\subset Q_1$ (see \S\ref{sec:Jacobian}), which is the case in all cases that we consider.
Note also that
\begin{equation}
\f\cdot d-\hi_Q(d,d)
=\dim R^\f(Q,d)-\dim G_d=\dim M^{\f,\te}(Q,d).
\end{equation}
\end{remark}

Next we define unframed DT invariants.
Due to the fact that $G_d$ on $R(J,d)$ (or $R^\te(J,d)$) generally acts non-freely, 
we can only consider refined invariants for the corresponding stacks, and these invariants may not have a specialization at $y=1$. For trivial stability condition, 
we define the generating function of \idef{unframed refined invariants} (often called \qq{stacky} invariants) informally as 
\begin{equation}
\label{defhOm}
\hOm(x)
=\sum_{d\in \IN^{Q_0}}\hOm_0(d,y)x^d 
=\sum_{d\in \IN^{Q_0}} (-y)^{\hi_Q(d,d)}\frac{P(R(J,d))}{P(G_d)}x^d.
\end{equation}
Remark \ref{rm:exp motive} again applies.
In the presence of a cut $I\sbs Q_1$, we have a rigorous definition \cite{morrison_motivica}
\begin{equation}
\label{stacky inv with a cut}
\hOm(x)
=\sum_{d\in \IN^{Q_0}} (-y)^{\hi_Q(d,d)+2\ga_I(d)}\frac{P(R(J_I,d))}{P(G_d)}x^d,\qquad
\ga_I(d)=\sum_{(a:i\to j)\in I}d_id_j
\end{equation}
which will be useful for practical computations.

\begin{example}[Quantum dilogarithm]
\label{ex:dilog}
Let us consider a quiver $Q$ with one vertex and no loops.
Its category of representations is equivalent to the category of vector spaces.
The corresponding partition function
\begin{equation}
\hOm(x)
=\sum_{d\in\bN}\frac{(-q^\oh)^{d^2}}{P(\GL_d)}x^d
=\sum_{d\in\bN}\frac{(-q^\oh)^{-d^2}}{(q\inv)_d}x^d
=\Exp\rbr{\frac{-q^\oh x}{q-1}}
\end{equation}
is called the \idef{quantum dilogarithm}, denoted by $\bE(x)$. The plethystic exponential $\Exp$ used here is defined in \S\ref{sec:framed-unframed}.
\end{example}

For any stability function $\Z$ and any ray $\ell\sbs\bC$, we similarly define
\begin{equation}
\hOm_{\Z,\ell}(x)=\sum_{\Z(d)\in\ell}\hOm_\Z(d,y)x^d
=\sum_{\Z(d)\in\ell} (-y)^{\hi_Q(d,d)}\frac{P(R^\Z(J,d))}{P(G_d)}x^d.
\end{equation}
For any weight $\te\in\bR^{Q_0}$, we define
\begin{equation}
\hOm_{\te}(x)
=\sum_{\te\cdot d=0}\hOm_\te(d,y)x^d
=\sum_{\te\cdot d=0} (-y)^{\hi_Q(d,d)}\frac{P(R^\te(J,d))}{P(G_d)}x^d.
\end{equation}

\begin{remark}
The series $\hOm_{\Z,\ell}(x)$ and $\hOm_\te(x)$ are related as follows.
Assume that $\Z=-\te+\bi\rho$ and $\ell=\bR_{\ge0}z$ for some $z=x+\bi y\in\bC$.
Then $\Z(d)\in\ell$ if and only if $\Im(\Z(d)\ub z)=0$.
This means that $\te'(d)=0$ for $\te'=x\rho-y\te$ and we obtain 
$\hOm_{\Z,\ell}(x)=\hOm_{\te'}(x)$.
Conversely, given a weight $\te$, consider the stability function $\Z=-\te+i\de$ and the ray $\ell=\bR_{\ge0}\bi$.
Then $\Z(d)\in\ell$ if and only if $\te(d)=0$ and 
we obtain $\hOm_{\Z,\ell}(x)=\hOm_{\te}(x)$.
\end{remark}

\subsection{Relation between framed and unframed invariants}
\label{sec:framed-unframed}
In this section we will assume that the quiver $Q$ is symmetric, meaning that the number of arrows $i\to j$ is equal to the number of arrows $j\to i$ for all $i,j\in Q_0$.
The general case will be considered later in \S\ref{sec:framed-unframed2} after discussing
wall-crossing formulae.

\medskip

First, let us introduce the \idef{plethystic exponential} $\Exp:\AA^+\to 1+\AA^+$, where $\AA^+$ is the maximal ideal in 
\begin{equation}
\AA=\bQ\lser y\pser{x_1,\dots, x_n}
\end{equation}
(here $\bQ\lser y$ is the field of Laurent series).
It is a continuous map satisfying $\Exp(f+g)=\Exp(f)\Exp(g)$
and defined on monomials by 
\begin{equation}\label{Exp}
\Exp(y^kx^d)
=\sum_{m\ge0}y^{mk}x^{md}.
\end{equation}
Generally, one has
\begin{equation}
\Exp(f)= \exp\left( \sum_{m\ge1} \frac{1}{m} f(y^m,x_1^m,\dots,x_n^m) \right).
\end{equation}
Its inverse is the \idef{plethystic logarithm}
\begin{equation}\label{Log}
\Log(f)= \sum_{m\ge1} \frac{\mu(m)}{m} \log\left(
f(y^m,x_1^m,\dots,x_n^m)
\right).
\end{equation}
where $\mu$ is the M\"obius function.

\medskip

Assuming that  $Q$ is symmetric, we define invariants $\Om(d,y)$, called \idef{integer DT invariants}, by factorizing the generating function of stacky invariants \eqref{defhOm} as 
\begin{equation}
\label{hOmfromOm}
\hOm(x)=\Exp\rbr{\sum_{d\in \IN^{Q_0}\ms\set0}  \frac{\Om(d,y)\, x^d}{y\inv-y}}.
\end{equation}
Here $\Omega(d,y)$ is a priori a rational function in $y$, but in all our examples it is a Laurent polynomial in $y$ with integer coefficients. Note that the denominator
inside the plethystic exponential is sometimes chosen to be $y^2-1$ or $y-1/y$  in the literature, 
leading to different conventions for the integer DT invariants.

For any vector $w\in\bZ^n$, define the algebra homomorphisms 
\begin{equation}
S_w:\AA\to \AA,\qquad
x^d\mto (-y)^{w\cdot d}x^d. 
\end{equation}
\begin{equation}
\bar S_w:\AA\to \AA,\qquad
x^d\mto (-1)^{w\cdot d}x^d. 
\end{equation}


\begin{theorem}[See \cite{morrison_motivica,mozgovoy_wall-crossing,morrison_motivic}]
\label{th:NCDT from unframed}
Let $Q$ be a symmetric quiver. Then
\begin{equation}
Z^\rf_{\f,\NCDT}(x)
=S_{-\f}\Exp\rbr{ \sum_{d\in \IN^{Q_0}} 
\frac{(y^{2\f\cdot d}-1)\Om(d,y)\, x^d}{y\inv-y}}.
\end{equation}
In particular, numerical NCDT invariants satisfy
\begin{equation}
\label{eq:num NCDT}
Z_{\f,\NCDT}(x)
=\bar S_{\f}\Exp\rbr{
-\sum_{d\in \IN^{Q_0}} (\f\cdot d) \, \Om(d,1) \, x^d}.
\end{equation}
\end{theorem}

\begin{remark}
As we explain in \S\ref{sec:framed-unframed2}, these formulae follow by applying the semi-primitive wall-crossing formula, 
since $Z_{\f,\theta}^\rf$ is equal to $1$ for a stability parameter $\te$ with $\theta_i>0$ for $i\in Q_0$.
\end{remark}

\begin{theorem}[See \cite{morrison_motivica,mozgovoy_wall-crossing}]
\label{th:framed from unframed}
Let $Q$ be a symmetric quiver and let weight $\te$ be such that $\te\cdot d\ne0$ for all $d\in\bN^{Q_0}\ms\set0$.
Then the partition function of refined framed DT invariants is given  by
\begin{equation}
Z^\rf_{\f,\te}(x)=S_{-\f}\Exp\rbr{\sum
_{\ov{d\in\IN^{Q_0}}{\te\cdot d<0}}
\frac{(y^{2\f\cdot d}-1)\Om(d,y)\, x^d}{y\inv-y}}.
\end{equation}
In particular, numerical framed DT invariants satisfy
\begin{equation}
Z_{\f,\te}(x)
=\bar S_{\f}\Exp\rbr{-\sum
_{\ov{d\in\IN^{Q_0}}{\te\cdot d<0}}
(\f\cdot d)\Om(d,1)\, x^d}.
\end{equation}
\end{theorem}

\begin{remark}
The above formula implies that for $\f=\sum_{i\in Q_0} \f_i e_i$,
\begin{equation}
\bar S_\f Z_{\f,\te}(x)=\prod_{i\in Q_0}
(\bar S_{e_i}Z_{e_i,\te}(x))^{\f_i}.
\end{equation}
Therefore it is enough to determine the partition functions
$Z_{i,\te}(x)=Z_{e_i,\te}(x)$ for $i\in Q_0$. 
\end{remark}

\begin{example}
Let $Q$ be a quiver with one vertex and no loops.
The corresponding partition function of unframed invariants was computed in Example \ref{ex:dilog},
namely $\hOm(x)
=\Exp\rbr{\frac x{y\inv-y}}$.
Therefore
\begin{equation}
Z_{\f,\NCDT}(x)
=\bar S_\f\Exp(-\f\, x)
=\bar S_\f(1-x)^\f=(1-(-1)^\f x)^\f.
\end{equation}
This formula coincides with \eqref{ncdt vect spaces}.
\end{example}

\vspace*{1cm}

\section{Wall-crossing formulas and attractor indices \label{sec_wc}}

For a symmetric quiver $(Q,W)$, the generating functions of framed invariants defined
 in \eqref{part fun2} (or their refined counterpart in \eqref{partfun2ref}) can be deduced, 
 for any stability parameter $\te$, from universal  unframed (refined) invariants 
 $\Omega(d,y)$ obtained by factorizing the generating function of stacky invariants 
 as in \eqref{hOmfromOm}.
Physically, this can be understood as the appearance or disappearance of BPS states bound to an infinitely heavy defect (also known as framed BPS states), as the phase of the central charge of the defect is varied \cite{Gaiotto:2010be,Andriyash:2010qv}. In that case, the  unframed invariants $\Omega(d,y)$ count
BPS states with charge $d$ far away from the defect, and do not experience wall-crossing by themselves. In contrast, for non-symmetric quivers, the unframed (refined) invariants do
depend on the stability parameter $\te$, and their contributions to framed partition functions 
no longer commute. In this section we review
the wall-crossing formulae which govern this dependence, and we explain how 
framed and unframed invariants in any chamber can all be deduced from invariants in the
attractor chamber.

\subsection{Quantum affine space}
Let $Q$ be a quiver and $\hi_Q$ be its Euler form,
defined on the lattice $\Ga=\bZ^{Q_0}$ by
\begin{equation}
\label{defEuler}
\hi_Q(d,d')=\sum_{i}d_i d'_i-\sum_{a:i\to j}d_i d'_j.
\end{equation}
Consider the skew-symmetric form on $\Ga$
\begin{equation}
\label{defskEuler}
\ang{d,d'}=\hi_Q(d,d')-\hi_Q(d',d).
\end{equation}
Define the corresponding \idef{quantum affine space} (or more precisely, its coordinate ring) to be
\begin{equation}
\bA=\bop_{d\in\bN^{Q_0}}\bQ(y)x^d
\end{equation}
equipped with the multiplication
\begin{equation}
x^d\circ x^{d'}=(-y)^{\ang{d,d'}}x^{d+d'}.
\end{equation}
This algebra has a decreasing filtration according to the \qq{height} $\delta \cdot d$,
\begin{equation}
F_i\bA=\bop_{\de\cdot d\ge i}\bQ(y)x^d,\qquad \de=(1,\dots,1)\in\bZ^{Q_0},
\end{equation}
and we consider the corresponding completion
\begin{equation}
\label{defhbA}
\hat\bA=\ilim_i \bA/F_i\bA\iso\prod_{d\in\bN^{Q_0}}\bQ(y)x^d
\end{equation}
which we will also refer to as the quantum affine space.
All partition functions defined in the previous sections will be considered as elements of this algebra.

\subsection{Basic wall-crossing formulas}
In  section \S\ref{sec:refined unframed} we constructed, for any stability function $\Z:\Ga\to\bC$ and any ray $\ell\sbs\bC$, the  generating function of stacky invariants
\begin{equation}
\hOm_{\Z,\ell}(x)=\sum_{\Z(d)\in\ell}\hOm_\Z(d,y)x^d\in\hat\bA.
\end{equation}
depending on the stability function.
In particular, for any weight $\te\in\bR^{Q_0}$, we consider 
\begin{equation}
\Z=-\te+\bi \de,\qquad
\de=(1,\dots,1)\in\bZ^{Q_0},
\end{equation}
and $\ell=\bR_{\ge0}\bi$.
Then $\Z(d)\in\ell$ if and only if $\te\cdot d=0$ and we denote $\hOm_{\Z,\ell}(x)$ by $\hOm_\te(x)$.
For $\te=0$ we obtain the trivial stability condition as all objects are automatically semistable.
We denote the series $\hOm_0(x)$ by $\hOm(x)$.


\begin{theorem}[Wall-crossing formula]
For any stability function $\Z$, we have
\begin{equation}
\hOm(x)=\oprod_{\ell}\hOm_{\Z,\ell}(x),
\end{equation}
where the product runs over rays $\ell$ in the upper half-plane ordered clockwise.
In particular, the right hand side is independent of the stability function $\Z$.
\end{theorem}

We can write the above formula more explicitly as follows.
For any $\ga\in\bN^{Q_0}\ms\set0$, we have
\begin{equation}
\hOm(\ga,y)
=\sum_{\ov{\ga=\al_1+\dots+\al_n}{\al_1>_Z\dots>_Z\al_n}}
(-y)^{\sum_{i<j}\ang{\al_i,\al_j}}\prod_i\hOm_\Z(\al_i,y).
\end{equation}
where we write $\al<_Z\beta$ if $\mu_Z(\al)<\mu_Z(\beta)$.
This formula implies that we can express recursively invariants $\hOm_\Z(\ga,y)$ in terms of invariants $\hOm(\ga,y)$.
Therefore, for two stability functions $\Z,\Z'$, we can express
invariants $\hOm_{\Z'}(\ga,y)$ in terms of invariants $\hOm_\Z(\ga,y)$.
More precisely, we have

\begin{theorem}[Joyce-Reineke formula]
\label{JR}
Let $\Z$ and $\Z'$ be two stability functions.
Given a tuple $\al=(\al_1,\dots,\al_n)$ of vectors in $\bN^{Q_0}\ms\set0$ and $1\le k<n$, define $\al_{\le k}=\al_1+\dots+\al_k$, $\al_{>k}=\al_{k+1}+\dots+\al_n$ and
$$s_k(\al)=
\begin{cases}
-1& \al_k\le_\Z \al_{k+1}\text{ and } \al_{\le k}>_{\Z'}\al_{>k},\\
1&\al_k>_\Z\al_{k+1}\text{ and } \al_{\le k}\le_{\Z'}\al_{>k},\\
0&\text{otherwise}.
\end{cases}$$
Then
\begin{equation}
\hOm_{\Z'}(\ga,y)
=\sum_{\ga=\al_1+\dots+\al_n}
\prod_{k=1}^{n-1}s_k(\al)\cdot
(-y)^{\sum_{i<j}\ang{\al_i,\al_j}}\cdot
\prod_i\hOm_\Z(\al_i,y).
\end{equation}
\end{theorem}
The above formula  was proved for the trivial stability $\Z$ in \cite{reineke_harder-narasimhan} and 
in full generality  in \cite{joyce_configurations}.

\subsection{Stacky, Rational, Integer DT}
We say that a stability function $\Z:\Ga\to\bC$ is \idef{generic} if $\Arg\Z(d)=\Arg\Z(d')$ implies $d\parallel d'$ (hence $\ang{d,d'}=0$).
Similarly, a weight $\te:\Ga\to\bR$ is called \idef{generic} if
$\Ker\te$ has rank $1$ (hence $\ang{-,-}$ vanishes on $\Ker\te\sbs\Ga$).

\begin{remark}
Let $\Z=-\te+\bi\rho$ be such that $\te^\perp\cap\Ga=0$, where $\te^\perp=\sets{\ga\in\Ga_\bR}{\te(\ga)=0}$ (this means that $\te(e_i)\in\bR$ are linearly independent over $\bZ$) and $\rho(e_i)\in\bQ_{>0}$ for all $i\in Q_0$.
If $d,d'$ have equal phases, then $\te(d)/\rho(d)=\te(d')/\rho(d')$, hence $\te(md)=\te(nd')$ for $m=\rho(d')$ and $n=\rho(d)$.
This implies that $md=nd'$, hence $d\parallel d'$.
Therefore $\Z$ is generic.
On the other hand, for any indivisible vector $d\in\Ga$, consider the free abelian group $\Ga/\bZ d$ and an injective homomorphism $\te:\Ga/\bZ d\to\bR$.
It induces a homomorphism $\te:\Ga\to\bR$ such that $\te^\perp\cap\Ga=\bZ d$.
Therefore $\te:\Ga\to\bR$ is generic.
\end{remark}

Given a generic weight \te, we define invariants $\Om_\te(d,y)$ for $d\in \Ker\te$,
called \idef{integer DT invariants}, by the formula,
\cf \eqref{hOmfromOm}
\begin{equation}
\sum_{\te(d)=0}\hOm_\te(d,y)x^d
=\Exp\rbr{\frac{\sum_{\te(d)=0}\Om_\te(d,y)x^d}{y\inv-y}}.
\end{equation}
Similarly, we define invariants $\bOm_\te(d,y)$ for $d\in \Ker\te$,
called \idef{rational DT invariants}, by the formula
\begin{equation}
\sum_{\te(d)=0}\hOm_\te(d,y)x^d
=\exp\rbr{\frac{\sum_{\te(d)=0}\bOm_\te(d,y)x^d}{y\inv-y}}
\end{equation}
or more explicitly
\begin{equation}
\bOm_\te(d,y)
=(y\inv-y)\sum_{k\mid d}\frac{(-1)^{k-1}}{k}\hOm_\te(d/k,y).
\end{equation}
These families of invariants are related by the formula
\begin{equation}
\bOm_\te(d,y)
=\sum_{k\mid d}\frac 1k\frac{y\inv-y}{y^{-k}-y^k}
\Om_\te(d/k,y^k).
\end{equation}
In \cite{Manschot:2010qz}, this relation was interpreted physically as the effect of replacing Bose-Fermi statistics by Boltzmann statistics. 
Assuming that specializations at $y=1$ exist, we obtain, in the unrefined limit,
\begin{equation}
\bOm_\te(d,1)
=\sum_{k\mid d}\frac 1{k^2}
\Om_\te(d/k,1).
\end{equation}

Similarly, for a generic stability function $\Z=-\te+\bi\rho$ and any dimension vector $d$, we define \idef{integer DT invariants} and \idef{rational DT invariants} by
\begin{equation}
\Om_\Z(d,y)=\Om_{\te'}(d,y),\qquad
\bOm_\Z(d,y)=\bOm_{\te'}(d,y)
\end{equation}
where the stability parameters $\te'$ are chosen such that $\te'(d')=\Im(\Z(d')\bar\Z(d))$.
Equivalently, for any ray $\ell\sbs\bC$, we have
\begin{equation}
\sum_{\Z(d)\in\ell}\hOm_\Z(d,y)x^d
=\Exp\rbr{\frac{\sum_{\Z(d)\in\ell}\Om_\Z(d,y)x^d}{y\inv-y}}
=\exp\rbr{\frac{\sum_{\Z(d)\in\ell}\bOm_\Z(d,y)x^d}{y\inv-y}}.
\end{equation}

\subsection{Expressing framed invariants in terms of unframed (general case)}
\label{sec:framed-unframed2}
In \S\ref{sec:framed-unframed} we described how generating functions of framed invariants (refined and numerical) can be obtained from the generating functions of unframed refined invariants in the case of symmetric quivers.
A similar result, albeit more complicated, can be formulated for arbitrary quivers.

\begin{theorem}[See \cite{morrison_motivica,mozgovoy_wall-crossing,morrison_motivic}]
For any framing vector $\f\in\bN^{Q_0}$, we have
\begin{equation}
Z^\rf_{\f,\NCDT}(x)=S_\f \hOm(x)\circ S_{-\f}\hOm(x)\inv,
\end{equation}
where the inverse $(S_{-\f}\hOm)\inv=S_{-\f}(\hOm\inv)$
is taken in the quantum affine plane $\hat\bA$.
\end{theorem}

A more general result can be proved for any (generic) stability parameter $\te\in\bR^{Q_0}$.
One can uniquely decompose the series $\hOm=\hOm(x)\in\hat\bA$ in the form (\cf \eqref{decomp2})
\begin{equation}
\hOm=\hOm_{\te,+}\circ\hOm_\te\circ\hOm_{\te,-},
\end{equation}
where $\hOm_{\te,+},\hOm_\te,\hOm_{\te,-}$ are supported at $\ga$ satisfying respectively $\te(\ga)>0$, $\te(\ga)=0$, $\te(\ga)<0$.
More precisely, consider any stability function $\Z=-\te+\bi\rho$.
Then $\Arg\Z(\ga)=\frac\pi2$ if and only if $\te(\ga)=0$ and $\Arg\Z(\ga)<\frac\pi2$ if and only if $\te(\ga)<0$.
We have
\begin{equation}
\hOm_{\te,+}=\oprod_{\Arg\ell>\frac\pi2}\hOm_{\Z,\ell},\qquad
\hOm_{\te}=\hOm_{\Z,\ell},\quad \ell=\bR_{\ge0}\bi,\qquad
\hOm_{\te,-}=\oprod_{\Arg\ell<\frac\pi2}\hOm_{\Z,\ell}.
\end{equation}
The series $\hOm_{\te,-}$ counts objects with HN filtrations having slopes $<0$.
Equivalently, these are objects $M$ 
such that $\te(N)<0$ for all subobjects $N\sbs M$.

\begin{theorem}[See \cite{morrison_motivica,mozgovoy_wall-crossing}]
Let $\te\in\bR^{Q_0}$ be such that $\te\cdot \ga\ne0$ for all $\ga\in\bN^{Q_0}\ms\set0$.
Then for any framing vector $\f$ we have
\begin{equation}
\label{ZffromA}
Z^\rf_{\f,\te}(x)
=S_\f \hOm_{\te,-}(x)\circ \rbr{S_{-\f}\hOm_{\te,-}(x)}\inv.
\end{equation}
\end{theorem}

\begin{remark}
The relation \eqref{ZffromA} can be understood physically as follows: the spectrum of framed BPS states with stability $\te$ and charge $\gamma$ is the union of all bound states between the infinitely
heavy defect defining the framing, and the unframed BPS states with charge $\gamma'$ such that 
$\theta(\gamma')<\theta(\gamma)$. The formula \eqref{ZffromA} comes by applying the 
semi-primitive wall-crossing formula to all such bound states \cite{Andriyash:2010qv}. 
\end{remark}

\subsection{Relation to wall-crossing structures}
In this subsection we explain the relation to the wall-crossing structures studied in 
\cite{kontsevich_wall,gross_canonical}. 
Consider the complete quantum affine space~$\hat\bA$ defined in \eqref{defhbA}.
Let $\fg$ 
be the maximal ideal
$\fg=\cbop_{d\in\bN^{Q_0}\ms\set0}\bQ(y)x^d\sbs\hat\bA$ 
and let $G=1+\fg\sbs\hat\bA$.
Then $\fg$ is a pro-nilpotent Lie algebra and $G$ is the associated pro-unipotent algebraic group such that $\exp:\fg\to G$ is a bijection.
The Lie algebra $\fg$ is graded by $\Ga=\IZ^{Q_0}$, hence $\fg=\cbop_{\ga\in\Ga}\fg_\ga$
with $\fg_\ga=\bQ(y)x^\ga$ for $\ga\in\bN^{Q_0}\ms\set0$ and zero otherwise, such that $[\fg_\ga,\fg_{\ga'}]\sbs\fg_{\ga+\ga'}$.
Importantly, we have
\begin{equation}
\ang{\ga,\ga'}=0\imp[\fg_\ga,\fg_{\ga'}]=0
\end{equation}
The generating series $\hOm=\hOm(x)$ of stacky invariants corresponding to the trivial stability condition is an element  $\hOm\in G$. Since a \idef{wall-crossing structure} 
can be defined as  an element of the pro-unipotent group $G$ \cite{kontsevich_wall,gross_canonical}, we interpret the generating series $\cA$ as such a wall-crossing structure.

\medskip

Similarly, the generating series $\cA_\te$, for a weight $\te\in\Ga_\bR\dual$, defines an element in $G$,
and  the series $\bOm_\te$ (up to the factor $y\inv-y$) can be interpreted as an element of the Lie algebra $\bOm_\te\in\fg$ such that $\exp(\bOm_\te)=\cA_\te$.
We have seen that $\cA_\te$ can be obtained from $\cA$ by a wall-crossing formula.
An alternative way to formulate this property is as follows.
We have a decomposition
\begin{equation}
\fg=\fg^\te_+\oplus\fg^\te_0\oplus\fg^\te_-, \nn
\end{equation}
\begin{equation}
\label{decomp1}
\fg^\te_+=\bop_{\te(\ga)>0}\fg_\ga,\qquad
\fg^\te_0=\bop_{\te(\ga)=0}\fg_\ga,\qquad
\fg^\te_-=\bop_{\te(\ga)<0}\fg_\ga.
\end{equation}
These are Lie subalgebras of $\fg$, which in turn define Lie subgroups $G_\star^\te\sbs G$ for $\star=\set{+,0,-}$.
The map induced by multiplication 
\begin{equation}
\label{decomp2}
G^\te_+\xx G^\te_0\xx G^\te_-\to G
\end{equation}
is a bijection.
Therefore we have a projection $p^\te:G\to G^\te_0$.
The wall-crossing formula implies
\begin{equation}
\hOm_\te=p^\te(\hOm).
\end{equation}

\subsection{Attractor DT invariants and the attractor tree formula}
As reviewed above, the wall-crossing formulae allow to compute DT invariants $\Omega_\Z(d,y)$
for any stability function $\Z$ once they are known for a particular stability function $\Z_0$ (for all dimension vectors). 
However, no choice of $\Z_0$ appears to lead to simple answers (with the exception of the  trivial stability for quivers without relations \cite{1043.17010}).
As explained in the introduction, physics suggests that simplicity can be found by instead 
adapting the stability condition to the dimension vector $d$, namely choosing the self-stability
$\ang{-,d}$ \cite{MPSunpublished,bridgeland2016scattering,Alexandrov:2018iao}. More precisely,
we define the \idef{attractor DT invariants}
as 
\be
 \Omstar(d,y)=\Om_{\te_d}(d,y)
\ee
where $\te_d$, called an \idef{attractor stability},
is a generic perturbation of the weight $\ang{-,d}$ inside $d^\perp=\sets{\te\in\Ga_\bR\dual}{\te(d)=0}$.
We define similarly the \idef{rational attractor DT invariants} $\bOm_*(d,y)$ and invariants $\hOm_*(d,y)$.
The collection of invariants $(\hOm_*(d,y))_d$ is called the \idef{initial data} of the wall-crossing structure $\hOm$ \cite{kontsevich_wall,gross_canonical}.
It encodes the same information as the collection of attractor DT invariants $(\Om_*(d,y))_d$.
More precisely, for any indivisible dimension vector $d$, we have
\begin{equation}
\sum_{n\ge0}\hOm_*(nd,y)x^n
=\exp\rbr{
\sum_{n\ge1}\frac{\bOm_*(nd,y)x^n}{y\inv-y}}
=\Exp\rbr{
\sum_{n\ge1}\frac{\Om_*(nd,y)x^n}{y\inv-y}}.
\end{equation}

\begin{theorem}[see {\cite[Theorem 1.21]{gross_canonical}}]
Invariants 
$\Om_*(d,y),\bOm_*(d,y),\hOm_*(d,y)$ are independent of the choice of a generic perturbation of $\ang{-,d}$.
\end{theorem}
\begin{proof}
Let $d_0\in\bR_{\ge0}d\cap\Ga$ be the indivisible vector and let $\te=\ang{-,d}$.
Then the decomposition \eqref{decomp1} has the form
\begin{equation}
\fg^\te_+=\bop_{\ang{\ga,d}>0}\fg_\ga,\qquad
\fg^\te_0=\bop_{\ang{\ga,d}=0}\fg_\ga,\qquad
\fg^\te_-=\bop_{\ang{\ga,d}<0}\fg_\ga.
\end{equation}
We can further decompose
\begin{equation}
\fg^\te_0=\fg^\parallel_0\oplus\fg_0^\perp,
\end{equation}
where $\fg^\parallel_0=\bop_{\ga\in\bN d_0}\fg_{\ga}$ and $\fg^\perp_0$ involves the rest of the summands $\fg_\ga$ in $\fg^\te_0$.
One can show that $\fg^\perp_0$ is an ideal in $\fg_0^\te$.
Therefore we have a group homomorphism $\Psi:G^\te_0\to G^\parallel_0=\exp(\fg^\parallel_0)$
with kernel $G^\perp_0=\exp(\fg^\perp_0)$.
We claim that if we  apply \Psi to $p^\te(\cA)=\cA_\te\in G^\te_0$, the corresponding component $\Psi(\hOm_\te)_d$ in degree $d$ is equal to $\cA_{\te'}(d,y)$, where $\te'$ is a generic perturbation of $\te$ inside $d^\perp$. Since $\te'$ is a generic perturbation of $\te$, 
we can assume that $\te(d')>0$ implies $\te'(d')>0$ for $0\le d'\le d$.
Applying the decomposition \eqref{decomp2} we can write
$g=\cA=g^\te_+\cdot g^\te_0\cdot g^\te_-$
and then $h=g^\te_0
=h^{\te'}_+\cdot h^{\te'}_0\cdot h^{\te'}_-$.
The support of $h^{\te'}_\pm$ is contained in $\te^\perp$ (in degrees $0\le d'\le d$).
Indeed, if $0\le d'\le d$ is in the support of $h^{\te'}_+$, then $\te'(d')>0$, hence $\te(d')\ge0$.
Similarly for the support of $h^{\te'}_-$. Now by the uniqueness of the decomposition of $h$ \wrt $\te$, we conclude that $\te(d')=0$ in both cases.
The support of $h^{\te'}_\pm$ also doesn't intersect $(\te')^\perp\cap\Ga=\bZ d_0$.
Therefore $h^{\te'}_\pm\in G^\perp_0$ and we obtain
$\Psi(\hOm_\te)=\Psi(g^\te_0)=h^{\te'}_0$.
On the other hand $g=(g^\te_+h^{\te'}_+)h^{\te'}_0(h^{\te'}_-g^\te_-)$.
By our assumption on $\te'$, we have $g^\te_+\in G^{\te'}_+$ and $g^\te_-\in G^{\te'}_-$ (in degrees $0\le d'\le d$), therefore $p_{\te'}(g)=h^{\te'}_0=\Psi(\hOm_\te)$ (in degree $d$).
We conclude that 
$\cA_{\te'}(d,y)=p_{\te'}(g)_d=\Psi(\hOm_\te)_d$.
This proves that $\cA_{\te'}(d,y)$ is independent of the choice of a perturbation.
\end{proof}


\begin{lemma}
If $Q$ is an acyclic quiver, then the only self-stable objects are simple.
Moreover,
\begin{equation}
\Omstar(d,y)=\begin{cases}
1&d=e_i\text{ for some }i\in Q_0,\\
0&\text{otherwise}.
\end{cases}
\end{equation}
\end{lemma}
\begin{proof}
Choose an order on $Q_0$ such that $i\le j$ for any arrow $a:i\to j$.
Let $M$ be a self-stable representation
and $j\in Q_0$ be a maximal vertex for this ordering such that $M_j\ne 0$.
Then $M$ has a simple subrepresentation $S_j$.
Assume that $d=\udim M$ is supported not only at $j$.
Then there exists an arrow $a:i\to j$ such that $d_i\ne0$ (otherwise $M$ would be decomposable).
Therefore $\ang{e_j,d}=\sum_{a:i\to j}d_i>0$
and this contradicts to the assumption that $M$ is self-stable.
Therefore $M\iso S_j^{\oplus n}$ for some $n\ge1$.
By stability of $M$, we actually have $M\iso S_j$ hence $d=e_j$.

Now, consider a dimension vector $d$
and a stability function $\Z=-\te+\bi\de$, where $\te$ is a generic perturbation of $\ang{-,d}$ and $\de=(1,\dots,1)$.
If $M$ is a \Z-semistable representation such that $\mu_Z(M)=\mu_Z(d)$, then $M$ is supported at only one vertex $j\in Q_0$ (and therefore $M\iso S_j^{\oplus n}$ for some $n\ge1$) by the same argument as above.
This implies, for $\ell=\bR_{\ge0}\Z(d)$ (\cf Example \ref{ex:dilog})
\begin{equation}
\hOm_{\Z,\ell}(x)=\sum_{n\ge0}\frac{(-y)^{n^2} }{P(\GL_n)}x_j^n=\Exp\rbr{\frac{x_j}{y\inv-y}}.
\end{equation}
Therefore $\Omstar(d,y)=1$ for $d=e_j$ and zero otherwise.
\end{proof}

The above result can be generalized as follows.
Let us say that a quiver $Q$ is \idef{strongly connected}
if, for any two vertices $i,j\in Q_0$, there exists a path from $i$ to $j$.
For any $d\in\bZ^{Q_0}$, let $\supp (d)=\sets{i\in Q_0}{d_i\ne0}$ and let $Q|_{\supp(d)}$ be the corresponding subquiver of $Q$.

\begin{theorem}
Let $(Q,W)$ be a quiver with potential,
$d\in\bN^{Q_0}$, and $\Om_*(d,y)$ be the corresponding attractor invariant.
If $Q|_{\supp(d)}$ is not strongly connected,
then $\Om_*(d,y)=0$.
\end{theorem}
\begin{proof}
It is enough to show that $\cA_*(d,y)=0$.
Let $M$ be a representation having dimension vector $d$.
According to our assumption, there exists a decomposition $\supp(d)=A\sqcup B$ such that $A,B\ne\es$ and there are no arrows from $A$ to $B$.
Let $d'=\sum_{i\in A}d_ie_i$ and $d''=\sum_{i\in B}d_ie_i=d-d'$.
The subspace $M'=\bop_{i\in A}M_i\sbs M$ is a subrepresentation having dimension vector $d'$.

Assume that there exist arrows from $B$ to~$A$.
Then $\ang{d',d''}>0$, hence $\ang{d',d}>0$.
Given a generic perturbation $\te\in d^\perp$ of $\ang{-,d}$, we have $\te(M')=\te(d')>0$
and $M$ is not $\te$-semistable.

Assume that there are no arrows from $B$ to $A$.
Then $M''=\bop_{i\in B}M_i\sbs M$ is also a subrepresentation of $M$.
Given a generic perturbation $\te\in d^\perp$ of $\ang{-,d}$, we have $\te(d')+\te(d'')=\te(d)=0$ and $\te(d'),\te(d'')\ne0$ as otherwise $d'$ and $d''$ would be proportional which is not possible.
We conclude that either $\te(M')=\te(d')>0$ or $\te(M'')=\te(d'')>0$, hence $M$ is not $\te$-semistable.

In both cases $M$ is not semistable, hence $\cA_*(d,y)=0$.
\end{proof}

\medskip

As explained in \cite{Alexandrov:2018iao}, the DT invariants $\Om_\Z(d,y)$ can then be computed for any dimension vector $d$ using the flow tree formula. There are actually several (conjecturally equivalent) versions of this formula, three of them stated explicitly in \cite{Alexandrov:2018iao}, and 
two more which are implicit in \cite{Alexandrov:2019rth}. We shall present a formula which can 
be extracted from  \cite{Alexandrov:2019rth}, and which has the advantage of not requiring any perturbation of the dimension vectors, unlike the earlier versions in \cite{Alexandrov:2018iao}. 
This formula was proven in \cite{Mozgovoy:2021iwz}, after the first version of this article
was released.

\begin{conjecture}[Attractor tree formula]
For any generic stability function $\Z$, we have
\begin{equation}
\label{AFTform}
\bOm_{\Z}(\ga,y)
=\sum_{\ga=\al_1+\dots+\al_n}
F(\al,\Z)\cdot
\frac{(-y)^{\sum_{i<j}\ang{\al_i,\al_j}}}
{(y\inv-y)^{n-1}}
\cdot
\prod_i\bOm_*(\al_i,y).
\end{equation}
where, for a tuple $\al=(\al_1,\dots,\al_n)$ of vectors in $\bN^{Q_0}\ms\set0$,
we have $F(\al,\Z)=1$ for $n=1$ and
\begin{equation}
\label{AFTdefF}
F(\al,\Z)
=\sum_{T\in\bT_n} \left[ (-1)^{|V_T|-1}
\rbr{g_\Z(v_0)-g_*(v_0)}
\prod_{v\in V_T\ms\set{v_0}}
g_*(v)
\right],\qquad n\ge2,
\end{equation}
is a sum over all rooted plane trees with leaves labeled by $1,\dots,n$ (and all internal vertices having at least two children).
Here $V_T$ is the set of internal vertices of $T$ (including the root) and $v_0$ is the root.
The symbol $g_\Z(v)$
is a shorthand for 
$g_\Z\rbr{(\al_u)_{u\in\Ch v}}$,
where $\Ch v$ denotes the (ordered) set of children of $v\in V_T$,
$\al_u=\sum_{i<u}\al_i$ is the sum over leaves $i$ descendant of $u$ and
\begin{equation}
g_\Z(\al_1,\dots,\al_r)
=\frac{(-1)^{m}}{2^r}\cdot
\frac{(-1)^{m_0}+1}{m_0+1},
\end{equation}
\begin{equation}
m=\#{\sets{1\le k<r}{\mu_\Z(\al_{\le k})>\mu_\Z(\al_{>k}}},
\
m_0=\#{\sets{1\le k<r}{\mu_\Z(\al_{\le k})=\mu_\Z(\al_{>k})}}
\end{equation}
Similarly, $g_*(v)$ is a shorthand for 
$g_*\rbr{(\al_u)_{u\in\Ch v}}$, where $g_*(\al_1,\dots,\al_r)$ is defined as above for the self-stability $\te_*=\ang{-,\ga}$ with $\ga=\sum_{i=1}^r\al_i$.
\end{conjecture}

\begin{figure}[ht]
\begin{center}
\begin{tikzpicture}[inner sep=2mm,scale=1]
 \node (0) at (0,4) [circle,draw] {$v_0$};
\node (1) at (-2,3) [circle,draw] {$v_1$};
\node (2) at (4,3) [circle,draw] {$v_2$};
\node (5) at (-1.3,1.5) [circle,draw] {$v_3$};
\node (10) at (-2,0) [circle,draw] {$2$};
 \node (11) at (-.5,0) [circle,draw] {$3$};
 \node (4) at (-3.5,0) [circle,draw] {$1$};
\node (3) at (.9,0) [circle,draw] {$4$};
\node (6) at (2.2,0) [circle,draw] {$5$};
\node (7) at (3.5,0) [circle,draw] {$6$};
\node (8) at (4.8,0) [circle,draw] {$7$};
\node (9) at (6.1,0) [circle,draw] {$8$};
\draw[->] (0) to (1);
\draw[->] (0) to (2);
\draw[->] (0) to (3);
\draw[->] (1) to (4);
\draw[->] (1) to (5);
\draw[->] (2) to (6);
\draw[->] (2) to (7);
\draw[->] (2) to (8);
\draw[->] (2) to (9);
\draw[->] (5) to (10);
\draw[->] (5) to (11);
\end{tikzpicture}
\end{center}
\caption{Example of a rooted planar tree with $8$ leaves.
The vertices $i=1,\dots,8$ are decorated with dimension vectors $\al_i$.
The vertices $v_0,v_1,v_2,v_3$ are decorated with dimension vectors $\al_1+\dots+\al_8$, $\al_1+\al_2+\al_3$, $\al_5+\al_6+\al_7+\al_8$, $\al_2+\al_3$ respectively.}
\label{Schtree}
\end{figure}
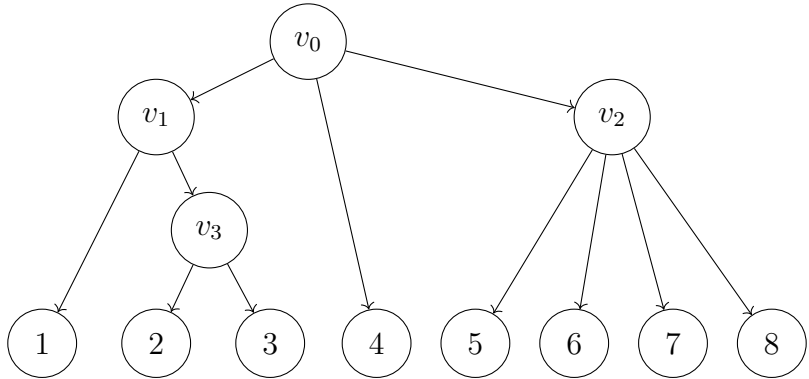

\begin{remark}
The sum in \eqref{AFTform} runs over all \idef{ordered} decompositions $\gamma=\sum _i \alpha_i$
with $\alpha_i\in \IN^{Q_0}\backslash \{0\}$. It may be written equivalently as a sum over  \idef{unordered} decompositions, by inserting a sum over permutations in \eqref{AFTdefF}
and inserting a  \qq{Boltzmann statistics} factor $1/\n{\Aut((\alpha_i)_i)}$ as in  
\cite{Manschot:2010qz,Alexandrov:2018iao}. At the attractor point $\te_\ga$ (an infinitesimal perturbation of $\ang{-,\ga}$), the formula \eqref{AFTform} collapses to 
$\bOm_{\te_\ga}(\ga,y)=\bOm_*(\gamma,y)$ as expected. 
Computer experiments indicate
that for non-generic stability parameters, in particular trivial stability, the formula \eqref{AFTform} produces a Laurent polynomial with rational coefficients, whose interpretation is currently 
unclear~\cite{Alexandrov:2020dyy}. 
\end{remark}

\begin{example}
Consider the generalized Kronecker quiver with $m>0$ arrows $1\to 2$.
Let $e_1=(1,0)$, $e_2=(0,1)$
and $\gamma=(1,1)$.
Note that $\ang{e_1,e_2}=-m$.
Given a stability parameter $\te\in\bR^2$, we consider $\Z=-\te+\bi\de$, $\de=(1,1)$, and the corresponding slope function $\mu_\te(\ga)=\te(\ga)/\de(\ga)$.
The self-stability $\ang{-,\ga}$ corresponds to $\te_*=(-m,m)$, equivalent to the stability parameter $(0,1)$.
There are two decompositions of $\gamma$ with $n\geq 2$ elements, namely
 $(\alpha_1,\alpha_2)=(e_1,e_2)$ and $(\alpha_1,\alpha_2)=(e_2,e_1)$,
and a single rooted plane tree in both cases.
Since $\Omstar(\ga)=0$ and $\Omstar(e_i)=1$, 
the formula gives
\be
\Om_{\Z}(\ga,y) = \frac12 \frac{(-y)^{-m}-(-y)^{m}}{y^{-1}-y}
\left[ -\sgn\left(\mu_\te(e_1)- \mu_\te(e_2) \right)
+\sgn\left(\mu_*(e_1)- \mu_*(e_2) \right) \right].
\ee
For the stability parameter $\te=(1,0)$, where
the slopes $\mu_\te(e_i)$ have opposite order to the attractor slopes $\mu_*(e_i)$, we obtain
\begin{equation}
\Om_{\Z}(\ga,y) 
=\frac{(-y)^{m}-(-y)^{-m}}{y^{-1}-y}
=(-y)^{1-m}\cdot P(\bP^{m-1};y)
\end{equation}
which is the (virtual) motive of $\bP^{m-1}$, the moduli space $M^\te(Q,\ga)$ of \te-semistable representations.
For the stability parameter $\te=(0,1)$, where
the slopes $\mu_\te(e_i)$ have the same order to the 
attractor slopes $\mu_*(e_i)$, we obtain
$\Om_{\Z}(\ga,y)=0$, corresponding to the fact 
that the moduli space $M^\te(Q,\ga)$ is empty.
For dimension vector $\gamma>(1,1)$, the number of contributing trees grows rapidly, but
there are many cancellations. We have checked that the result agrees with the Reineke formula
up to $\gamma=(4,4)$.
\end{example} 

\begin{remark}
\label{remarkOmS}
For a compact CY 3-fold, the attractor index $\Omstar(\gamma,y)$ is expected to include contributions both from single-centered BPS black holes with charge $\gamma$, as well as 
from  \qq{scaling configurations} of BPS black holes with charge $\alpha_i$. These
scaling configurations correspond to boundaries of the phase space of multicentered
black hole solutions where the centers come arbitrarily close \cite{Bena:2012hf} and become
indistinguishable from single-centered black holes. In contrast, the r.h.s.\ of \eqref{AFTform} incorporates compact components of the phase space which are free from such boundaries.
The Coulomb branch formula developed in \cite{Manschot:2011xc,Manschot:2013sya} gives a prescription for disentangling these
scaling contributions from genuine single-centered black holes, counted by the `single-centered' invariant $\OmS(\gamma,y)$ (also known as pure Higgs  \cite{Bena:2012hf}  or quiver invariant \cite{Lee:2012naa}). For symmetric quivers, the invariants $\OmS(\gamma,y)$ coincide with the
DT invariants $\Omega_{\Z}(\gamma,y)$, which are independent of the stability condition $\Z$.
We refer to \cite{Manschot:2014fua} for a concise review of the
Coulomb branch formula, which plays a marginal role in the present work.
\end{remark}

\begin{remark}
\label{mathematica}
The attractor tree formula and its variants, the Joyce-Reineke formula and the Coulomb branch formula, as well as several tools for dealing with brane tilings, are implemented in the Mathematica package {\tt CoulombHiggs.m} maintained by the second-named author and available from 
\url{www.lpthe.jussieu.fr/~pioline/computing.html}.
See also the SageMath package 
\url{https://github.com/smzg/msinvar}
by the first-named author.
\end{remark}

\subsection{Mutations}
Upon varying the central charge $\Z$, it may happen that $\Z(e_k)$ for $k\in Q_0$ no longer lies in the upper half-plane $\cH$ (see \S\ref{sec_stab}). Across the locus where $\Z(e_k)$ becomes real,
which is known as a wall of the second kind \cite{ks}, the heart of the $t$-structure $\cA$ 
turns into the category of representations $\cA'=\Rep Q'$ of a different quiver $(Q',W')$
related to $(Q,W)$ by a mutation (or, in physics parlance, by Seiberg duality) 
\cite{Berenstein:2002fi,zbMATH05848698,Aganagic:2010qr,Alim:2011ae}. The quiver $Q'$
has the same vertex set $Q'_0=Q_0$, but the set of arrows $Q'_1$ is obtained from $Q_1$
by adding arrows $a'_{ij}:i\to j$ for each pair of arrows $a:i\to k$ and $b:k\to j$ in $Q_1$, 
and then reversing
the orientation of all arrows starting or ending at the vertex $k$ \cite{2001math4151F}. 
The potential $W'$ is obtained from $W$ by adding cubic terms $a'_{ij} a'_{ki} a'_{jk}$
and eliminating two-cycles in $W'$ which have a quadratic term  \cite{zbMATH05573998}
(see \cite[\S2.5]{Alim:2011ae} for a more detailed prescription). 
Depending on whether the central charge $\Z(e_k)$ exits $\cH$
through the positive $(\varepsilon=1)$ or negative $(\varepsilon=-1)$ real axis, 
the dimension vectors associated to the vertices of $Q'$  are given by
\be
\label{mutDSZgen0}
e'_i=\begin{cases}
-e_k & \text{if }i=k\\
e_i + \max(0,\varepsilon\, \gamma_{ik}) \, e_k&  \text{if }i\neq k
\end{cases}
\ee
where $\gamma_{ij}=\ang{e_i,e_j}$ is the skew-symmetric Euler form. 
Under this transformation, the skew-symmetric Euler form 
$\gamma'_{ij}=\ang{e'_i,e'_j}$ becomes
\be
\gamma'_{ij}=
\begin{cases}
-\gamma_{ij} &\text{if }i=k\text{ or }j=k \\
 \ \gamma_{ij} + \max(0, 
 \gamma_{ik} \gamma_{kj})\,\sgn(\gamma_{kj}) & 
 \text{if }i,j\neq k
 \end{cases}
\ee
in line with the above prescription for constructing the arrows of $Q'$.
The dimension vector $d=\sum_{i\in Q_0} d_i e_i$ can be expressed as a linear combination 
$\sum_{i\in Q'_0} d'_i e'_i$
of the basis vectors $e'_i$ associated to the vertices of $Q_0'$ via
\be
d'_i=\begin{cases}
-d_k+ \sum_{j\neq k} d_j \, 
\max(0,\varepsilon \gamma_{jk}) &\text{if }i=k\\
d_i &\text{if }i\neq k
\end{cases}
\ee
Let $\te\in \IR^{Q_0}$ be the stability weight on $Q$ corresponding to the stability condition 
$\Z$ before $\Z(e_k)$ has crossed into the lower half-plane. The condition that $\Z(e_k)$ exits $\cH$
through the positive $(\varepsilon=1)$ or negative $(\varepsilon=-1)$ real axis
implies that $\sgn \te_k=\varepsilon$. We
define the stability weight $\te'$ on $Q'$ by 
\be
\te'_i= \begin{cases} 
-\te_k & \hbox{if $i=k$} \\
\te_i+ \text{max}(0, \varepsilon \gamma_{ik}) \, \te_k & \hbox{if  $i\ne k$}
\end{cases}
\ee   
such that $\sum_{i\in Q_0} d_i \te_i =\sum_{i\in Q'_0} d'_i \te'_i$. Note that the self-stability
condition $\te=\ang{-,d}$ is invariant under this assignment.

\begin{conjecture}[\cite{Manschot:2013dua}]
Let $\Om_\te(d,y)$ and $\Om'_{\te'} (d',y)$ be the DT invariants  associated to the quivers $(Q,W)$ and $(Q',W')$, with dimension vectors and stability weights related as above. Suppose that 
$d$ is not collinear to $e_k$, or equivalently that $d'$ is not collinear to $e'_k$. Then
\begin{enumerate}
\item If $d\in \IN^{Q_0}$ and $d'\in \IN^{Q'_0}$, then $\Om_\te(d,y)=\Om'_{\te'} (d',y)$ 
\item If $d\in \IN^{Q_0}$ but $d'\notin \IN^{Q'_0}$, then $\Om_\te(d,y)=0$
\item if $d'\in \IN^{Q'_0}$ but $d\notin \IN^{Q_0}$, then $\Om'_{\te'}(d',y)=0$
\end{enumerate}
In particular, the attractor invariants coincide or vanish under the same conditions.
\end{conjecture}

\begin{remark}
This conjecture was put forward in \cite{Manschot:2013dua} and checked for several examples
in  \cite{Manschot:2013dua,Kim:2015fba}.
Attractor invariants were not discussed in \cite{Manschot:2013dua}, but a similar property was
stated for single-centered invariants $\OmS(d,y)$, allowing for the possibility that 
$\OmS(e_k,y)\neq 1$ for simple representations. The formulae above follow from Eqs.~(1.7)-(1.13) in \cite{Manschot:2013dua} upon setting $M=1$, and flipping the sign of $\gamma_{ij}$ due to our
choice of conventions.
\end{remark}
   
\vspace*{1cm}

\section{Brane tilings and crystals \label{sec:tilingcrystal}}

For a toric CY3-fold $\cX$, the derived category of coherent sheaves $D^b(\coh \cX)$ is isomorphic
to the derived category of representations of a quiver $(Q,W)$. In physical parlance, the dynamics of D-branes wrapped on $\cX$ is described by a gauge theory with fields encoded in the quiver $Q$ and superpotential $W$. For a certain class of toric CY3-folds, $(Q,W)$ can be read off from a brane tiling,
and framed BPS states in the non-commutative chamber can be represented as molten configurations
of a crystal constructed from the brane tiling. In this section we review aspects of these well-known
facts. 

\subsection{Basics of toric geometry}
\label{sec:toric}
A $d$-dimensional toric variety $\cX$ is a complex manifold with an action of $(\IC^\times)^d$ having a dense open orbit. It may be decomposed into a set of complex tori $(\IC^\times)^{d-k}$ (the orbits)
with $0\leq k\leq d$, associated to $k$-dimensional cones in a fan $\Sigma$ in $\IZ^d$.
Let $(v_1,\dots, v_{r})$ be the set of indivisible vectors
in $\IZ^d$ generating one-dimensional cones (corresponding to toric divisors in $\cX$). For each of the $r-d$ linear relations $q_1 v_1+\dots +q_{r} v_r=0$
among the $v_i$'s, consider the $\IC^\times$ action on $\IC^{r}$
\begin{equation}
\la\cdot (z_1,\dots z_r)= 
( \lambda^{q_1} z_1,\dots, \lambda^{q_r} z_r)
\end{equation}

\begin{remark}
Consider an exact sequence (we assume that $v$ is surjective)
\begin{equation}
0\to L\to\bZ^{r}\xto v N=\bZ^d\to0,\qquad
v(e_i)=v_i.
\end{equation}
where $L$ is the kernel of the matrix with columns $v_i$.
It induces a homomorphism of tori 
$L\ts_\bZ\bC^\xx\iso(\bC^\xx)^{r-d}\emb (\bC^\xx)^{r}$.
Therefore we have an action of $(\bC^\xx)^{r-d}$ on $\bC^{r}$.
\end{remark}

The variety $\cX$
is then isomorphic to $(\IC^{r}\backslash F_\Sigma)/(\IC^\times)^{r-d}$, where $F_\Sigma$ is 
the intersection of unions $\bigcup_{i\in I}\set{z_i=0}$ for all subsets $I\sbs\set{1,\dots,r}$ such that $\sets{v_i}{i\notin I}$ spans a cone in \Si. More explicitly,  $F_\Si$ is the zero locus
\begin{equation}
F_\Si=Z(z^{\hat\si}\col \si\in\Si),\qquad
z^{\hat\si}=\prod_{i\col v_i\notin\si}z_i.
\end{equation}
The toric divisors $D_i=\{z_i=0\}$ in $\cX$
satisfy the linear equivalence relations $\sum_i a_iD_i=0$ whenever $
a\cdot q:=\sum_ia_iq_i=0$ for $q\in L$.
The canonical line bundle is equal to $K_\cX=\cO(-\sum_iD_i)$.
If $\cX$ is  Calabi-Yau, then the canonical divisor $-\sum_iD_i$ is trivial, hence $\sum_i q_i=0$ for all $q\in L$.
This means that the vector $\rho=(1,\dots,1)\in\bZ^{r}$ satisfies $\rho\cdot q=0$ for $q\in L$, hence induces a linear map $\rho:N\to \bZ$ such that $\rho(v_i)=1$ for all $i$.
Hence all vectors $v_i$ lie in the same hyperplane.
We project this hyperplane to $\IZ^{d-1}$ (for example, taking the first $d-1$ coordinates)
and consider the corresponding vectors $\hat v_i\in\bZ^{d-1}$.
The toric diagram (defined up to $\GL(d-1,\IZ)$-transformations) is the convex hull $P$ of the vectors $\sets{\hat v_i}{i=1\dots r}$, triangulated by the restriction of the fan to $\IZ^{d-1}$. Internal lattice points correspond to compact divisors in $\cX$, while lattice points on the boundary correspond to non-compact divisors (subject to linear equivalences).
The virtual Poincar\'e polynomial of $\cX$ is obtained by summing up the Poincar\'e polynomials of the toric orbits,
\be
P(\cX) = \sum_{k=0}^d n_{k} (q-1)^{d-k} 
\ee
where $n_k$ is the number of $k$-dimensional cones in $\Sigma$ (hence $n_0=1, n_1=r$).

\begin{lemma}
\label{lm:Ptoric}
Let \cX be a smooth toric CY3-fold given by a convex toric diagram having $i$ internal vertices and $b$ boundary vertices. Then
\begin{equation}
\label{Ptoric}
P(\cX)=q\,(q^2+(i+b-3)q+i).
\end{equation}
In particular, the virtual Poincar\'e polynomial is independent of the triangulation. 
\end{lemma}
\begin{proof}
By Pick's theorem, the number of triangles in the diagram (all having area $\oh$) is equal to $2i+b-2$.
By Euler's formula, the number of edges in the diagram is equal to $(i+b)+(2i+b-2)-1=3i+2b-3$.
Therefore
\begin{equation}
P(\cX)=(q-1)^3+(i+b)(q-1)^2
+(3i+2b-3)(q-1)+(2i+b-2)
\end{equation}
which simplifies to \eqref{Ptoric}.
\end{proof}

\begin{example}
The toric fan for the conifold is spanned by the vectors 
\be
v_0 = \begin{pmatrix} 0\\0\\1 \end{pmatrix}, \quad 
v_1 = \begin{pmatrix} 1\\0\\1 \end{pmatrix},\quad 
v_2 = \begin{pmatrix} 1\\1\\1 \end{pmatrix},\quad 
v_3 = \begin{pmatrix} 0\\1\\1 \end{pmatrix},
\ee
Its crepant resolution $\cX=\cO(-1)\oplus \cO(-1)\rightarrow \IP^1$ corresponds to a triangulation of the toric diagram.
Different triangulations are related by a flop transition.
The toric diagram encloses the vectors (see Figure \ref{extoric}, left)
\be
\hat v_0 = \begin{pmatrix} 0\\0 \end{pmatrix}, \quad 
\hat v_1 = \begin{pmatrix} 1\\0 \end{pmatrix}, \quad 
\hat v_2 = \begin{pmatrix} 1\\1 \end{pmatrix}, \quad 
\hat v_3 = \begin{pmatrix} 0\\1 \end{pmatrix}
\ee
There is no internal lattice point so the resolution of the conifold does not support any compact divisor.
The four corners correspond to linearly equivalent non-compact divisors.
The diagonal edge corresponds to the basis $\IP^1$ of the fibration.  
Since $i=0$, $b=4$,
the virtual Poincar\'e polynomial is $P(\cX)=q^2(q+1)$, as expected for 
a rank two vector bundle over $\bP^1$.
\end{example}

\begin{example}
\medskip The toric fan for the orbifold $\IC^3/\IZ_3$ is spanned by the vectors 
\be
v_0 = \begin{pmatrix} 0\\0\\1 \end{pmatrix}, \quad 
v_1 = \begin{pmatrix} 1\\0\\1 \end{pmatrix},\quad 
v_2 = \begin{pmatrix} 0\\1\\1 \end{pmatrix},\quad 
v_3 = \begin{pmatrix} -1\\-1\\1 \end{pmatrix}.
\ee
The toric diagram encloses the vectors (see Figure \ref{extoric}, right)
\be
\hat v_0 = \begin{pmatrix} 0\\0 \end{pmatrix}, \quad 
\hat v_1 = \begin{pmatrix} 1\\0 \end{pmatrix}, \quad 
\hat v_2 = \begin{pmatrix} 0\\1 \end{pmatrix}, \quad 
\hat v_3 = \begin{pmatrix} -1\\-1 \end{pmatrix}
\ee
The non-zero vectors $(\hat v_1,\hat v_2, \hat v_3)$ span the toric fan for the surface $S=\IP^2$.
The crepant resolution $\cX$ of the orbifold is isomorphic to the canonical bundle $K_{\bP^2}$.
It is given by the triangulation of the toric diagram.
The divisor $D_0$ associated to the internal point $v_0$ is the exceptional divisor 
in the crepant resolution, while $D_1=D_2=D_3$ is the non-compact divisor obtained by restricting $K_{\IP^2}$ to 
a line in $\IP^2$.
Since $i=1$, $b=3$,
the virtual Poincar\'e polynomial is
$P(\cX)=q(q^2+q+1)$, as expected for a line bundle over $\bP^2$.
\end{example}

\begin{figure}[ht]
\begin{center}
\begin{tikzpicture}[inner sep=2mm,scale=1]
\begin{scope}[shift={(-3,-.4)}]  
\draw[step=1cm,black!20,thin] (-.4,-.4) grid (1.4,1.4); 
\draw (0,0) -- (1,1);
\draw (0,0) -- (0,1) -- (1,1) -- (1,0) -- (0,0) ;
\filldraw [black]  (0,0) ellipse (0.1 and 0.1);
\filldraw [black]  (0,1) ellipse (0.1 and 0.1);
\filldraw [black]  (1,1) ellipse (0.1 and 0.1);
\filldraw [black]  (1,0) ellipse (0.1 and 0.1);
\end{scope}
\begin{scope}[shift={(3,0)}]  
\draw[step=1cm,black!20,thin] (-1.4,-1.4) grid (1.4,1.4); 
\draw (0,0) -- (1,0);
\draw (0,0) -- (0,1);
\draw (0,0) -- (-1,-1);
\draw (1,0) -- (0,1) -- (-1,-1) -- (1,0) ;
\filldraw [black]  (1,0) ellipse (0.1 and 0.1);
\filldraw [black]  (0,1) ellipse (0.1 and 0.1);
\filldraw [black]  (-1,-1) ellipse (0.1 and 0.1);
\filldraw [black]  (0,0) ellipse (0.1 and 0.1);
\end{scope}
\end{tikzpicture}
\end{center}
\caption{Toric diagrams for the resolved conifold and $\IC^3/\IZ_3$ (a.k.a. local $\IP^2$)
\label{extoric}}
\end{figure}


\subsection{Quivers and brane tilings}
\label{sec:tilings}
The problem of determining the dynamics of D-branes on a local CY3-fold $\cX$  has a long history, going back to the seminal works \cite{Douglas:1996sw,Douglas:1997de}. In general, the dynamics
at low energies is governed by a supersymmetric gauge theory with 8 supercharges, living in the non-compact directions of the D-brane worldvolume. For BPS states localized in the three spatial dimensions, this is  a supersymmetric quiver quantum mechanics of the type considered in \cite{Denef:2002ru},
with product gauge group $G=\prod_{i=1}^r U(N_i)$, chiral fields
$\Phi_{ij}^\alpha$ in bifundamental representations $(N_i,\overline{N_j})$, a gauge-invariant
holomorphic superpotential $W(\{\Phi_{ij}^\alpha\})$, and real Fayet-Iliopoulos parameters $\theta_i$ for each $U(N_i)$ factor in $G$. The ranks $N_i$ are the coefficients of the D-brane charge 
$\gamma=\sum_{i=1}^r N_i \gamma_i$ on a basis of charges $\gamma_i\in H^*(\cX)$ 
associated to a set of \qq{elementary D-branes}, 
and the net number of chiral fields $| \{\Phi_{ij}^\alpha \}|
-|\{\Phi_{ji}^\alpha\}|$ going from $i$ to $j$ is given by (minus) the skew-symmetrized Euler form 
$-\langle \gamma_i,\gamma_j\rangle$. 
 The full BPS spectrum, for given stability parameters 
$\theta_i$, is then obtained as supersymmetric bound states of these elementary constituents, represented by BPS ground states of the quiver quantum mechanics.
In the presence of an infinitely heavy defect of charge $\gamma_\f$, such as a D6-brane wrapping $\cX$ or D4-branes wrapping non-compact divisors in $\cX$, the quiver quantum mechanics obtains an additional gauge group $U(N_\infty)$ and arrows $\Phi_{\infty,i}^{\alpha}, \Phi_{i,\infty}^{\alpha}$, and computes the
number of framed BPS states.

\medskip

Mathematically, BPS grounds states are cohomology classes on the moduli space of \te-semistable representations  of the quiver with potential $(Q,W)$. The \qq{elementary D-branes},
or \qq{fractional branes} in the context of orbifolds, correspond to a \idef{tilting sequence} 
$T=\bop_{i=1}^rT_i$ in the derived category of coherent sheaves $D^b(\coh \cX)$, such that $T_i$  generate
$D^b(\coh \cX)$ and $\Ext^k(T,T)=0$ for $k\ne0$. 
When $\cX$ is the total space of the canonical bundle on a complex surface $S$, 
a tilting sequence $T$ can be constructed by lifting a strong exceptional collection of line bundles
on $S$ \cite{Herzog:2003zc,Aspinwall:2004vm}. Note however that the lifted sequence need not be exceptional, in particular $\End(T_i)=\Gamma(\cX,\cO_\cX)$ may have dimension $>1$. 
The triangulated category $D^b(\coh \cX)$ is then equivalent to the category of representations of the Jacobian algebra $J(Q,W)$ for a quiver with potential $(Q,W)$ associated to $T$ \cite{Aspinwall:2004bs, Aspinwall:2005ur}.

\medskip

For a wide class of toric CY threefolds, the construction of the tilting sequence 
$T$ can be by-passed and the  
quiver $(Q,W)$ can be read off from a  \idef{brane tiling}  \cite{Hanany:2005ve,Franco:2005rj}. The latter is a bipartite graph $G$ embedded in a $2$-dimensional (real) torus $\cT$, or equivalently
a periodic bipartite graph $\tilde G$ on $\IR^2$. Each vertex carries a color, black or white,
such that edges connect only vertices with different colors.  The quiver $Q$ is then the dual graph of $G$: the vertices $i\in Q_0$ correspond to faces of $G$ (\ie the connected components of $\cT\ms G$) and the arrows $a:i\to j \in Q_1$ to edges common to 
faces $i$ and $j$. The arrows are   oriented so that they go clockwise around white vertices of $G$ and go anti-clockwise around black vertices of~$G$.

\begin{figure}[ht]
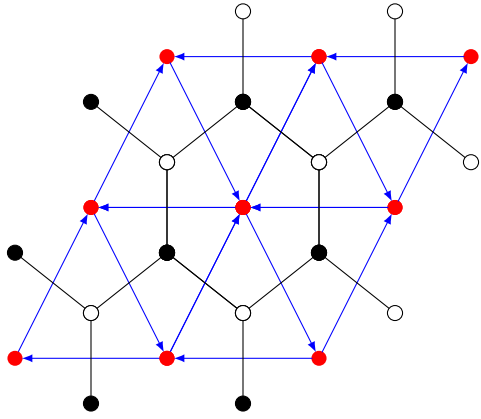

\begin{ctikz}
\path(0,0)pic{triang};
\path(2,0)pic{triang};
\path(1,2)pic{triang};
\path(3,2)pic{triang};
\end{ctikz}
\caption{A bipartite graph (in black and white) and the dual quiver (in red and blue)}
\end{figure}

Let $Q_2$ be the set of connected components of $\cT\ms Q$, or equivalently the set of vertices of $G$.
Let $Q_2^+$ and $Q_2^-$ correspond to the sets of white and black vertices of $G$.
For any face $F\in Q_2$, let $w_F$ be the cycle obtained by going along the arrows of $F$ (defined up to a cyclic shift). The potential  $W$ is then 
\begin{equation}
W=\sum_{F\in Q_2^+}w_F-\sum_{F\in Q_2^-}w_F.
\end{equation}
Note that $|Q_0|-|Q_1|+|Q_2|=0$, where $|Q_i|$ denotes the cardinality of the set $Q_i$, 
since the Euler number of the two-dimensional  torus vanishes.

\medskip

Conversely, starting with a quiver $(Q,W)$ with $|Q_0|-|Q_1|+|Q_2|=0$, it is straightforward to reconstruct the brane tiling: for every term $w$ in $W$, we construct a polygon with edges corresponding to the arrows of $w$ and orientation depending on the sign of $w$.
Then we glue these polygons along equal edges and obtain a torus with an embedded quiver.
Considering the dual graph, we obtain the required bipartite graph embedded in a torus. 
However, not all quivers corresponding to toric CY3 singularities satisfy the
torus condition  $|Q_0|-|Q_1|+|Q_2|=0$, see e.g. \cite{Aspinwall:2010mw} for some
counter-examples.

\begin{example}
\label{ex:bt of quotient}
Let $\Ga\sbs\SL_3(\bC)$ be an abelian group.
The corresponding $3$-dimensional representation of $\Ga$ can be decomposed as $\bC^3=\rho_1\oplus\rho_2\oplus\rho_3$, where $\rho_i$ are $1$-dimensional representations of \Ga.
The set of \Ga-representations can be identified with the group of characters $\hat\Ga=\Hom(\Ga,\bC^\xx)$ (non-canonically isomorphic to $\Ga$).
Then $\rho_i\in\Ga$ satisfy $\rho_1\rho_2\rho_3=1$.
The corresponding Mckay quiver has the set of vertices $Q_0=\Ga$ and arrows $a_{i,\rho}:\rho\rho_i\to\rho$ for all $\rho\in\hat\Ga$ and $i=1,2,3$.
For any $\rho\in\hat\Ga$ and a permutation $\pi\in S_3$, consider the cycle
\begin{equation}
\rho =\rho \rho_{\pi(1)}\rho_{\pi(2)} \rho_{\pi(3)} \to
\rho\rho_{\pi(1)}\rho_{\pi(2)}\to\rho\rho_{\pi(1)}\to\rho
\end{equation}
The potential $W$ is the sum over all such cycles (up to a cyclic shift), weighted with
 $\sgn(\pi)$.
This implies that $\n{Q_0}-\n{Q_1}+\n{Q_2}=\n{\Ga}-3\n{\Ga}+2\n{\Ga}=0$, hence by applying the gluing algorithm described earlier we obtain a torus and a brane tiling.
\end{example}

Conversely, starting from the  brane tiling, one may recover both the toric diagram of $\cX$ and
the tilting sequence \cite{Franco:2005rj,Hanany:2006nm,Mozgovoy:2009fi,Bender:2009fh}. In particular,
periodic perfect matchings are in correspondence with integer points on the toric diagram $P$. The 
correspondence is many-to-one except for corner points of $P$, which correspond to a single perfect
matching \cite{Franco:2005rj}. More directly, one reconstruct the full toric CY3-fold $\cX$ from the moduli space of D0-branes, as we discuss next.

\subsection{\tpdf{$\Lambda$}{L}-grading and torus action} 

In order to analyze Jacobian algebras associated to brane tilings, 
consider the chain complex $\cC$
of abelian groups \cite{Mozgovoy:2008fd}
\begin{equation}
\dots\to 0\to\bZ Q_2\xto{d_2}\bZ Q_1\xto{d_1}\bZ Q_0\to 0\to\dots,
\end{equation}
where $d_2(F)=\sum_{a\in F}a$ for $F\in Q_2$ and $d_1(a)=t(a)-s(a)$ for $a\in Q_1$.
Its homology $H_*(\cC)$ is isomorphic to $H_*(\cT,\IZ)$.
The quotient 
\begin{equation}
\La=\bZ Q_1/(d_2(F)-d_2(F')\col F,F'\in Q_2)
\end{equation}
is a free abelian group of rank $|Q_0|+2$.
The projection 
\begin{equation}\label{weight}
\wt:\bZ Q_1\to\La,
\end{equation}
which we will call the \idef{weight function},
induces a $\La$-grading on the path algebra~$\bC Q$.
For any arrow $a\in Q_1$, we have $\frac{\dd W}{\dd a}=\frac{\dd w_F}{\dd a}-\frac{\dd w_{F'}}{\dd a}$, where $F,F'$ are two faces that contain~$a$.
Note that these derivatives have the same weight in $\La$.
Therefore the ideal generated by $\frac{\dd W}{\dd a}$, for $a\in Q_1$, is homogeneous \wrt the \La-grading and the Jacobian algebra $J=J(Q,W)$ inherits the \La-grading.
This implies that the space of representations $R(J,d)$ is equipped with an action of the complex torus $T_\La=\Hom(\La,\bC^\xx)$
\begin{equation}
t\cdot (M_a)_{a\in Q_1}=(t^a M_a)_{a\in Q_1},\qquad
t^a=t(\wt(a)).
\end{equation}

Starting from a brane tiling and the corresponding Jacobian algebra $J=J(Q,W)$,
one may obtain the singular toric CY3 $\cX$ as the moduli space $M(J,\de)=R(J,\de)\GIT G_\de$,
for the dimension vector $\de=(1,\dots,1)\in\bZ^{Q_0}$ 
(physically, this is the moduli space for a single D0-brane). 
This is an affine variety with the coordinate ring $\bC[M(J,\de)]=\bC[R(J,\de)]^{G_\de}$ isomorphic to the center of $J$.
Note that $G_\de=(\bC^\xx)^{Q_0}$ and the moduli space $M(J,\de)$ is equipped with an action of the torus $T_\La/G_\de$.
This torus is isomorphic to $T_M=\Hom(M,\bC^\xx)$,
where $M=\Ker(d_1:\La\to\bZ Q_1)$ is a rank $3$ lattice.

\medskip

To see that $M(J,\de)$ is a 3-dimensional toric variety
with an action of the torus $T_M$, we need to describe an open toric orbit in $M(J,\de)$.
Consider the open subset $T'=R(J,\de)\cap(\bC^\xx)^{Q_1}\sbs R(J,\de)$ (which is a torus).
Its points can be identified with maps $x:Q_1\to\bC^\xx$ such that $\prod_{a\in F}x(a)$ is independent of $F\in Q_2$.
We obtain a group homomorphism $x:\bZ Q_1\to\bC^\xx$ such that $x(d_2(F))=x(d_2(F'))$ for all $F,F'\in Q_2$.
Hence we can identify $x$ with a homomorphism $x:\La\to\bC^\xx$ and therefore we obtain $T'=T_\La$.
Taking the quotient by $G_\de=(\bC^\xx)^{Q_0}$ we obtain the 3-dimensional torus
$T_\La/G_\de\iso T_M$ inside $M(J,\de)$.

\medskip

Note that $M(J,\de)$ is generally a singular $3$-dimensional toric CY variety. By considering moduli spaces 
\begin{equation}
M^\te(J,\de)=R^\te(J,\de)/G_\de
\end{equation}
for generic stability parameters $\te$, one obtains crepant resolutions of the variety $M(J,\de)$,
which are still toric since the torus $T_M$ is also contained in these moduli spaces.
One can construct the tilting bundle $T=\bop_{i\in Q_0}T_i$ on $X=M^\te(J,\de)$ as follows
(see \eg \cite{Bender:2009fh}).
The diagonal $\bC^\xx\sbs G_\de=(\bC^\xx)^{Q_0}$ acts trivially on $R^\te=R^\te(J,\de)$
and we have a free action of $G_\de/\bC^\xx$ on $R^\te$.
Let us fix a vertex $i_0\in Q_0$.
For any vertex $i\in Q_0$, we equip the trivial line bundle $L_i=R^\te\xx\bC$ with an action of $G_\de$ given by
\begin{equation}
t\cdot (M,c)=(t\cdot M,t_it_{i_0}\inv c),
\end{equation}
where $t\cdot M$ denotes the standard action of $G_\de$ on $R^\te(J,\de)$.
Then the diagonal $\bC^\xx\sbs G_\de$ acts trivially on $L_i$ and we have a free action of $G_\de/\bC^\xx$ on $L_i$.
Taking the quotients, we obtain a line bundle
\begin{equation}
T_i=L_i/G_\de\to R^\te/G_\de=M^\te(J,\de)
\end{equation}
over $M^\te(J,\de)$.
Then $T=\bop_{i\in Q_0}T_i$ is the required tilting bundle (it is the universal family over $M^\te(J,\de)$ \cite{king_moduli})
and induces an equivalence of derived categories \cite{vandenbergh_non}
\begin{equation}
D^b_c(\coh X)\iso 
D^b(\mmod J),\qquad F\mto\RHom(T,F).
\end{equation}
For this reason the Jacobian algebra $J$ is
called a non-commutative crepant resolution of $M(J,\de)$.

\sec[Crystals]
In the previous section we introduced the weight function
$\wt:\bZ Q_1\to\La$ \eqref{weight} which equips the Jacobian algebra $J=J(Q,W)$ with a \La-grading.
Under certain consistency condition on the brane tiling (satisfied in all our examples), this grading can be used to equip framed moduli spaces of $J$-representations with a torus action and parametrize torus-fixed points (and therefore, as we recall below, framed BPS states in the non-commutative chamber) as molten crystals or pyramids \cite{Mozgovoy:2008fd,Ooguri:2008yb}.

\medskip

Under the above-mentioned consistency conditions, 
it was proved in \MR that two paths $u,v:i\to j$ in $Q$ induce equal elements in $J$ if and only if $\wt(u)=\wt(v)$ (we will say in this case that $u,v$ are equivalent).
Moreover, any nontrivial path has a nonzero weight.
For any path $u:i\to j$, we define $s(u)=i$, the source of $u$, and $t(u)=j$, the target of $u$.
Note that if $s(u)=s(v)$ and their weights are equal, then $t(u)=t(v)$.

\medskip

Let $\De_i$ be the set of equivalence classes of paths
that start at the vertex $i$.
It is a basis of the projective $J$-module $P_i=Je_i$.
We define a partial order on $\De_i$ by the condition that $u\le v$ if there exists a path $w$ such that $wu\sim v$.
A subset $\cI\sbs\De_i$ is called an ideal (or a lower set) if $u\le v$ and $v\in\cI$ imply $u\in\cI$.
In physics literature the poset $\De_i$ is called  a \idef{crystal}, an element $u\in \De_i$ is called an \idef{atom} and the target $t(u)\in Q_0$ is called the \idef{color} of~$u$.
The complement $\De_i\ms\cI$ of a (finite) ideal $\cI$ is called a \idef{molten crystal}.

\medskip

Note that we have an embedding 
\begin{equation}
\wt:\De_i\hookrightarrow\La.
\end{equation}
Therefore we can parametrize atoms by their weights in $\La$ and interpret the crystal $\De_i$ as a subset of the lattice $\La\iso \bZ^{\n{Q_0}+2}$.

\begin{example}
\label{crystal for C3}
Consider the brane tiling
\begin{ctikz}
\path(0,0)pic{triang};
\end{ctikz}
where we identify parallel sides of the parallelogram.
The corresponding quiver has one vertex $1$, three loops $x,y,z$, and potential $W=xyz-xzy$.
Since $d_2(F)=d_2(F')$ for the two terms in $W$,
we have $\La=\bZ Q_1\iso\bZ^3$.
The poset $\De_1$ can be identified with $\bN^3$, where a triple $(k,l,m)$ corresponds to the monomial $x^ky^lz^m\in\bC[x,y,z]\iso J(Q,W)$.
A finite ideal $\cI\sbs\De_1\iso\bN^3$ is a finite subset such that if $(k,l,m)\in\cI$ and $0\le k'\le k$, $0\le l'\le l$, $0\le m'\le m$, then $(k',l',m')\in\cI$.
Such subsets are also called \idef{plane partitions}.
\end{example}

When $|Q_0|>1$, we can still embed $\De_i$ in lower dimensional spaces to make it more visual.
The first embedding was introduced in \MR.
Let $\tl Q$ be the universal cover of~$Q$, embedded in $\bR^2$.
Consider a lift of $i$ in $\tl Q$ which we continue to denote by $i$.
Any path in $Q$ starting at $i$ can be identified with a path in $\tl Q$ that starts at $i$.
It was proved in \MR that for any two vertices $i,j$ in $\tl Q$, there exists a path $v_{ij}:i\to j$ in $\tl Q$ (unique up to an equivalence; called the shortest path from $i$ to $j$) such that any other path $u:i\to j$ in $\tl Q$ is equivalent to $\om^n_j v_{ij}\sim v_{ij}\om_i^n$, where $\om_i$ is a cycle that starts at $i$ and goes along some face (it is unique up to an equivalence) and $n\ge0$.
Usually we omit the starting point of $\om_i$ and write $\om$.
Therefore we obtain a bijection
\begin{equation}
\De_i\xto\sim\tl Q_0\xx\bN\sbs\bR^3,\qquad \om^n v_{ij}\mto (j,n).
\end{equation}
The corresponding partial order on $\tl Q_0\xx\bN$ is generated by the following relations.
For any arrow $a:j\to k$ we have either $av_{ij}\sim v_{ik}$ or $av_{ij}\sim \om v_{ik}$.
Therefore
\begin{equation}
(j,n)<(k,n)\qquad\text{or}\qquad (k,n)< (j,n)<(k,n+1).
\end{equation}

There is an alternative embedding, called \idef{pyramid embedding},
\begin{equation}
\De_i\emb\tl Q_0\xx\bN\sbs\bR^3,
\end{equation}
where the first map is not necessarily a bijection.
For any poset $P$ and any element $u\in P$, define its height $\ell(u)$
to be the length $k$ of the maximal chain $u_0<u_1<\dots<u_k=u$.
Note that if  all cycles in $W$ have equal length, then equivalent paths also have equal length and therefore the height of $u\in\De_i$ is just the length of the path $u$.
Every $u\in\De_i$ can be interpreted as a path $u:i\to j$ in $\tl Q$ (up to equivalence).
The pyramid embedding maps it to
\begin{eqnarray}
u\mto (j,\ell(u))\in\tl Q_0\xx \bN.
\end{eqnarray}
Note that $u$ is mapped to level zero (meaning that $\ell(u)=0$) only if $u$ is the trivial path $e_i$ at $i$.
In this way we depict $\De_i$ as an (upside down) pyramid with the apex $e_i$.

\begin{figure}[ht]
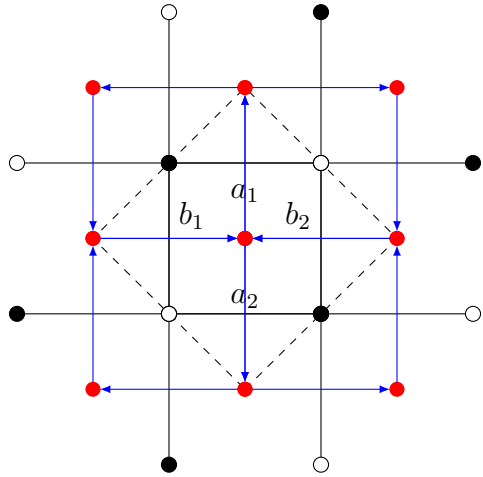

\begin{ctikz}[scale=2]
\draw[dashed](.5,-.5)--(1.5,.5)--(.5,1.5)
--(-.5,.5)--(.5,-.5);
\begin{scope}
	\draw(0,-1)node[crc3]{}--(0,0)--(0,1)node[crc3]{};
	\draw(-1,0)node[crc3]{}--(0,0)node[crc4]{}
	--(1,0)node[crc3]{};
\end{scope}[shift={(-5,0)}]
\begin{scope}[shift={(1,1)}]
	\draw(0,-1)node[crc3]{}--(0,0)--(0,1)node[crc3]{};
	\draw(-1,0)node[crc3]{}--(0,0)node[crc4]{}
	--(1,0)node[crc3]{};
\end{scope}
\begin{scope}[shift={(0,1)}]
	\draw(0,-1)node[crc4]{}--(0,0)--(0,1)node[crc4]{};
	\draw(-1,0)node[crc4]{}--(0,0)node[crc3]{}
	--(1,0)node[crc4]{};
\end{scope}
\begin{scope}[shift={(1,0)}]
	\draw(0,-1)node[crc4]{}--(0,0)--(0,1)node[crc4]{};
	\draw(-1,0)node[crc4]{}--(0,0)node[crc3]{}
	--(1,0)node[crc4]{};
\end{scope}
\begin{scope}[shift={(-.5,-.5)}]
	\draw[ar1](0,0)--(0,1)node[rcrc]{};
	\draw[ar1](0,1)--(1,1)node[rcrc]{};
	\draw[ar1](1,1)--(1,0)node[rcrc]{};
	\draw[ar1](1,0)--(0,0)node[rcrc]{};
\end{scope}
\begin{scope}[shift={(.5,.5)}]
	\draw[ar1](0,0)--(0,1)node[rcrc]{};
	\draw[ar1](0,1)--(1,1)node[rcrc]{};
	\draw[ar1](1,1)--(1,0)node[rcrc]{};
	\draw[ar1](1,0)--(0,0)node[rcrc]{};
\end{scope}
\begin{scope}[shift={(-.5,.5)}]
	\draw[ar1](0,0)--(1,0)node[rcrc]{};
	\draw[ar1](1,0)--(1,1)node[rcrc]{};
	\draw[ar1](1,1)--(0,1)node[rcrc]{};
	\draw[ar1](0,1)--(0,0)node[rcrc]{};
\end{scope}
\begin{scope}[shift={(.5,-.5)}]
	\draw[ar1](0,0)--(1,0)node[rcrc]{};
	\draw[ar1](1,0)--(1,1)node[rcrc]{};
	\draw[ar1](1,1)--(0,1)node[rcrc]{};
	\draw[ar1](0,1)--(0,0)node[rcrc]{};
\end{scope}
\node at(.5,.8){$a_1$};
\node at(.5,.1){$a_2$};
\node at(.15,.65){$b_1$};
\node at(.85,.65){$b_2$};
\end{ctikz}
\caption{The bipartite graph and the dual quiver for the conifold \label{conifoldgraph}}
\end{figure}

\begin{example}[Conifold]
Consider the brane tiling in Figure \ref{conifoldgraph}.
The corresponding quiver has vertices $1,2$, arrows $a_1,a_2:1\to2$, $b_1,b_2:2\to 1$, and the potential
\begin{equation}
\label{Wconifold}
W=a_1b_1a_2b_2-a_1b_2a_2b_1.
\end{equation}
We draw the crystal $\De_1$ as a pyramid in Figure \ref{pyramid}.
We denote paths $u$ with $t(u)=1$ by yellow stones and paths $u$ with $t(u)=2$ by red stones.

\begin{figure}[ht]
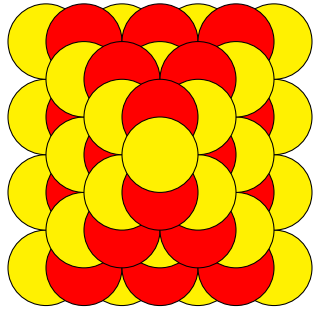

\begin{ctikz}[scale=.5]
\foreach \x in {-3,-1,1,3}
	\foreach \y in {-3,-1,1,3}
		\draw[crc1](\x,\y)circle;
\foreach \x in {-2,0,2}
	\foreach \y in {-3,-1,1,3}
		\draw[crc2](\x,\y)circle;
\foreach \x in {-2,0,2}
	\foreach \y in{-2,0,2}
		\draw[crc1](\x,\y)circle;
\foreach \x in {-1,1}
	\foreach \y in{-2,0,2}
		\draw[crc2](\x,\y)circle;
\foreach \x in {-1,1}
	\foreach \y in{-1,1}
		\draw[crc1](\x,\y)circle;
\draw[crc2](0,1)circle;
\draw[crc2](0,-1)circle;
\draw[crc1](0,0)circle;
\end{ctikz}
\caption{Visualization of the crystal $\De_1$ as a pyramid}
\label{pyramid}
\end{figure}

\end{example}

\sec[NCDT and molten crystals]
\label{NCDT and moltent crystals}
In \eqref{part fun1} we defined the partition function $Z_{\f,\NCDT}(x)$ of NCDT invariants as a generating function of invariants of the moduli spaces $M^{\f,\NCDT}(J,d)$ of $\NCDT$-stable framed representations of the Jacobian algebra. 
The moduli space $M^{\f,\NCDT}(J,d)$ is equipped with an action of the torus $T_\La=\Hom(\La,\bC^\xx)$.
For the framing vector $\f=e_i$, the torus fixed points
are parametrized by finite ideals $\cI\sbs\De_i$ and the partition function has a simple form \MR
\begin{equation}
\label{NC from crystals}
Z_{i,\NCDT}(x)=\sum_{\cI\sbs\De_i}(-1)^{d_i+\hi(d,d)}x^{\udim\cI},
\end{equation}
where $\udim\cI=\sum_{u\in\cI}e_{t(u)}\in\bZ^{Q_0}$ is the dimension vector of the representation corresponding to the ideal $\cI$.
In this way we obtain an interpretation of the partition function of NCDT invariants as counting molten crystals (with signs).
By Theorem \ref{th:NCDT from unframed} we have (for symmetric quivers)
\begin{equation}
\label{eq NC from DT}
Z_{i,\NC}(x)
=\bar S_{e_i}
\Exp\rbr{-\sum_{d}d_i\,\Om(d,1)\ x^d}.
\end{equation}
Let us define the operator
\begin{equation}
\bar T:x^d\mto(-1)^{\hi(d,d)}x^d.
\end{equation}
Then the generating function counting melting crystals in $\De_i$ can be written as
\begin{equation}
\label{eq crystals from DT}
Z_{\De_i}(x)
=\sum_{\cI\sbs\De_i}x^{\udim\cI}
=\bar S_{e_i}\bar T Z_{i,\NC}(x)
=\bar T
\Exp\rbr{-\sum_{d}d_i\,\Om(d,1)\ x^d}.
\end{equation}

\begin{remark}[Algorithms]
\label{algo1}
The above partition function can be computed using the following algorithm,
which works for any subposet (or subcrystal) $P\sbs\De_i$.
For simplicity we assume that all cycles in $W$ have the same length.
We start with the set $P_0=\set{u_1,\dots,u_n}$ of minimal elements in $P$.
Then we apply all admissible arrows to all elements in $P_0$ and obtain the set $P_1$ of elements having height $\le 1$.
Assuming that we constructed the set $P_k$ of elements having height $\le k$, we apply all admissible arrows to the elements of $P_{k}\ms P_{k-1}$ and obtain the set $P_{k+1}$ of all elements having height $\le k+1$ in $P$.
In order to find the terms of the partition function up to degree $n$, we need only ideals $\cI\sbs P$ having $\le n$ elements.
Any element of such semi-ideal has height $\le n-1$, hence it is contained in $P_n$.
This means that we just need to find all ideals (of size $\le n$) in the finite poset $P_n\sbs P$.
This is done using standard algorithms for finding all ideals in finite posets.
See \url{https://github.com/smzg/msinvar} for an implementation.
An alternative algorithm based on the Quiver Yangian of \cite{Li:2020rij} has been
developed by the second-named author \cite{BPinprogress} and is included in the Mathematica package {\tt CoulombHiggs.m}, see Remark
\ref{mathematica}. 
\end{remark}

The above partition function of (numerical) NCDT invariants is itself an important object associated with the Jacobian algebra.
Later we will study partition functions of unframed refined invariants of Jacobian algebra.
Having these refined partition functions, we can determine $Z_{\f,\NCDT}$ using the results of \S\ref{sec:framed-unframed} and \S\ref{sec:framed-unframed2}.
Then we can compare the result with the direct computation obtained by applying the algorithms in Remark \ref{algo1}.

\vspace*{1cm}

\section{BPS indices for small crepant resolutions}
\label{sec_small}
In this section we will consider small crepant resolutions of affine singular CY3 varieties, meaning that the exceptional locus of a resolution has dimension $\le1$.
This implies that compactly supported sheaves on a resolution have support of dimension $\le1$.
In all our examples, we will actually study Jacobian algebras $J(Q,W)$ that are non-commutative crepant resolutions of their centers. Moreover, an important feature of all small crepant resolutions
is that the quiver $Q$ is symmetric.
Therefore the wall-crossing formula translates to the fact that integer DT invariants $\Om(d,y)$ are independent of the stability parameter and satisfy
(see \cf \eqref{hOmfromOm})
\begin{equation}
\hOm(x)=\Exp\rbr{\frac{\sum_d\Om(d,y)x^d}{y\inv-y}},
\end{equation}
where $\hOm(x)$ is the generating function of stacky unframed  invariants corresponding to the trivial stability. Our goal will be to compute this generating function explicitly (or review existing results 
in the literature)
and determine 
the DT invariants $\Om(d,y)=\Om_*(d,y)=\OmS(d,y)$. 
In this way we also shall get access
 to the  partition functions of framed invariants (refined and numerical) as explained in \S\ref{sec:framed-unframed}.
Examples of affine CY3 varieties admitting small crepant resolutions are
\begin{enumerate}
\item Quotients $\bC^2/G\xx \bC$, where $G\sbs\SL_2(\bC)$ is finite.
\item Quotients $\bC^3/G$, where $G\sbs\SO(3)$ is finite.
\item Affine toric CY3 varieties with toric diagrams that don't contain internal nodes.
They are of the form $\set{xy-z^{N_0}w^{N_1}=0}$ for $0\le N_1\le N_0$ or $\bC^3/(\bZ_2\xx\bZ_2)$ with an action of the group given by $(1,0)\mto\diag(-1,-1,1)$,
$(0,1)\mto\diag(1,-1,-1)$.
See Figure \ref{fig:small toric} for the corresponding toric diagrams.
\end{enumerate}

\begin{figure}[ht]
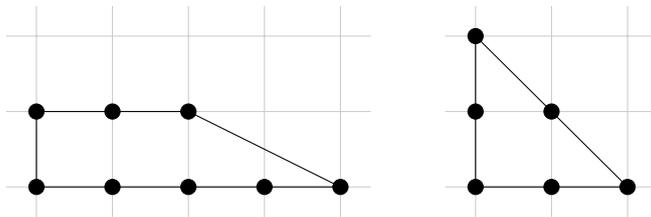

\begin{center}
\tikz{
\draw[step=1cm,black!20,thin] (-.4,-.4) grid (4.4,2.4); 
\draw(0,0)--(4,0)--(2,1)--(0,1)--(0,0);
\filldraw [black]  (0,0) ellipse (0.1 and 0.1);
\filldraw [black]  (1,0) ellipse (0.1 and 0.1);
\filldraw [black]  (2,0) ellipse (0.1 and 0.1);
\filldraw [black]  (3,0) ellipse (0.1 and 0.1);
\filldraw [black]  (4,0) ellipse (0.1 and 0.1);
\filldraw [black]  (2,1) ellipse (0.1 and 0.1);
\filldraw [black]  (1,1) ellipse (0.1 and 0.1);
\filldraw [black]  (0,1) ellipse (0.1 and 0.1);}
\qquad
\tikz{
\draw[step=1cm,black!20,thin] (-.4,-.4) grid (2.4,2.4); 
\draw(0,0)--(2,0)--(0,2)--(0,0);
\filldraw [black]  (0,0) ellipse (0.1 and 0.1);
\filldraw [black]  (1,0) ellipse (0.1 and 0.1);
\filldraw [black]  (2,0) ellipse (0.1 and 0.1);
\filldraw [black]  (0,1) ellipse (0.1 and 0.1);
\filldraw [black]  (1,1) ellipse (0.1 and 0.1);
\filldraw [black]  (0,2) ellipse (0.1 and 0.1);
}
\end{center}
\caption{Toric diagram for $\set{xy-z^{N_0}w^{N_1}=0}$, with $N_0=4$ edges at the bottom and $N_1=2$ edges at the top (left).
Toric diagram for $\bC^3/(\bZ_2\xx\bZ_2)$ (right).}
\label{fig:small toric}
\end{figure}

Note that for the action of $\bZ_N$ on $\bC^3$ given by $1\mto\diag(1,\om,\om\inv)$, $\om=e^{2\pi\bi/N}$, the corresponding quotient $\bC^3/\bZ_N$ is contained in all three families (for the second family one needs to consider a subgroup of $\SO(3)$ conjugate to $\bZ_N\sbs\SL_3(\bC)$; for the third family one considers $N_0=N$, $N_1=0$).
Note also that the quotient $\bC^3/(\bZ_2\xx\bZ_2)$ is contained in the second and the third family.

\subsection{Invariants of \tpdf{$\bC^3$}{C3}}
We consider the quiver $Q$ with one vertex, three loops $x,y,z$ and potential $W=xyz-xzy$.
The corresponding Jacobian algebra is $J(Q,W)=\bC[x,y,z]$.

The following result for unframed stacky invariants of $\bC^3$ was obtained in \cite{behrend_motivic}.
The authors actually study the refined NCDT invariants, but the refined unframed invariants written below can be also obtained from their proof.
For another proof of the formula see \cite{mozgovoy_motivicb}.

\begin{theorem}
\label{thm:BBS}
The generating function of (unframed) stacky invariants for $\bC^3$ is
\begin{equation}
\hOm(x)
=\Exp\rbr{\frac{q^2\sum_{n\ge1} x^n}{q-1}}
=\Exp\rbr{\frac{-y^3\sum_{n\ge1} x^n}{y\inv-y}}.
\end{equation}
\end{theorem}

This implies that the corresponding DT invariants are
\begin{equation}
\label{OmD0C3}
\Om(n,y)=-y^3=(-y)^{-3}\cdot P(\bC^3;y),\qquad n\ge1.
\end{equation}
This is the virtual motive of $\bC^3$ \cite{behrend_motivic} (not to be confused with the virtual Poincar\'e polynomial
$P(\bC^3;y)=y^6$).
Physically, $\Om(n,y)$ counts bound states of $n$ D0-branes on $\bC^3$. The fact that it does not depend on $n$ is the basic property which allows to view D0-branes as Kaluza-Klein gravitons in M-theory \cite{Witten:1995ex}. We note that \eqref{OmD0C3} agrees  in the unrefined limit with an independent computation \cite{Banerjee:2018syt} based on exponential networks, which provide a dual representation of BPS states as  D3-branes wrapped on special Lagrangian cycles in the mirror CY 3-fold \cite{Eager:2016yxd}. It would be interesting to refine
the computations of  \cite{Banerjee:2018syt} so as to include spin dependence, along the lines of \cite{Galakhov:2014xba}.

\begin{remark}
\label{remarkBBS}
One has a similar formula for counting (unframed) D0 invariants of a smooth CY3-fold $\tcX$ \cite{behrend_motivic}
\begin{equation}
\cA_{\tcX,0}(x)
=\Exp\rbr{\frac{q\inv[\tcX]\sum_{n\ge1} x^n}{q-1}}
=\Exp\rbr{\frac{(-y)^{-3}[\tcX]\sum_{n\ge1} x^n}{y\inv-y}},
\end{equation}
where $[\tcX]=P(\tcX;y)$ is the virtual Poincar\'e polynomial defined in \eqref{defP}.
In particular,
let $\cX$ be a singular toric CY3 variety arising from a brane tiling or the corresponding quiver with potential $(Q,W)$.
Let $\tcX$ be a crepant resolution of $\cX$ and $J=J(Q,W)$ be the Jacobian algebra.
Then we have an equivalence of derived categories $D^b_c(\coh\tcX)\iso D^b(\mmod J)$ so that the class of $n$ D0-branes on $\tcX$ is mapped to the dimension vector $n\de$, where $\de=(1,\dots,1)\in\bZ^{Q_0}$.
This dimension vector is contained in the kernel of the skew-symmetric form, hence the DT invariants $\Om_\te(n\de,y)$ are independent of a (generic) stability parameter $\te$ and, in particular, coincide with the attractor invariant $\Om_*(n\de,y)$.
We obtain from the above formula that
\begin{equation}\label{n-delta}
\Om_*(n\de,y)=(-y)^{-3} [\tcX],
\end{equation}
Applying Lemma \ref{lm:Ptoric}, we obtain
\begin{equation}
\Om_*(n\de,y)=-y\inv(y^4+(i+b-3)y^2+i),
\end{equation}
where $i$ and $b$ are respectively the numbers of internal and boundary lattice points in the toric diagram of $\cX$.
\end{remark}

According to \S\ref{sec:framed-unframed},
for any framing vector $\f\in\bN$,
we have
\begin{equation}
Z^\rf_{\f,\NC}(x)
=S_{-\f}\Exp\rbr{ \sum_{n\ge1} 
\frac{(y^{2\f n}-1)\Om(n,y)\, x^n}{y\inv-y}}
=S_{-\f}\Exp\rbr{\sum_{n\ge1}
\frac{y^4(y^{2\f n}-1)x^n}{y^2-1}}.
\end{equation}
In particular, for $\f=1$, we have (\cf\cite{behrend_motivic})
\bea
Z^\rf_{\NC}(x)
&=&S_{-1}\rbr{\Exp\rbr{\sum_{n\ge1}
\frac{q^2(q^{n}-1)x^n}{q-1}}} \nn\\
&=&S_{-1}\rbr{
\prod_{n\ge1}\prod_{k=0}^{n-1}
(1-q^{k+2}x^n)\inv}
=\prod_{n\ge1}\prod_{k=0}^{n-1}
(1-q^{k+2-n/2}(-x)^n)\inv
.
\eea
The partition function of numerical NCDT invariants is given by specialization at $q=1$
\begin{equation}
\label{eq:num NCDT for C3}
Z_\NC(-x)=\prod_{n\ge1}\frac1{(1-x^n)^n}
\end{equation}
which is the MacMahon function.
According to \eqref{NC from crystals}, it can be also computed by counting molten crystals
\begin{equation}
Z_{\NC}(-x)=\sum_{\cI\sbs\De_1}x^{\udim\cI},
\end{equation}
where the crystal is $\De_1=\bN^3$ (see Example \ref{crystal for C3}),
ideals $\cI\sbs\bN^3$ are plane partitions and $\udim \cI=\n \cI$.
It is a theorem of MacMahon that the generating function of plane partitions is given by
\eqref{eq:num NCDT for C3}. Note that the refined topological vertex \cite{Iqbal:2007ii} postulates a different deformation of the generating series \eqref{eq:num NCDT for C3}, which is  
tantamount to $\Om(n,y)=-y$ rather $\Om(n,y)=-y^3$ as in \eqref{OmD0C3}.

\subsection{Invariants of \tpdf{$\bC^2/G\xx\bC$}{C2/GxC}}
\label{sec:inv of quotients}
Let $G\sbs\SL_2(\bC)$ be a finite subgroup and $\cX=\IC^2/G\times \IC$.
The corresponding McKay quiver is isomorphic to the double quiver (meaning that we add opposite arrows) of a quiver $Q$ with an underlying diagram being an extended Dynkin diagram of type $\tl A_n$ for $n\ge0$, $\tl D_n$ for $n\ge4$ or $\tl E_n$ for $n=6,7,8$.
Let $\hat Q$ be the \idef{Ginzburg quiver}, obtained from $Q$ by adding opposite arrows $a^*:j\to i$ for arrows $a:i\to j$ and loops $\ell_i:i\to i$ for $i\in Q_0$.
The potential on $\hat Q$ is given by 
\begin{equation}
\label{eq: =potential in Ginzb quiv}
W=\sum_{(a:i\to j)\in Q_1}(\ell_j aa^*-\ell_i a^*a).
\end{equation}
Then the Jacobian algebra $J=J(\hat Q,W)$ is Morita equivalent to the skew group algebra $G\ltimes\bC[x,y,z]$, where $G$ acts on the coordinates $x,y,z$.
Therefore the derived category $D^b(\mmod J)$ is equivalent to the derived category of $G$-equivariant coherent sheaves with compact support on $\bC^3$.

\medskip

In oder to write down the formula for unframed stacky invariants, we need to recall the root systems of affine type.
Let $Q$ be of type $\tl X_{N-1}=\tl A_{N-1}$, $\tl D_{N-1}$ or $\tl E_{N-1}$.
Let $\De^\f_+$ be the set of positive finite roots of type $X_{N-1}$ and 
$\De^\f=\De^\f_+\cup(-\De^\f_+)$ be the set of all finite roots.
Let $\de$ be the indivisible imaginary root.
Then the positive real roots of type $\tl X_{N-1}$ are
\begin{equation}
\De^\re_+=\De^\f_+\cup\sets{\De^\f+n\de}{n\ge1}
\end{equation}
while the positive imaginary roots are $\De^\im_+=\sets{n\de}{n\ge1}$.
We define $\De_+=\De_+^\re\cup\De_+^\im$.

\begin{remark}
\label{roots A_N}
In type $\tl A_{N-1}$ we can also describe real roots as follows.
Let $\al_0,\dots,\al_{N-1}$ be simple roots (so that $\de=\al_0+\dots+\al_{N-1}$).
For any $i,j\in\bZ_N$, consider the element $\al_{ij}=\al_i+\al_{i+1}+\dots+\al_j$.
Then $\De^\re_+=\sets{\al_{ij}}{j\ne i-1}+\bN\de$.

\end{remark}

We can identify the root lattice of type $\tl X_{N-1}$ and the lattice $\Ga=\bZ^{Q_0}$.
Under this identification positive roots correspond to dimension vectors of indecomposable $Q$-representations. 
The following result was proved in \cite{mozgovoy_motivicb} (see also \cite{bryan_motivica} for the $\tl A_{N-1}$ case).

\begin{theorem}
\label{th:DT inv of quotients}
The generating function of unframed stacky invariants for $J(\hat Q,W)$ is
\begin{equation}
\cA(x)=\Exp\rbr{
\frac{\sum_{d\in\De^\re_+} qx^d
+q(q+N-1)\sum_{d\in\De^\im_+}x^d}{q-1}}.
\end{equation}
\end{theorem}
Note that for the trivial group $G$ (hence $N=1$) we obtain Theorem \ref{thm:BBS}.
The above theorem implies that the corresponding DT invariants are
\begin{equation}
\label{OmZN}
\Om(d,y)=\begin{cases}
-y&d\in\De^\re_+\\
-y(y^2+N-1)&d\in\De^\im_+\\
0&\text{otherwise}
\end{cases}
\end{equation}
The imaginary roots correspond to bound states of $n$ D0-branes, with $\Omega(n\delta)=(-y)^{-3}P(\wtl\cX)$. Real roots instead correspond to D2-D0 bound states, with DT invariants $\Omega(d,y)$ 
interpreted as refined Gopakumar-Vafa invariants \cite{Gopakumar:1998jq,Choi:2012jz}.
For $N=2$, the result \eqref{OmZN} agrees in the unrefined limit with independent computation  based on exponential networks \cite{Banerjee:2019apt}.

\medskip

Applying the results of \S\ref{sec:framed-unframed}, we can determine partition functions of framed invariants for any framed vector $\f$.
In particular, for $\f=e_0$ (corresponding to the extended vertex as well as the trivial representation of $G$),
we obtain numerical NCDT invariants \eqref{eq:num NCDT}
\begin{equation}
Z_{0,\NC}(x)=\bar S_{e_0}
\Exp\rbr{
\sum_{d\in\De^\re_+} d_0x^d
+N\sum_{d\in\De^\im_+}d_0x^d}.
\end{equation}
This implies
\begin{equation}
\label{eq:NC for Z_N}
Z_{0,\NC}(-x_0,x_1,\dots,x_{N-1})=
\prod_{n\ge1}\rbr{(1-x^{n\de})^{N}
\prod_{d\in\De^\f}(1-x^{d+n\de})}^{-n}.
\end{equation}
According to \S\ref{NCDT and moltent crystals} this is also the generating function counting molten crystals.
Equation \eqref{eq:NC for Z_N} for the group $G=\bZ_N$ was originally proved in \cite{young_generating} using the molten crystals interpretation.

\subsection{Invariants of the resolved conifold}
\label{sec:conifold}
Consider the conifold singularity
\begin{equation}
xy-zw=0.
\end{equation}
Its toric diagram is shown on Figure \ref{extoric}. 
%
To see the relation between the singularity and the toric variety, consider an exact sequence
\begin{equation}
0\to\bZ\xto{\smat{1&-1&-1&1}} \bZ^{4}
\xto{\smat{0&0&1&1\\0&1&0&1\\1&1&1&1}} \bZ^3\to0
\end{equation}
where the columns of the second matrix correspond to the rays of the fan.
The corresponding toric variety is equal to $\bC^4/T$, where the action of $T=\bC^*$ on $\bC^4$ is given by 
\begin{equation}
t(x_1,x_2,x_3,x_4)=(tx_1,t\inv x_2,t\inv x_3,tx_4).
\end{equation}
The map
\begin{equation}
(x_1,x_2,x_3,x_4)\mto (x,y,z,w)=(x_1x_2,x_3x_4,x_1x_3,x_2x_4)
\end{equation}
is $T$-equivariant and maps $\bC^4$ to the conifold.
Taking the quotient, we obtain an isomorphism.
On the hand we can represent this toric variety as a moduli space of quiver representations.
Consider the quiver $Q$
\begin{ctikzcd}
0\ar[rr,->>,bend left,"a_1\ a_2"]&&1\lar[ll,->>,bend left,"b_1\ b_2"]
\end{ctikzcd}
Then we can identify $R(Q,\de)$, $\de=(1,1)$, with $\bC^4$
\begin{equation}
M\mto (M_{a_1},M_{b_1},M_{b_2},M_{a_2})\in\bC^4
\end{equation} 
so that $M(Q,\de)=R(Q,\de)/G_\de=\bC^4/T$.

We equip $Q$ with the potential \eqref{Wconifold}
and consider the Jacobian algebra $J=J(Q,W)$.
Note that in this example we actually have $M(J,\de)=M(Q,\de)$.
Crepant resolutions are obtained by considering a generic stability parameter $\te$ (for example $\te=(1,0)$ or $\te=(0,1)$) and the corresponding moduli space $M^\te(J,\de)=M^\te(Q,\de)$.
They are isomorphic to the vector bundle $\cO(-1)\oplus\cO(-1)\to\bP^1$
and have a derived category (of coherent sheaves with compact support) equivalent to the derived category of $J$.
The following result was proved in \cite{morrison_motivica}.

\begin{theorem}
The generating function of unframed stacky invariants for $J(Q,W)$ is
\begin{equation}
\hOm(x)=\Exp\rbr{\frac{(q+q^2)x_0x_1
-q^\oh(x_0+x_1)}{q-1}\sum_{k\ge0}x^{k\de}}.
\end{equation}
\end{theorem}

Equivalently, the only non-zero unframed DT invariants are 
\be
\Omega(n\delta)=-y^3-y\ ,\quad \Omega(n\delta-e_0) = \Omega(n\delta-e_1) = 1 \quad (n\geq 1)
\ee
corresponding to D0-branes and D2-D0 bound states, respectively. These invariants were
first computed in the unrefined limit $y\to 1$ in \cite{Szendroi:2007nu,Joyce:2008pc,gholampour_counting}, and have been recovered
using exponential networks in \cite[(3.29)]{Banerjee:2019apt}. 

\medskip

Applying results of \S\ref{sec:framed-unframed}, we can determine the partition functions of framed invariants (refined and numerical) for any framing vector $\f\in \IN^{Q_0}$ and any stability parameter.
In particular, for $\f=e_0$,
we obtain the partition function of numerical NCDT invariants \eqref{eq:num NCDT}
\bea
\label{eq:NCDT conifold}
Z_{0,\NC}(-x_0,x_1)
&=&\bar S_{e_0} Z_{0,\NC}(x)
=\Exp\rbr{\sum_{k\ge1}
2k(x_0x_1)^{k}-kx_0^{k}x_1^{k-1}-(k-1)x_0^{k-1}x_1^{k} \nn}\\
&=&\prod_{k\ge1}
\frac{(1-x_0^kx_1^{k-1})^k(1-x_0^{k}x_1^{k+1})^{k}}{(1-x_0^kx_1^k)^{2k}}
\eea

According to \S\ref{NCDT and moltent crystals},
the generating function counting melting crystals in the crystal $\De_0$ (see Figure \ref{pyramid}) is equal to
\begin{equation}
\label{eq:conif part1}
Z_{\De_0}(x)=
\sum_{\cI\sbs\De_0}x^{\udim\cI}
=Z_{0,\NC}(x_0,-x_1)
=\prod_{k\ge1}
\frac{(1+x_0^kx_1^{k-1})^k(1+x_0^{k}x_1^{k+1})^{k}}{(1-x_0^kx_1^k)^{2k}}
\end{equation}
This formula was conjectured in 
\cite{szendroi_non-commutative} 
and proved in
\cite{young_computing} 
using the molten crystal interpretation.

\medskip

Generally, for any (generic) stability parameter $\te\in\bR^2$, we have by Theorem
\ref{th:framed from unframed}
\begin{equation}
\bar S_{e_0}Z_{0,\te}
=\Exp\rbr{-\sum_{\te(d)<0} d_0\Om(d,1)x^d},\qquad
\bar S_{e_0}Z_{0,\NC}=\Exp\rbr{-\sum_d  d_0\Om(d,1) x^d}.
\end{equation}
We see from Equation \eqref{eq:NCDT conifold} that the walls occur only for \te satisfying $\te(1,1)=0$, $\te(m,m-1)=0$ or $\te(m-1,m)=0$ for some $m\ge1$.
We can parametrize them as follows.
For any $m\ge0$, consider the rays
\begin{equation}
\ell_m^+=\bR_{\ge0}(1-m,m),\qquad 
\ell_m^-=\bR_{\ge0}(-m,m-1),\qquad 
\ell_\infty=\bR_{\ge0}(-1,1)
\end{equation}
and let $\cC_m^{\pm}$ be the chamber between $\ell_m^\pm$ and $\ell_{m+1}^\pm$ (see Figure \ref{fig:conifoldray}).

\begin{figure}
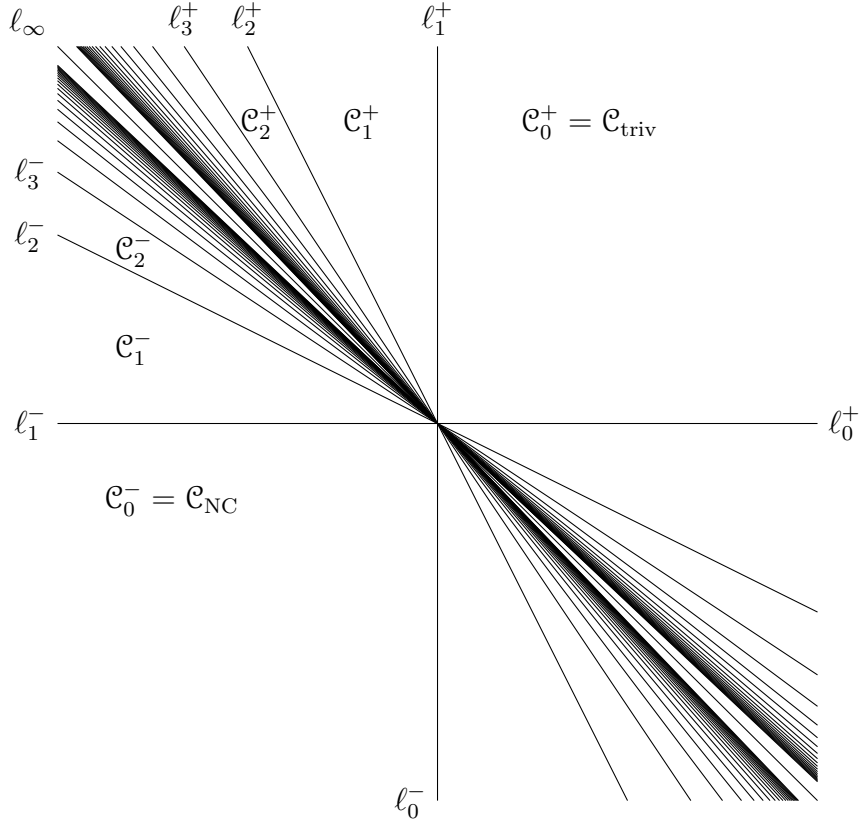

\begin{ctikz}
\foreach[evaluate=\x using (1-\i)/\i*5] 
	\i in {1,2,...,20} {
	\draw(-\x,-5)--(\x,5);
	\draw(5,\x)--(-5,-\x);}
\foreach[evaluate=\x using (1-\i)/\i*5] \i in {1,2,3}{
	\draw(\x,5)node[above]{$\ell^+_\i$};
	\draw(-5,-\x)node[left]{$\ell^-_\i$};}
\draw(0,0)--(-5,5)node[above left]{$\ell_\infty$};
\draw(0,0)--(5,-5);
\draw(5,0)node[right]{$\ell^+_0$};
\draw(0,-5)node[left]{$\ell^-_0$};
\draw(2,4)node{$\cC_0^+=\cC_{\mathrm{triv}}$};
\draw(-1,4)node{$\cC_1^+$};
\draw(-2.35,4)node{$\cC_2^+$};
\draw(-3.5,-1)node{$\cC_0^-=\cC_\NC$};
\draw(-4,1)node{$\cC_1^-$};
\draw(-4,2.3)node{$\cC_2^-$};
\end{ctikz}
\caption{Chamber structure for framed invariants of the resolved conifold. The spectrum of framed BPS states is trivial in the chamber $\cC_0^+$, finite in the chambers $\cC_m^+$ with $m>0$, infinite in the chambers $\cC_m^-$ with $m\geq 0$. The non-commutative chamber corresponds to $\cC_0^-$.}
\label{fig:conifoldray}
\end{figure}

The partition function $Z_{0,\te}$ is independent of $\te$ in a given chamber.
For example, let us consider $\te\in\cC_m^{-}$ of the form $\te=(-m,m-1+\eps)$ for $0<\eps\ll1$.
We have $\te(k,k)<0$ and $\te(k,k-1)<0$ for all $k\ge1$ and $\te(k-1,k)=m-k+\eps<0$ only for $k\ge m+1$.
Therefore, we obtain
\bea
\label{eq:conif framed}
Z_{0,\te}(-x_0,x_1)
&=&\Exp\rbr{\sum_{k\ge1}
\Big(2k(x_0x_1)^{k}-kx_0^{k}x_1^{k-1}\Big)
-\sum_{k\ge m+1}(k-1)x_0^{k-1}x_1^{k}} \nn\\
&=&\prod_{k\ge1}
\frac{(1-x_0^kx_1^{k-1})^k}{(1-x_0^kx_1^k)^{2k}}
\cdot\prod_{k\ge m}(1-x_0^{k}x_1^{k+1})^{k}
\eea
for $\te\in\cC_m^-$.
This generating function can be used to compute melting crystals in a crystal having $m$ stones at the top (for example, the crystal $\De_0$ has one stone at the top). 

\begin{theorem}[see 
{\cite[Theorem 5.13]{nagao_counting}}]
For any $\te\in\cC_m^-$, $m\ge1$, the framed moduli space $M^{e_0}_\te(J,d)$ is isomorphic to
\begin{equation}
M^{e_0}_\te(J,d)\iso M_\NC(J^\f,(\bar d,1)),\qquad
\bar d=\pmat{m&-m+1\\ m+1&-m}d,
\end{equation}
where $J^\f$ is the quotient of $\bC Q^\f$, for the framing vector $\f=me_0$ (with new arrows $r_i:\infty\to 0$ for $1\le i\le m$), with relations induced by $W$ and $a_2r_i=a_1r_{i+1}$ for $1\le i<m$.
\end{theorem}

\begin{remark}
The chamber structure for the conifold was first studied in \cite{nagao_counting,Jafferis:2008uf}. 
The relations $a_2r_i=a_1r_{i+1}$ for $1\le i<m$ can be implemented by adding arrows 
$p_i:1\to\infty$, for $1\le i<m$, and considering a potential $W'=W+\sum_{i=1}^{m-1} p_i(a_2 r_i -a_1 r_{i+1})$ \cite{nagao_counting,Chuang:2008aw}.
\end{remark}

There is a torus action on $M_\NC(J^\f,(\bar d,1))$ so that the fixed points are parametrized by ideals in a poset (pyramid) $\De_0^{(m)}$ having $m$ stones at the top.
Under the map $d\mto \bar d$ described above, we have
\begin{equation}
(k,k)\mto(k,k),\qquad
(k,k-1)\mto(k+m-1,k+m),\qquad
(k-1,k)\mto(k-m,k-m-1).
\end{equation}

Using the formula \eqref{eq:conif framed} for the partition function  $Z_{0,\te}(x)$, we can describe the partition function of melting crystals
\begin{equation}
Z_{\De_0^{(m)}}(x)=
\sum_{\cI\sbs\De_0^{(m)}}x^{\udim\cI},
\end{equation}
\bea
Z_{\De_0^{(m)}}(-x_0,-x_1)
&=&\Exp\rbr{\sum_{k\ge1}
\Big(2k(x_0x_1)^{k}-kx_0^{k+m-1}x_1^{k+m}\Big)
-\sum_{k\ge m+1}(k-1)x_0^{k-m}x_1^{k-m-1}} \nn \\
&=&\Exp\rbr{\sum_{k\ge1}
2k(x_0x_1)^{k}-kx_0^{k+m-1}x_1^{k+m}
-(k+m-1)x_0^{k}x_1^{k-1}} \nn\\
&=&\prod_{k\ge1}
\frac{
(1-x_0^{k}x_1^{k-1})^{k+m-1}
(1-x_0^{k+m-1}x_1^{k+m})^{k}
}{(1-x_0^kx_1^k)^{2k}}
\eea
Note that for $m=1$ we recover
\eqref{eq:conif part1}.
This formula was conjectured in \cite{szendroi_non-commutative} 
and proved in \cite{young_computing,nagao_counting}.

\subsection{Invariants of toric small crepant resolutions}
Consider the toric singularity
\begin{equation}
xy-z^{N_0}w^{N_1}=0,
\end{equation}
where $0\le N_1\le N_0$. This singularity reduces to the conifold for $(N_0,N_1)=(1,1)$,
and is sometimes known as a generalized conifold. 
Its toric diagram has the form as in Figure \ref{fig:small toric},
with $N_0$ edges at the bottom and $N_1$ edges at the top.

The small crepant resolutions are obtained by triangulations of the toric diagram (see Figure \ref{fig:small toric1}) .
We will parametrize them following \MN.
Let $N=N_0+N_1$.
Given a triangulation $\si$,
we parametrize every triangle by the left end of its horizontal edge 
and enumerate triangles from right to left.
Then we obtain a bijection
\begin{equation}
\si=(\si_x,\si_y):I_N=\set{0,\dots,N-1}\to (I_{N_0}\xx\set0)\cup (I_{N_1}\xx \set1),
\end{equation}

\begin{figure}[ht]
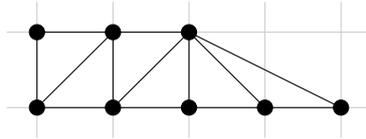

\begin{ctikz}
\draw[step=1cm,black!20,thin] (-.4,-.4) grid (4.4,1.4); 
\draw(0,0)--(4,0)--(2,1)--(0,1)--(0,0)
--(1,1)--(1,0)--(2,1)--(2,0);
\draw(2,1)--(3,0);
\filldraw [black]  (0,0) ellipse (0.1 and 0.1);
\filldraw [black]  (1,0) ellipse (0.1 and 0.1);
\filldraw [black]  (2,0) ellipse (0.1 and 0.1);
\filldraw [black]  (3,0) ellipse (0.1 and 0.1);
\filldraw [black]  (4,0) ellipse (0.1 and 0.1);
\filldraw [black]  (2,1) ellipse (0.1 and 0.1);
\filldraw [black]  (1,1) ellipse (0.1 and 0.1);
\filldraw [black]  (0,1) ellipse (0.1 and 0.1);
\end{ctikz}
\caption{Example of a triangulation of the toric diagram on Figure \ref{fig:small toric}
with $\si=((3,0),(2,0),(1,0),(1,1),$ $(0,0),(0,1))$
and $J=\set{0,1}$}
\label{fig:small toric1}
\end{figure}

We will usually identify $\I=I_N$ with $\bZ_N$.
Define
\begin{equation}
\J=\sets{i\in \I}{\si_y(i)=\si_y(i+1)}
\end{equation}
which enumerates $i\in I$ such that triangles $T_i$, $T_{i+1}$ have adjoint horizontal edges (we consider triangles $T_{N-1},T_0$ for $i=N-1$).
We will assume that the diagonal between the origin and $(1,1)$ is in $\si$ (if $N_1>0$). Then $\si$ is uniquely determined by $J$.
Note that the parity of $J$ is equal to the parity of $I$, see Remark \ref{parity}.

\medskip

The corresponding quiver with potential is constructed as follows.
Let $I=\bZ_N$
and $J\sbs I$ be a subset having the same parity as $I$.
Define a quiver $Q$ with the set of vertices $Q_0=\I$ and with edges
\begin{equation}
\sets{h_i:i\to i+1,\bar h_i:i+1\to i}{i\in \I}\bigcup
\sets{r_i:i\to i}{i\in\J}
\end{equation}
equipped with the potential
\begin{equation}
W
=\sum_{i\in\J}\pm r_i(\bar h_ih_i-h_{i-1}\bar h_{i-1})+
\sum_{i\notin \J}
\pm\bar h_ih_i h_{i-1}\bar h_{i-1},
\end{equation}
where the signs are chosen in such way that every arrow appears with opposite signs (this is possible as $J$ has the same parity as $I$). The quiver $(Q,W)$ can be glued into a brane tiling \MN as was explained in \S\ref{sec:tilings}. Different triangulations lead to equivalent quivers related by mutations.

\begin{example}
Consider $I=\set{0,1,2}$ and $J=\set0,\set1$ or $\set{0,1,2}$.
The corresponding potentials are
\begin{gather}
W=r_0(\bar h_0h_0-h_2\bar h_2)-\bar h_1h_1h_0\bar h_0+\bar h_2h_2h_1\bar h_1,\\
W=\bar h_0h_0h_2\bar h_2
+r_1(\bar h_1h_1-h_0\bar h_0)-h_2h_2h_1\bar h_1,\\
W=r_0(\bar h_0h_0-h_2\bar h_2)
+r_1(\bar h_1h_1-h_0\bar h_0)
+r_2(\bar h_2h_2-h_1\bar h_1)
.
\end{gather}
\end{example}

\begin{example}[Quotient $\bC^2/\bZ_N\xx\bC$]
\label{quotient ex}
Consider $N_1=0$, $N_0=N$ and a triangulation
\begin{ctikz}
\draw[black!20](-.4,-.4)grid(4.4,1.4); 
\draw(0,0)--(4,0);
\filldraw [black]  (0,1) ellipse (0.1 and 0.1);
\foreach \x in {0,1,2,3,4} \draw(0,1)--(\x,0);
\foreach \x in {0,1,2,3,4} \filldraw [black]  (\x,0) ellipse (0.1 and 0.1);
\end{ctikz}
Then $I=J=\set{0,\dots,N-1}$ and the quiver $Q$ is the Ginzburg quiver of the cyclic quiver $C_N$, obtained by taking the double quiver of $C_N$ and adding loops at all vertices.
The potential was described in \eqref{eq: =potential in Ginzb quiv}.
\end{example}

\begin{example}[Conifold]\label{conifold ex}
Consider $N_0=N_1=1$ and a triangulation
\begin{ctikz}
\draw[step=1cm,black!20,thin] (-.4,-.4) grid (1.4,1.4); 
\draw(0,0)--(0,1)--(1,1)--(1,0)--(0,0)--(1,1);
\filldraw [black]  (0,0) ellipse (0.1 and 0.1);
\filldraw [black]  (1,0) ellipse (0.1 and 0.1);
\filldraw [black]  (0,1) ellipse (0.1 and 0.1);
\filldraw [black]  (1,1) ellipse (0.1 and 0.1);
\end{ctikz}
Then $I=\bZ_2$ and $J=\es$.
We obtain the same quiver with potential as was studied in \S\ref{sec:conifold}, with an identification
\begin{equation}
a_1=h_0,\qquad a_2=\bar h_1,\qquad b_1=\bar h_0,
\qquad b_2=h_1.
\end{equation}
\end{example}

\begin{example}[Suspended pinch point]\label{spp ex}
Consider $N_0=2$, $N_1=1$ and a triangulation
\begin{ctikz}
\draw[step=1cm,black!20,thin] (-.4,-.4) grid (2.4,1.4); 
\draw(0,0)--(2,0)--(1,1)--(0,1)--(0,0)--(1,1)--(1,0);
\filldraw [black]  (0,0) ellipse (0.1 and 0.1);
\filldraw [black]  (1,0) ellipse (0.1 and 0.1);
\filldraw [black]  (0,1) ellipse (0.1 and 0.1);
\filldraw [black]  (1,1) ellipse (0.1 and 0.1);
\filldraw [black]  (2,0) ellipse (0.1 and 0.1);
\end{ctikz}
Then $I=\bZ_3$, $\J=\set0$ and we obtain the quiver shown in Figure \ref{quivSPP}.
Labelling the vertices by $0,1,2$ and the arrows by
\begin{equation}
h_i:i\to i+1,\qquad \bar h_i:i+1\to i,\qquad r_0:0\to0.
\end{equation}
the potential is given by
\begin{equation}
W=r_0(\bar h_0 h_0-h_2\bar h_2)
-\bar h_1 h_1h_0h_0
+\bar h_2 h_2h_1\bar h_1.
\end{equation}
\end{example}

\begin{figure}[ht]
\begin{tikzcd}
&0\ar[<->,dr]\ar[<->,dl]\ar[loop above]\\
1\ar[<->,rr]&&2
\end{tikzcd}
\hspace*{2cm}
\raisebox{-2cm}{\includegraphics[height=5cm]{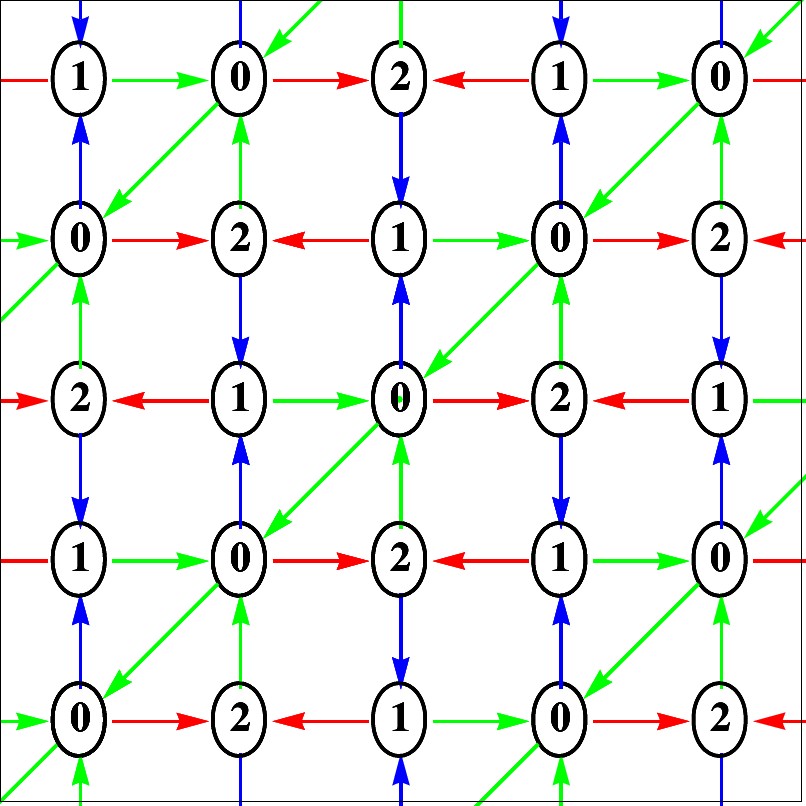}}
\caption{Quiver and tiling for the suspended pinched point}
\label{quivSPP}
\end{figure}

\begin{remark}\label{parity}
In what follows we will consider only triangulations with $\J=\set{0,\dots,N'-1}$ for some $N'\ge0$.
One can show that such triangulation is unique and that
$N'=N_0-N_1$ (assuming that the diagonal between $(0,0)$ and $(1,1)$ is in the triangulation).
Note that $N'$ and $N$ have the same parity.
Since mutations (swapping a diagonal) do not change the parity of~$J$, we conclude that $J$ always has the same parity as $I$.
\end{remark}

In the statement of the following theorem we use the root system of type $\tl A_{N-1}$ which was described in Remark \ref{roots A_N}.

\begin{theorem}[see {\MN}]
Let $N=N_0+N_1$, $N'=N_0-N_1$, $I=\bZ_N$, $\J=\set{0,\dots,N'-1}$ and $(Q,W)$ be the quiver with potential defined above.
Then the generating function of unframed stacky invariants for $J(Q,W)$ is given by
\begin{equation}
\cA(x)=\Exp\rbr{
\frac{\sum_{d} \Omega(d,y) x^d}{y^{-1}-y}},
\end{equation}
where, for $d\in\bN^N$, 
\begin{enumerate}
\item $\Omega(d,y)=-y$ if $d\in\De^\re_+$ and $\sum_{i\notin\J}d_i$ is even.
\item $\Omega(d,y)=1$
if $d\in\De^\re_+$ and $\sum_{i\notin\J}d_i$ is odd.
\item $\Omega(d,y)=-y(y^2+N-1)$ if $d\in\De^\im_+$.
\item $\Omega(d,y)=0$ otherwise.
\end{enumerate}
\end{theorem}
Cases 1 and 2 correspond to D2-D0 bound states wrapped on a rational curve with normal bundle 
$\cO(-1)\oplus\cO(-1)$ and $\cO(0)\oplus\cO(-2)$ respectively, while case 3 corresponds to D0-branes, see Remark \ref{remarkBBS}.

\begin{example}[Suspended pinch point]
Consider the quiver with potential from Example~ \ref{spp ex}.
We have $\de=\al_0+\al_1+\al_2=(1,1,1)$ and
$$\De_+^\re=\set{\al_0,\al_1,\al_2,\al_0+\al_1,\al_1+\al_2,\al_0+\al_2}+\bN\de.$$
We have $J=\set0$, $s(d)=\sum_{i\notin\J}d_i=d_1+d_2$ and $s(\de)$ even.
The real roots $d$ with even $s(d)$ are
$$\set{\al_0,\al_1+\al_2}+\bN\de.$$
The real roots $d$ with odd $s(d)$ are
$$\set{\al_1,\al_2,\al_0+\al_1,\al_0+\al_2}+\bN\de.$$
Therefore
\begin{equation}
\label{spp refined}
\hOm(x)=\Exp\rbr{\frac{
q(x_0+x_1x_2)
-q^\oh(x_1+x_2+x_0x_1+x_0x_2)
+q(q+2)x^\de
}{q-1}\sum_{n\ge0}x^{n\de}}.
\end{equation}
Applying \eqref{eq NC from DT} we obtain the partition function
of NCDT invariants
\bea
\bar S_{e_0}Z_{0,\NCDT}(x)
&=&\Exp\rbr{
\sum_{k\ge0}
((k+1)x_0+kx_1x_2-(k+(k+1)x_0)(x_1+x_2)
+3(k+1)x^\de)x^{k\de}} \nn\\
&=&\prod_{k\ge0}\frac
{(1-x_1x^{k\de})^k(1-x_2x^{k\de})^k(1-x_0x_1x^{k\de})^{k+1}(1-x_0x_2x^{k\de})^{k+1}}
{(1-x_0x^{k\de})^{k+1}(1-x_1x_2x^{k\de})^k(1-x^{(k+1)\de})^{3(k+1)}}
\eea
Similarly, applying \eqref{eq crystals from DT}
we obtain the partition function of melting crystals in the crystal $\De_0$
$$Z_{\De_0}(x_0,-x_1,-x_2)=\bar T Z_{\De_0}(x)
=\bar S_{e_0} Z_{0,\NC}(x).$$
Therefore
\begin{equation}
Z_{\De_0}(x)
=\prod_{k\ge0}\frac
{(1+x_1x^{k\de})^k(1+x_2x^{k\de})^k(1+x_0x_1x^{k\de})^{k+1}(1+x_0x_2x^{k\de})^{k+1}}
{(1-x_0x^{k\de})^{k+1}(1-x_1x_2x^{k\de})^k(1-x^{(k+1)\de})^{3(k+1)}}.
\end{equation}
This formula proves a conjecture from \MR.
\end{example}

\subsection{Invariants of \tpdf{$\bC^3/(\bZ_2\xx\bZ_2)$}{C3/Z2xZ2}}
Consider the quotient $\bC^3/(\bZ_2\xx\bZ_2)$, where the action of the group is given by $(1,0)\mto\diag(-1,-1,1)$,
$(0,1)\mto\diag(1,-1,-1)$.
The corresponding toric diagram is contained in Figure \ref{fig:small toric}.
We construct the corresponding quiver with potential following Example \ref{ex:bt of quotient}.
The McKay quiver $Q$ has vertices $0,1,2,3$, arrows $a_{ij}:i\to j$ for all $i\ne j$ and 
potential 
\begin{equation}
W=\sum_{i,j,k}\pm a_{ki}a_{jk}a_{ij}
\end{equation}
where the sum runs over all triples of elements in $Q_0$.
The quiver and brane tiling are shown on Figure \ref{quivZ22}.
The following result is proved in \cite{mozgovoy_invariantsa}.
Here we identify $Q_0$ with $\bZ_4$ and write $\de=(1,1,1,1)\in\bZ^{Q_0}$.

\begin{figure}[ht]
\begin{tikzcd}[sep=2cm]
0\rar[<->]\ar[dr,<->]&1\dar[<->]\ar[dl,<->]\\
3\uar[<->]&2\lar[<->]
\end{tikzcd}
\hspace*{2cm}
\raisebox{-2cm}{\includegraphics[height=5cm]{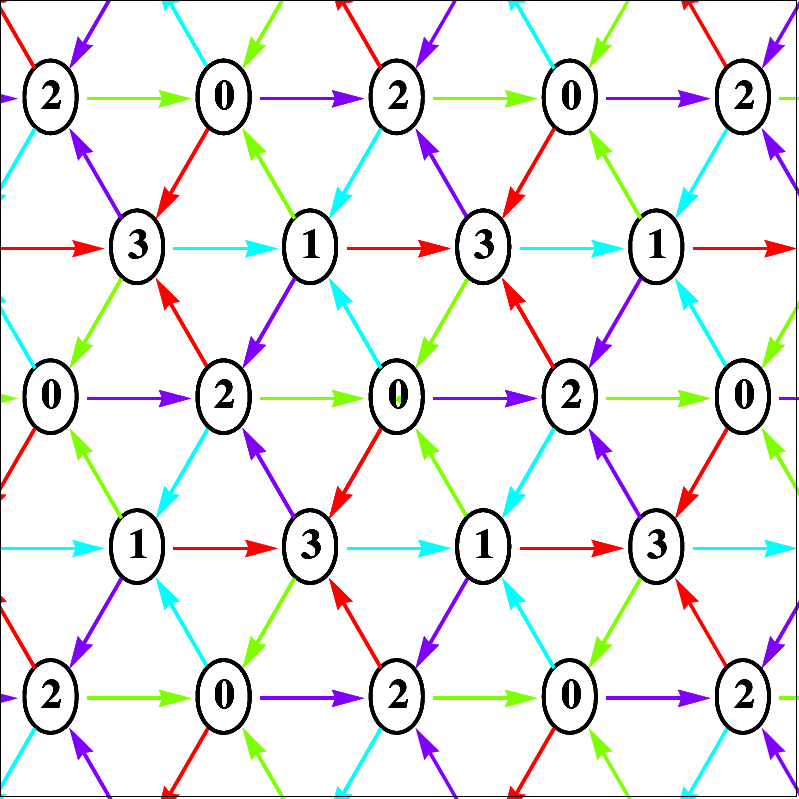}}
\caption{Quiver and tiling for $\bC^3/(\bZ_2\xx\bZ_2)$}
\label{quivZ22}
\end{figure}

\begin{theorem}
The generating function of unframed stacky invariants for $J(Q,W)$ is
\begin{equation}
\hOm(x)=\Exp\rbr{
\frac{q\sum_{i<j}x_ix_j-q^\oh\sum_i(x_i+x_ix_{i+1}x_{i+2})+q(q+3)x^\de}{q-1}\sum_{n\ge0}x^{n\de}
}.
\end{equation}
\end{theorem}
Note that the summands of the form $x_ix_{i+2}$ above don't correspond to any roots of type~$\tl A_3$.
Applying \eqref{eq NC from DT} we obtain the partition function of unrefined NCDT invariants
\bea
Z_{0,\NC}(-x_0,x_1,x_2,x_3)\hspace*{10cm}  \nn\\
=\Exp\Bigg(\sum_{n\ge1}nx^{n\de}\bigg(
4+\sum_{\ov{i\ne j}{i,j\ne0}}(x_ix_j+(x_i x_j)\inv)
-x_1x_2x_3-(x_1x_2x_3)\inv-\sum_{i\ne0}(x_i+x_i\inv)
\bigg)\Bigg) \nn \\
=M(1,x^\de)^4\frac
{\tl M(x_1x_2,x^\de)\tl M(x_1x_3,x^\de)\tl M(x_2x_3,x^\de)}
{\tl M(x_1x_2x_3,x^\de)\tl M(x_1,x^\de)
\tl M(x_2,x^\de)\tl M(x_3,x^\de)},
\eea
where we define
\begin{gather}
M(q,x)=\prod_{n\ge1}(1-qx^n)^{-n}
=\Exp\rbr{\sum_{n\ge1}nqx^n},\\
\tl M(q,x)=M(q,x)M(q\inv,x)
=\Exp\rbr{\sum_{n\ge1}n(q+q\inv)x^n}.
\end{gather}
This formula was proved in
\cite{young_generating} using the molten crystal interpretation.

\begin{remark}
Throughout this section, we have considered quivers with superpotential coming from a brane tiling.
Donaldson-Thomas invariants for deformations of the standard potential for $\IC^3$, $\IC^2/\IZ_N\times \IC$ and the conifold have been computed in
\cite{Cazzaniga:2015fwa} and references therein, and exhibit jumps in complex codimension 2 as
the deformation parameter is varied.
\end{remark}

\vspace*{1cm}

\section{Attractor indices for local surfaces \label{sec_surface}}
In this section we study attractor invariants of some local surfaces (line bundles over smooth projective surfaces). As explained in \S\ref{sec:tilingcrystal}, they arise as crepant resolutions of affine toric CY3 varieties associated to brane tilings.
In contrast to the previous section, these crepant resolutions are not small, so there are exceptional divisors (corresponding to internal points of the toric diagram), hence two-dimensional compact subvarieties.
On the algebraic side the problem becomes significantly more difficult since the quiver is no
longer symmetric and therefore the quantum affine plane is not commutative.
Instead of the DT invariants $\Om(d,y)$ considered in the symmetric case (independent of a stability parameter), in the non-symmetric case we will study the attractor DT invariants $\Om_*(d,y)$.
It was conjectured in \cite{beaujard_vafa} that they have a particularly nice behavior, namely they vanish unless $d=e_i$ or $d$ is contained in the kernel of the skew-symmetric form $\ang{-,-}$.
In this section we will compute attractor DT invariants $\Om_*(d,y)$ explicitly for small $d$.
They turn out to have a particularly simple form, which suggests 
a natural conjecture for the value of $\Om_*(d,y)$ for arbitrary dimension vectors. 
As explained in \S\ref{sec_wc}, having a general formula for attractor DT invariants, we can compute 
 DT invariants (both framed and unframed) for any stability parameter. In particular, we shall
 compute the framed DT invariants in the non-commutative chamber, and find agreement with
 the counting of molten crystals in the unrefined limit.

\subsection{Double dimensional reduction}
\label{double reduction}
Let $(Q,W)$ be a quiver with a potential and let  $I\sbs Q_1$ be a cut.
Then the generating function of unframed refined invariants is given by 
\eqref{stacky inv with a cut} (the first dimensional reduction)
\begin{equation}
\hOm(x)
=\sum_{d\in \IN^{Q_0}} (-y)^{\hi_Q(d,d)+2\ga_I(d)}\frac{P(R(J_I,d))}{P(G_d)}x^d,\qquad
\ga_I(d)=\sum_{(a:i\to j)\in I}d_id_j,
\end{equation}
where
\begin{equation}
\label{J_I}
J_I=J_I(Q,W)=\bC Q'/(\dd W/\dd a\col a\in I),\qquad Q'=(Q_0,Q_1\ms I).
\end{equation}


Assume that there is another cut $I'\sbs Q_1$ disjoint from $I$.
Define $Q''=(Q_0,Q_1\ms(I\cup I'))$ and consider the forgetful map
\begin{equation}
\label{proj}
\pi:R(J_I,d)\to R(Q'',d)
\end{equation}
having linear fibers.
Given a $Q''$-representation $M$, let $\vi(M)$ denote the dimension of the fiber $\pi\inv(M)$.
Then $\vi(M)$ is quadratic, meaning that
there exist values $\vi(M,N)$ such that $\vi(M)=\vi(M,M)$ and $\vi(\bop_i M_i,\bop_j N_j)=\sum_{i,j}\vi(M_i,N_j)$.

\medskip

Let $\cI$ be the set parameterizing all indecomposable $Q''$-representations.
Then every $Q''$-representation can be written in the form $M=\bop_{X\in \cI}X^{\oplus m_X}$
for some map $m:\cI\to\bN$ with finite support.
Therefore the above generating function can be written in the form (the second dimensional reduction)
\bea
\label{part function}
\hOm(x)
&=& \sum_{\ov{M\in\Rep Q''}{d=\udim M}}
(-q^\oh)^{\hi_Q(d,d)+2\ga_I(d)}\frac{q^{\vi(M)}}{[\Aut M]}x^{d} \nn\\
&=&\sum_{m:\cI\to\bN}
\frac{(-q^\oh)^{-\sum_{M,N\in\cI}m_M m_N\si(M,N)}}
{\prod_{M\in\cI}(q\inv)_{m_M}}
x^{\sum_{M\in\cI} m_M\udim M}
\eea
where $(q)_n$ was defined in \eqref{PGLn}, and  
 the `interaction form' $\si:\cI\xx\cI\to\bZ$ is
given by
\begin{equation}\label{sigma}
\si(M,N)=2h(M,N)-2\vi(M,N)-\rho(M,N),
\end{equation}
\begin{equation}
h(M,N)=\dim\Hom(M,N),\qquad
\rho(d,e)=\hi_Q(d,e)+2\sum_{(a:i\to j)\in I}d_ie_j,
\end{equation}
where $\rho(M,N)=\rho(\udim M,\udim N)$.
In the following examples, we shall choose the disjoint cuts $I,I'$ such that 
the remaining quiver $Q''$ is  simple enough that we are able to classify all its indecomposable representations. In that case we can then apply \eqref{part function} to compute the generating function $\hOm(x)$ and read off the DT invariants for any stability parameter using the methods explained in \S\ref{sec_wc}.

\subsection{General action of \tpdf{$\IZ_N$ on $\bC^3$}{ZN on C3}}

\label{sec:general action}
Consider a finite subgroup $\bZ_N\sbs\SL_3(\bC)$.
Choosing a basis of $\bC^3$ we can represent the action of $\bZ_N$ on $\bC^3$ as
\begin{equation}
1\mto\diag(\om,\om^k,\om^{-k-1}),
\end{equation}
where $\om$ is a primitive $N$-th root of $1$ and $0\le k< N/2$. The toric diagram 
is the convex hull of the vectors (equivalent to Figure 4 in \cite{Hanany:2005ve})
\be
\hat v_1 = \begin{pmatrix} 1\\0 \end{pmatrix}, \quad 
\hat v_2 = \begin{pmatrix} 0\\1 \end{pmatrix}, \quad 
\hat v_3 = \begin{pmatrix} -k\\k+1-N \end{pmatrix}
\ee
The corresponding McKay quiver  (\cf Example \ref{ex:bt of quotient}) has vertices $Q_0=\bZ_N$,  arrows
\begin{equation}
Q_1:\quad a_i:i\to i+1,\qquad b_i:i\to i+k,\qquad c_i:i\to i-k-1,
\qquad i\in\bZ_N.
\end{equation}
and potential 
\begin{equation}
\label{potential for Z_N}
W=\sum_{i\in \IZ_N} c_{i+k+1}(b_{i+1}a_i-a_{i+k}b_i).
\end{equation}
We choose two cuts
\begin{equation}
I=\sets{c_i}{i\in \bZ_N},\qquad
I'=\sets{b_i}{i\in \bZ_N}.
\end{equation}
such that the quiver $Q''=Q\ms(I\cup I')$ is the cyclic quiver $C_N$ (having vertices $i\in\bZ_N$ and arrows $a_i:i\to i+1$ for $i\in\bZ_N$). As recalled below, we can easily parametrize the
set $\cI$ of indecomposable representations of $C_N$.
For any representation $M$ of this quiver, the arrows $b_i\in I'$ correspond to a morphism $M\to \Si^k M$, where the functor $\Si$ is defined by 
\begin{equation}
\Si:\Rep C_N\to\Rep C_N,\qquad
(\Si M)_i=M_{i+1}.
\end{equation}
This implies that the dimension of the fiber of \eqref{proj} over $M$ is equal to
$\vi(M)=h(M,\Si^k M)$, hence we have $\vi(M,M')=h(M,\Si^k M')$.
Using notation from \S\ref{double reduction},
we obtain
\begin{equation}
\label{eq:interaction form2}
\si(M,M')=2h(M,M')-2h(M,\Si^k M')-\rho(M,M'),
\end{equation}
\begin{equation}
\rho(d,e)=\hi_Q(d,e)
+2\sum_{(a:i\to j)\in I}d_ie_j
=\sum_i d_i(e_i-e_{i+1}-e_{i+k}+e_{i-k-1}).
\end{equation}
Now we have all necessary ingredients to compute the generating function $\hOm(x)$ of stacky unframed invariants of $(Q,W)$ using the formula
\eqref{part function}.

\medskip

Let us now describe the indecomposable representations of the quiver $Q''=C_N$ in more detail.
Dimension vectors $d$ of indecomposable representations are parametrized by the (positive) roots of type $\tl A_{N-1}$ (see Remark \ref{roots A_N}). The representations associated to $d$ depend on whether $d$ is a real or imaginary positive root:
\begin{itemize}
\item For any real root $d\in\De_+^\re\sbs\bN^{Q_0}$, there is just one indecomposable representation having dimension vector $d$.
More precisely, for $d=\al_{ij}+n\de$ with $i,j\in\bZ_N$ and $n\ge0$
(recall from Remark \ref{roots A_N} that $\al_{ij}=\al_i+\al_{i+1}+\dots+\al_j$ and $j\neq i-1$ for real roots), the corresponding indecomposable representation $X_{i,j,n}$ has a basis consisting of vectors $e_i,e_{i+1},\dots,e_{i+\ell+nN}$, where $0\le \ell<N$ is such that $\ell\equiv j-i\pmod N$.
We place $e_k$ at the vertex $k\pmod N$ of $C_N$. 
The arrows of $C_N$ send $e_k$ to $e_{k+1}$ for $i\le k<i+\ell+nN$ and send $e_{i+\ell+nN}$ to zero.
Note that all of these representations are nilpotent.
\item For any imaginary root $d=n\de$, there are $N$  nilpotent representations
of the form $X_{i,i-1,n-1}$ with $i\in \IZ_N$, as well as one parameter families of 
representations $X_{n,\la}$  with $\la\in\bC^*$  (we call them invertible representations)
of the following form:  $X_{n,\la}$ has a vector space $\bC^n$ at every vertex, identity matrices for 
all arrows $a_i:i\to i+1$ except for one arrow $a_j$ with Jordan block $J_{n,\la}$ (different choices
of $j\in \IZ_N$ lead to isomorphic representations).
\end{itemize}

\begin{remark}
The above indecomposable nilpotent representations can be alternatively parameterized by pairs $(i,j)$, where $i\le j$ are integers and $0\le i<N$.
Given representations $M,M'$ corresponding to pairs $(i,j)$ and $(k,\ell)$ respectively, the dimension of the vector space $\Hom(M,M')$ is equal to the number of integers $k\le s\le \ell$ such that $s\equiv i\pmod N$ and $j-i\ge\ell-s$.
\end{remark}

We note that the interaction form \eqref{eq:interaction form2} is zero if one of the representations $M,M'$ is invertible.
This implies that we can decompose
\begin{equation}
\hOm(x)=\hOm^\bi(x)\cdot\hOm^\bn(x),
\end{equation}
where $\hOm^\bi(x)$ and $\hOm^\bn(x)$
are defined as in \eqref{part function} with the sums running over (indecomposable) invertible and nilpotent representations of $C_N$ respectively.
We can compute the series $\hOm^\bi(x)$ explicitly as follows.
The corresponding series for one eigenvalue $\la\in\bC^*$ is given by
\begin{equation}
\hOm^\bi_0(x)
=\prod_{n\ge1}\rbr{\sum_{m\ge0}\frac{x^{mn\de}}{(q\inv)_m}}
=\Exp\rbr{\sum_{n\ge1}\frac{qx^{n\de}}{q-1}}.
\end{equation}
Therefore the generating function for all invertible representations is given by
\begin{equation}
\hOm^\bi(x)
=\hOm^\bi_0(x)^{q-1}
=\Exp\rbr{q\sum_{n\ge1}x^{n\de}}.
\end{equation}

The generating function $\hOm^\bn(x)$ is significantly more complicated and it is unclear if one can find a closed formula for it.
Yet we have all the necessary ingredients to determine it for small dimension vectors using a computer.
Then we can find the corresponding attractor invariants $\Om_*^\bn(d,y)$ for the series $\hOm^\bn(x)$ so that attractor invariants for the series $\hOm(x)$ are given by
\begin{equation}
\Om_*(d,y)=\Om_*^\bn(d,y)+\Om_*^\bi(d,y)
=\Om_*^\bn(d,y)+
\begin{cases}
y^2(y\inv-y)&d=n\de,\ n\ge1,\\
0&\text{otherwise.}
\end{cases}
\end{equation}

\label{forgetful}
Before proceeding to discuss the simplest example $\cX=\IC^3/\IZ_3$,
we introduce the following construction which will be useful later.
Consider the cyclic quiver $C_N$ and a vertex $i\in \bZ_N$.
Then we define 
a functor that forgets the vertex~$i$
\begin{equation}
F_{i}:\Rep C_N\to \Rep C_{N-1},\qquad M\mto (M_0,\dots,M_{i-1},M_{i+1},\dots, M_{N-1}),
\end{equation}
where the linear maps between vector spaces 
are the same as before except for $M_{i-1}\to M_{i+1}$ which is the composition $M_{i-1}\to M_i\to M_{i+1}$.
If $i=0$, then we start with $M_1$.

More generally, let $I=\set{i_1,\dots,i_m}$, where $0\le i_1<\dots<i_m<N$.
Then the functor
\begin{equation}
F_{I}=F_{i_1}\circ\dots\circ F_{i_m}:\Rep C_N\to \Rep C_{N-m}
\end{equation}
forgets vertices from $I$.

\subsection{Invariants of \tpdf{$\IP^2$}{P2}}
\label{secP2}

Consider the action of $\bZ_3$ on $\bC^3$ given by
\begin{equation}
1\mto\diag(\om,\om,\om)
\end{equation}
where $\om=e^{2\pi \bi/3}$.
The corresponding McKay quiver with potential $(Q,W)$ was described in the previous section,
and was studied in detail in the physics literature \cite{Douglas:2000qw}. It is shown along with
the brane tiling in Figure \ref{quivP2}.

\begin{figure}[ht]
\begin{tikzcd}[sep=2cm,column sep=1.3cm]
&1\drar[->>>]\\
0\ar[ur,->>>]&&2\ar[ll,->>>]
\end{tikzcd}
\hspace*{2cm}
\raisebox{-2cm}{\includegraphics[height=5cm]{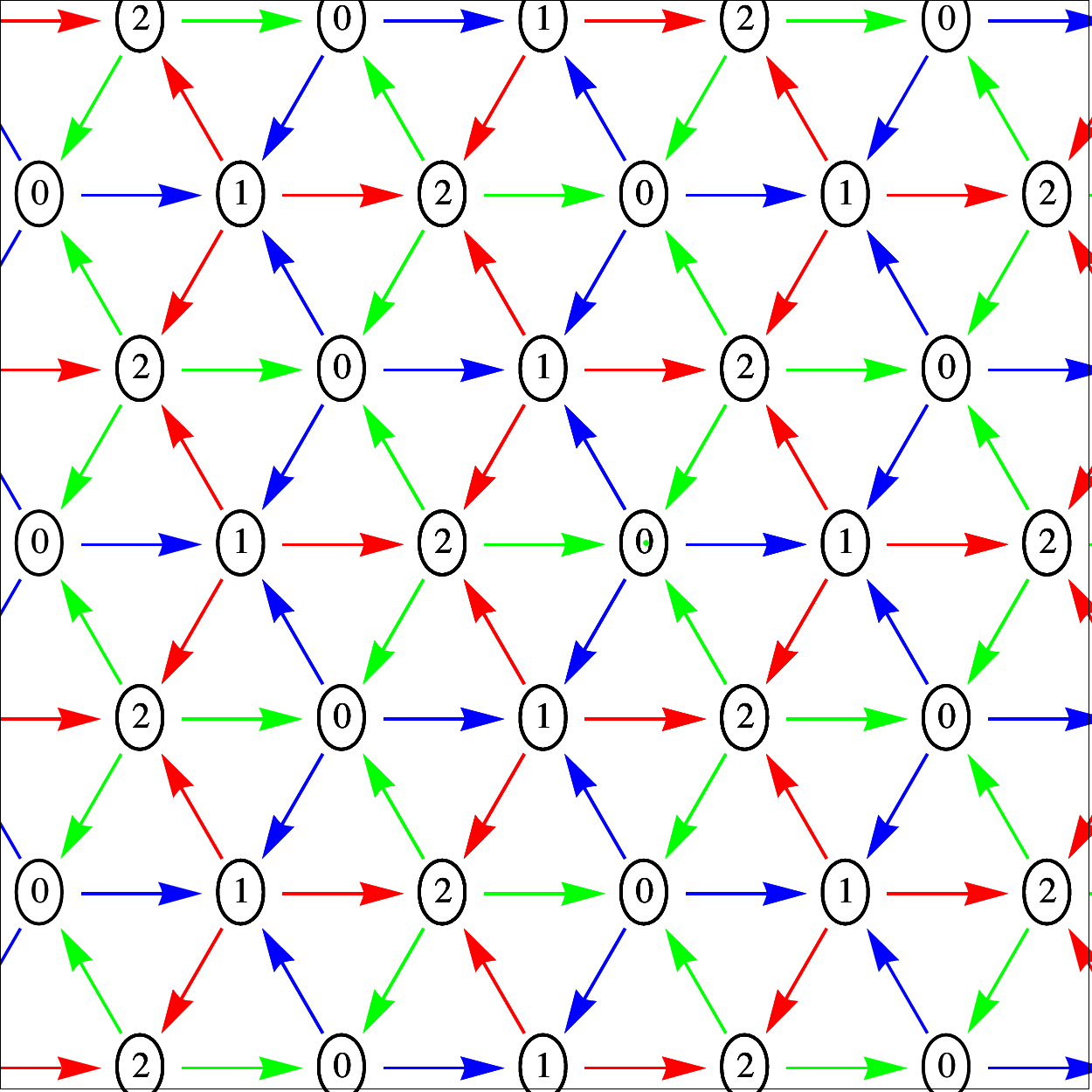}}
\caption{Quiver and tiling for $\IC^3/\IZ_3$}
\label{quivP2}
\end{figure}

Note that the quotient $\cX=\bC^3/\bZ_3$ has a crepant resolution $\widetilde\cX=K_{\bP^2}=\cO(-3)$, the canonical bundle of $\bP^2$ (see the toric diagram in Figure \ref{extoric}).
On the other hand, the Jacobian algebra $J=J(Q,W)$ is a non-commutative crepant resolution of $\bC^3/\bZ_3$.
Let us choose a cut $I=\set{a_2,b_2,c_2:2\to0}$.
Then the algebra $J_I=J_I(Q,W)$ \eqref{J_I} is given by the quiver
\begin{ctikzcd}[sep=2cm]
0\rar[->>>,"a_0\ b_0\ c_0"]&
1\rar[->>>,"a_1\ b_1\ c_1"]&2
\end{ctikzcd}
subject to the relations 
\begin{equation}
a_1b_0=b_1a_0,\qquad
c_1b_0=b_1c_0,\qquad
a_1c_0=c_1a_0.
\end{equation}
Note that it is isomorphic to the algebra of endomorphisms $\End(\cO\oplus\cO(1)\oplus\cO(2))$ corresponding to the Beilinson exceptional sequence $\cO,\cO(1),\cO(2)$ on $\bP^2$.
Morally, under these identifications, the canonical embedding $J_I(Q,W)\emb J(Q,W)$ corresponds to the projection $K_{\bP^2}\to\bP^2$ and the forgetful map $J(Q,W)\to J_I(Q,W)$ corresponds to the zero section $\bP^2\emb K_{\bP^2}$.
We have a commutative diagram of derived categories
\begin{ctikzcd}
D^b(\coh\bP^2)\rar["\sim"]\dar[hook]&D^b(\mmod J_I)\dar[hook]\\
D^b_c(\coh K_{\bP^2})\rar["\sim"]& D^b(\mmod J)
\end{ctikzcd}
where $D^b_c(\coh K_{\bP^2})$ denotes the bounded derived category of coherent sheaves with compact support on $K_{\bP^2}$.
Therefore, counting objects on $\bP^2$ corresponds to counting objects on the CY3-fold $K_{\bP^2}$ (supported on the zero section) or counting objects in $D^b(\mmod J)$.

\medskip

As was discussed in the previous section, we have an explicit formula for the generating series $\hOm^\bn(x)$ (and $\hOm(x)$) that can be determined in small degrees using a computer.
Then we can apply the Joyce-Reineke formula (see Theorem \ref{JR}) to find the attractor invariants $\Om^\bn_*(d,y)$ and $\Om_*(d,y)$.
The following conjecture was obtained using this approach and was verified up to degree $(4,4,4)$.

\begin{conjecture}
We have
\begin{equation}
\Om_*(e_i,y)=1,\qquad 
\Om_*(n\de,y)=-y\inv(y^4+y^2+1),\qquad n\ge1.
\end{equation}
All other attractor invariants vanish.
\end{conjecture}
Note that the above vanishing of the attractor invariants was conjectured in \cite{Beaujard:2020sgs}.
The value of $\Om_*(n\de,y)$, corresponding to $n$ D0-branes on $K_{\bP_2}$, was left undetermined
in this reference, but according to Remark 
\ref{remarkBBS} we have 
$\Om_*(n\de,y)=(-y)^{-3}[\tcX]$.
For $\wtl\cX=K_{\bP^2}$,
we have $[\widetilde\cX]=q(q^2+q+1)$ and the corresponding value of $\Om_*(n\de,y)$ is consistent with our computations. It was conjectured in  \cite{Beaujard:2020sgs} that the single-centered invariants
(see Remark \ref{remarkOmS}) satisfy $\Om_*(n\de,y)=\OmS(n\de, y)-y-1/y$, so 
Eq.~\eqref{OmSgen}
appears to be verified in this case.


\medskip

Using the prescriptions in \S\ref{sec:framed-unframed2}, or equivalently the attractor tree formula
for the framed quiver, we can compute the framed refined DT invariants in the non-commutative
chamber,
\bea
Z_{0,\NCDT}(x)&=&
 1+x_0 + \left(y^2+1+1/y^2\right)  x_0 x_1  + \left(y^2+1+1/y^2\right) x_0 x_1^2  
- \left(y^3+y+1/y \right) x_0 x_1 x_2  \nn\\
&&+ \left( y^4+2y^2+3+2/y^2+1/y^4\right)  x_0 x_1^2 x_2 + x_0 x_1^3 
- \left(y^3+y+1/y \right)x_0^2  x_1 x_2 \nn\\
&&
\hspace*{-5mm} 
-\left( y^5+y^3+y+1/y+1/y^3+1/y^5 \right)  x_0 x_1^3 x_2 
+\left( y^4+2y^2+3+2/y^2+1/y^4\right)  x_0 x_1^2 x_2^2 \nn\\&& 
-\left( y^5 +2y^3+3y+2/y+1/y^3 \right) (x_0^2 x_1^2 x_2+ x_0^2 x_1^3 x_2 )
+ \left(y^2+1+1/y^2\right)  x_0 x_1^2 x_2^3 
\nn\\&&
+\left( y^8+y^6+2y^4+2y^2+3+2/y^2+2/y^4+1/y^6+1/y^8\right) x_0 x_1^3 x_2^2  \nn\\&& 
+\left(y^6+3y^4+6y^2+6+ 4/y^2+1/y^4 \right)
 x_0^2 x_1^2 x_2^2  + \cO(x_i^7)
\eea
In the unrefined limit $y\to 1$, this agrees with the counting of molten crystals in
\cite[(8.16)]{Cirafici:2010bd}, see also  \cite[Remark A.5]{Young:2008hn}.
It is worth noting that the lack of invariance under $y\to 1/y$
in most of the coefficients in this expansion follows directly from the non-invariance of $\Om_*(n\de,y)$.

\medskip

For future reference, we record the unframed stacky invariants for trivial stability condition for small dimensions (multiplied by the motive $[G_d]= \prod_{i\in Q_0} [GL(d_i)]$, see \eqref{PGLn}), which provide a convenient starting point for applying the Joyce-Reineke formula:
$$
\begin{array}{|c|l|}
\hline
d & \cA(d,y)\cdot  [G_d] \\ \hline
(1,1,1) &
{y}^{2} \left({y}^{6}+ {y}^{4}-1 \right) 
\\
(2,1,1) &
- {y}^{7}\left( {y}^{8}+{y}^{6} -1\right) 
\\
(2,2,1) &
- {y}^{11}\left( 3\,{y}^{10}+{y}^{8}-{y}^{6}-2\,{y}^{4}-{y}^{2}+1 \right) 
\\
(2,2,2) &
{y}^{8}\left( {y}^{20}+5\,{y}^{18}+8\,{y}^{16}-7\,{y}^{14}-13\,{y}^{12}-3\,{y}^{10}+8\,{y}^{8}+7\,{y}^{6}-4\,{y}^{4}-2\,{y}^{2}+1 \right) 
\\
\hline
\end{array}
$$
Stacky invariants for dimension vectors with support on 1 or 2 vertices are easily computed
since the relations $\partial W/\partial a=0$ are trivial obeyed. 

\medskip

\subsection{Invariants of \tpdf{$\IF_0$}{F0}}
\label{secF0}
We now consider  the Hirzebruch surface $\IF_0=\bP^1\xx\bP^1$ and the corresponding CY3-fold $\widetilde\cX=K_{\IF_0}$.
It can be realized as a crepant resolution of a singular toric CY3 variety $\cX$ associated to the brane tiling
shown in Figure \ref{quivF0} (for the so-called phase II in the terminology of \cite{Beaujard:2020sgs}),
corresponding to a $\IZ_2$-orbifold of the conifold.
The quiver has vertices indexed by $i\in \bZ_4$, arrows $a^{(1)}_i,a^{(2)}_i:i\to i+1$ 
and potential
\begin{equation}
W=\sum_{i\ne j,k\ne l}
\sgn(i,j)\, \sgn(k,l)\,  a_3^{(l)}a_2^{(j)}a_1^{(k)}a_0^{(i)}.
\end{equation}
where $\sgn(1,2)=1, \sgn(2,1)=-1$.  We write $a_i=a^{(1)}_i$ and $b_i=a^{(2)}_i$.

\begin{figure}[ht]
\begin{tikzcd}[sep=2cm]
0\rar[->>]&1\dar[->>]\\
3\uar[->>]&2\lar[->>]
\end{tikzcd}
\hspace*{2cm}
\raisebox{-2cm}{\includegraphics[height=5cm]{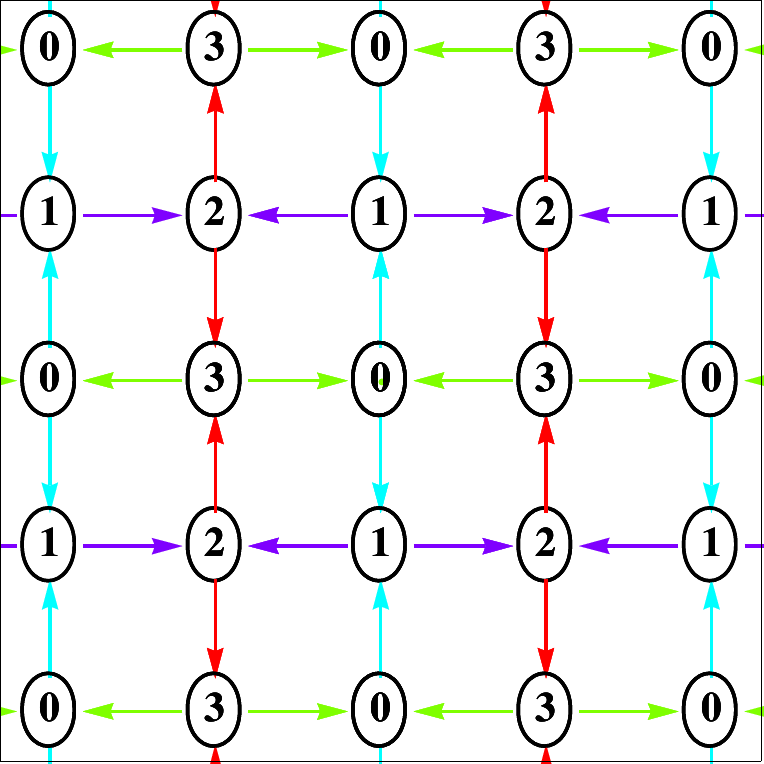}}
\caption{Quiver and tiling for $\IF_0$}
\label{quivF0}
\end{figure}

\medskip

Let us consider two disjoint cuts
\begin{equation}
I=\set{b_1,b_3},\qquad I'=\set{b_0,b_2}.
\end{equation}
As before, we define
\begin{equation}
J_I=\bk Q'/(\dd W/\dd a\col a\in I),\qquad
Q'=Q\ms I
\end{equation}
and $Q''=Q\ms(I\cup I')=C_4$, the cyclic quiver.
Then an element of a fiber of a forgetful map 
\begin{equation}
\pi:R(J_I,d)\to R(Q'',d)
\end{equation}
over a representation $M$ of $Q''=C_4$
corresponds to a choice of morphisms for the arrows
$b_0,b_2\in I'$.
The relations in the algebra $J_I$ amount to the
commutativity of the diagram 
\begin{ctikzcd}
M_0\rar[shift left,"a_1a_0"]\dar["b_0"']
&M_2\dar["b_2"]\lar[shift left,"a_3a_2"]\\
M_1\rar[shift left,"a_2a_1"]
&M_3\lar[shift left,"a_0a_3"]
\end{ctikzcd}
Therefore we can consider $(b_0,b_2)$ as a morphism
$F_{1,3}M\to F_{0,2}M$ between representations of the cyclic quiver $C_2$, where the forgetful functor $F_I$ was defined in \S\ref{forgetful}.
This means that the dimension of the fiber $\pi\inv(M)$ is equal to 
\begin{equation}
\vi(M)=h(F_{1,3}M,F_{0,2}M).
\end{equation}
As in \S\ref{double reduction}, let $\cI$ be the set of indecomposable representations of $Q''=C_4$ and let
$\si:\cI\xx\cI\to\bZ$ be the interaction form defined by \eqref{sigma}
\begin{equation}
\si(M,N)=2h(M,N)-2h(F_{1,3}M,F_{0,2}N)-\rho(M,N),
\end{equation}
\begin{equation}
\rho(d,d)=\hi_Q(d,d)
+2\sum_{(a:i\to j)\in I}d_id_j
=(d_0-d_1)^2+(d_2-d_3)^2.
\end{equation}
As was explained \S\ref{sec:general action} we have a simple parametrization of indecomposable representations of the cyclic quiver $C_4$.
Therefore we can compute the generating function $\hOm(x)$ of unframed stacky invariants of the Jacobian algebra $J(Q,W)$ \eqref{part function}
\begin{equation}
\hOm(x)
=\sum_{m:\cI\to\bN}
\frac{(-q^\oh)^{-\sum_{M,N\in\cI}m_M m_N\si(M,N)}}
{\prod_{M\in\cI}(q\inv)_{m_M}}
x^{\sum_{M\in\cI} m_M\udim M}.
\end{equation}

The following conjecture was verified in small degrees:
\begin{conjecture}
We have
\begin{equation}
\Om_*(e_i)=1,\qquad \Om_*(n\de)=-y\inv(y^2+1)^2,\qquad n\ge1,
\end{equation}
All other attractor invariants vanish.
\end{conjecture}
Note that this conjecture is compatible with 
$\Om_*(n\de)=(-y)^{-3}[\wtl\cX]$ \eqref{n-delta}
as $[\wtl\cX]=[K_{\IF_0}]=q(q+1)^2$. Moreover, it was conjectured in  \cite{Beaujard:2020sgs} that the single-centered invariants
(see Remark \ref{remarkOmS}) satisfy $\Om_*(n\de,y)=\OmS(n\de, y)-y-1/y$, so Eq.~\eqref{OmSgen}
is again verified. We note that DT invariants on $K_{\IF_0}$ were studied using exponential networks  in \cite{Banerjee:2020moh}. 

\medskip

Using the prescriptions in \S\ref{sec:framed-unframed2}, or equivalently the attractor tree formula
for the framed quiver, we can compute the framed, refined DT invariants in the non-commutative
chamber 
\bea
Z_{1,\NCDT}(x)&=&
1+x_0-(y+1/y) x_0 x_1+ x_0 x_1^2+ \left(y^2+2+1/y^2\right)  x_0 x_1 x_2
\nn\\&&  
 - \left(y^3+y+1/y+1/y^3\right) x_0  x_1^2   x_2
   -(y+1/y)  x_0 x_2^2 x_1-\left(y^3+2y+1/y\right)  x_0 x_1 x_2 x_3
     \nn\\&&
   -\left(y^3+2y+y^{-1}\right) x_0^2 x_1 x_2 x_3
    +  \left(y^4+y^2+2+1/y^2+1/y^4\right)  x_0 x_1^2  x_2^2 
     \nn\\&&
    +\left(y^2+2+1/y^2\right)  \left( x_0 x_1 x_2^2 x_3 +  x_0 x_1^2 x_2 x_3 \right) 
     -\left(y^3+y+1/y+1/y^3\right)  x_0 x_1^2 x_2^3
     \nn\\&&
     -(y+1/y) x_0   x_1 x_2^2 x_3^2 
     -\left(y^5+2y^3+4y+4/y+2/y^3+1/y^5\right) x_0 x_1^2 x_2^2 x_3 
     \nn\\&&
     +(y^2+2+1/y^2)   x_0^2 x_1 x_2^2 x_3 
     +  \left( y^4+3y^2+3+1/y^2\right) x_0^2 x_1^2 x_2 x_3 + \cO(x_i^7) 
\eea
Again, one may check that this agrees with the counting of molten crystals 
in the unrefined limit $y\to 1$.

\medskip

For future reference, we also record the unframed stacky invariants for trivial stability condition:
$$
\begin{array}{|c|l|}
\hline
d & \cA(d,y)\cdot [GL_d] \\ \hline
(1,1,1,1) & 
{y}^{2} \left({y}^{8}+ 2\,{y}^{6}-2\,{y}^{4}-{y}^{2}+1 \right)
\\
(1,1,1,2) &
- {y}^{7}\left(3\,{y}^{8}  -3\,{y}^{4}+1\right)  
\\
(1,1,2,2) & 
{y}^{10}\left({y}^{14}+{y}^{12}+3\,{y}^{10}-3\,{y}^{8}-5\,{y}^{6}+3\,{y}^{4}+2\,{y}^{2}-1 \right) 
\\
(1,2,2,2) &
-{y}^{9}\left(3\,{y}^{20}+4\,{y}^{18}-{y}^{16}-10\,{y}^{14}-{y}^{12}+6\,{y}^{10}+2\,{y}^{6}-{y}^{4}-2\,{y}^{2}+1 \right) 
\\
(2,2,2,2) & 
{y}^{10} \left({y}^{26}+7\,{y}^{24}+19\,{y}^{22}-15\,{y}^{20}-53\,{y}^{18}+3\,{y}^{16}+66\,{y}^{14}+10\,{y}^{12} \right. \\
& \left. -49\,{y}^{10}-3\,{y}^{8}+20\,{y}^{6}-3\,{y}^{4}-3\,{y}^{2}+1 \right) 
\\
\hline
\end{array}
$$
\medskip

\subsection{Invariants of \tpdf{$\bF_1=\dP_1$}{F1=dP1}}
\label{secdP1}
We now consider del Pezzo surface $\dP_1$, which is the blow-up of $\bP^2$ at one point.
Equivalently, it is the Hirzebruch surface $\bF_1=P(\cO_{\bP^1}\oplus\cO_{\bP^1}(1))$.
Its canonical bundle $K_{\dP_1}$ is a crepant resolution of a singular toric CY3 variety associated to 
the brane tiling shown in Figure \ref{quivF1}. The quiver $Q$ has 4 vertices, 10 arrows
and potential \cite{Feng:2001xr,Aspinwall:2005ur}
\be
W=
 \Phi_{41}^{1} \Phi_{34}^{2} \Phi_{23}^{2} \Phi_{12}
+\Phi_{41}^{2}\Phi_{34}^{1}   \Phi_{13} 
+\Phi_{42} \Phi_{34}^{3}  \Phi_{23}^{1} 
- \Phi_{41}^{2} \Phi_{34}^{2} \Phi_{23}^{1}  \Phi_{12}
-  \Phi_{41}^{1}  \Phi_{34}^{3} \Phi_{13}  
- \Phi_{42} \Phi_{34}^{1}   \Phi_{23}^{2}
\ee

\begin{figure}[ht]
\begin{tikzcd}[sep=2cm]
1\rar\drar&2\dar[->>]\\
4\uar[->>]\urar&3\lar[->>>]
\end{tikzcd}
\hspace*{2cm}
\raisebox{-2cm}{\includegraphics[height=5cm]{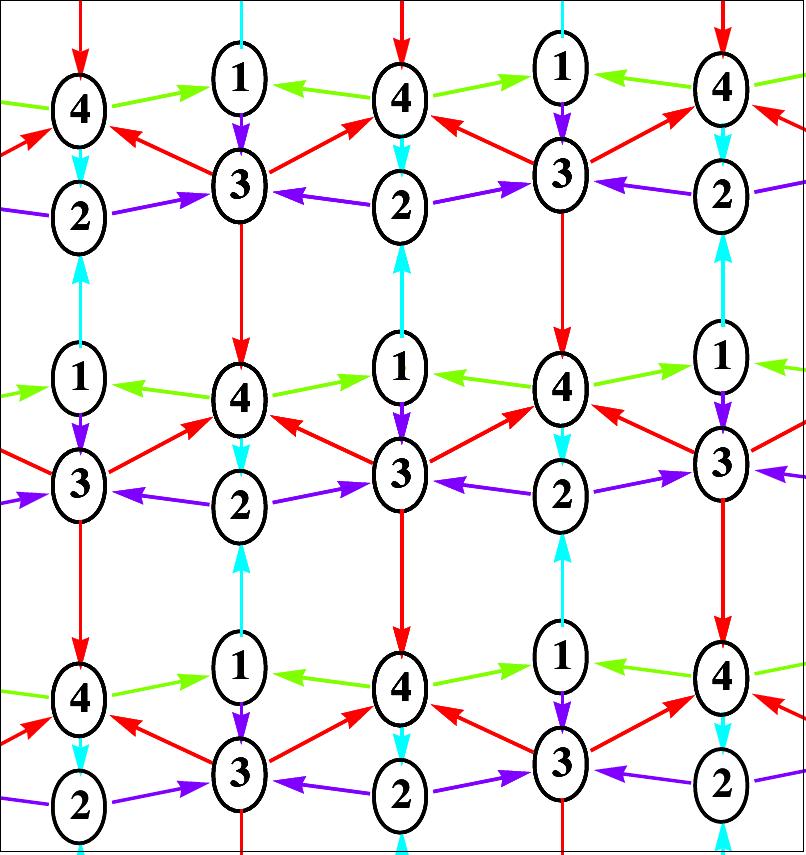}}
\caption{Quiver and tiling for $\IF_1$}
\label{quivF1}
\end{figure}

We choose the cuts
\begin{equation}
I=\set{\Phi_{23}^1,\Phi_{34}^1,\Phi_{41}^1},\qquad
I'=\set{\Phi_{13},\Phi_{34}^2,\Phi_{42}}.
\end{equation}
Then 
\begin{equation}
Q''=Q\ms(I\cup I')=\set{\Phi_{12},\Phi_{23}^2,\Phi_{34}^3,\Phi_{41}^2}
=\set{a_1,a_2,a_3,a_4}
\end{equation}
is a cyclic quiver with $4$ vertices.
An element of the fiber of
\begin{equation}
\pi:R(J_I,d)\to R(Q'',d)
\end{equation}
over a $Q''$-representation $M$ is encoded by the values of the arrows in $I''$ and induces a commutative diagram
\begin{ctikzcd}
1\dar\rar["a_2a_1"]&3\dar\rar["a_3"]
&4\dar\ar[ll,bend right,"a_4"']\\
3\rar["a_3"]&4\rar["a_1a_4"]&2\ar[ll,bend left,"a_2"]
\end{ctikzcd}

The arrows at the top produce a representation $F_2 M$ of $C_3$ (we forget the vertex $2$).
The arrows at the bottom produce a representation $\Si F_1M$ (we forget the vertex $1$ and perform the cyclic shift).
The above analysis implies that the fiber of $\pi$ can be identified with
\begin{equation}
\pi\inv(M)\iso\Hom_{C_3}(F_{2}M,\Si F_{1}M).
\end{equation}
Using notation from \S\ref{double reduction}, we obtain
\begin{equation}
\si(M,N)=2h(M,N)-2h(F_{2}M,\Si F_{1}N)-\rho(M,N),
\end{equation}
\begin{equation}
\rho(d,e)=\hi_Q(d,e)+2\sum_{(a:i\to j)\in I}d_ie_j
=\sum_i d_ie_i-d_1(e_2+e_3)-d_3e_4-d_4e_2.
\end{equation}
As was discussed in \S\ref{double reduction}, this gives us all ingredients to compute the generating function $\cA(x)$ of unframed stacky invariants of the Jacobian algebra $J(Q,W)$.
Then we apply the Joyce-Reineke formula (see Theorem \ref{JR}) to find the attractor invariants $\Om_*(d,y)$.
The following conjecture was verified in small degrees.

\begin{conjecture}
We have
\begin{equation}
\Om_*(e_i)=1,\qquad
\Om_*(n\de)=-y\inv(y^4+2y^2+1),\qquad n\ge1.
\end{equation}
All other attractor invariants vanish.
\end{conjecture}
Note that this conjecture is compatible with 
$\Om_*(n\de)=(-y)^{-3}[K_{\dP_1}]$ 
\eqref{n-delta}
as $[K_{\dP_1}]=q(q+1)^2$.
Moreover it was conjectured in 
 \cite{Beaujard:2020sgs} that  $\Om_*(n\de,y)=\OmS(n\de, y)-y-1/y$, in line with Eq.~\eqref{OmSgen}.

\medskip

Using the prescriptions in \S\ref{sec:framed-unframed2}, or equivalently the attractor tree formula
for the framed quiver, we can compute the framed, refined DT invariants in the non-commutative
chamber 
\bea
Z_{1,\NCDT}&=&
1+x_1+ x_1 x_2+x_1 x_3 + \left(y^2+1+1/y^2\right) x_1 x_2 x_3+x_1 x_3 x_4
\nn\\&&
   -\left(y+1/y\right) x_1^2 x_3 x_4 +\left(y^2+1+1/y^2\right) x_1 x_2 x_3^2 
 -(y^3+2y+1/y) x_1 x_2 x_3 x_4 +\dots
   \nn\\
Z_{2,\NCDT}&=&
1+x_2-2 x_2 x_3+  x_2 x_3^2+\left(y^2+2+1/y^2\right) x_2 x_3 x_4
- \left(y+1/y\right) \left( x_3  x_4 x_2^2+ x_2 x_3 x_4^2 \right)\nn\\&&
   +\left(y^4+y^2+1+1/y^2+1/y^4\right) x_2 x_3^2 x_4 
   -(y^3+2y+1/y) x_1 x_2 x_3 x_4+\dots
   \nn\\
Z_{3,\NCDT}&=&
1+x_3+\left(y^2+1+1/y^2\right)  x_3 x_4+ \left(y^2+1+1/y^2\right) x_3 x_4^2
+\left(y^2+2+1/y^2\right) x_1 x_3 x_4\nn\\&& +x_2 x_3 x_4
+ x_3
   x_4^3+\left( y^4+2y^2+3+2/y^2+1/y^4\right)
   x_1 x_3 x_4^2-(y+1/y) x_2 x_3 x_4^2\nn\\&&
   +x_1^2 x_3 x_4-(y^3+2y+1/y)  x_1 x_2 x_3 x_4 +\dots
\nn\\
Z_{4,\NCDT}&=&
1+x_4+ x_2 x_4-(y+1/y) x_1 x_4+ x_1^2 x_4 + \left( y^2+2+1/y^2 \right) x_1 x_2 x_4\nn\\&&
+   \left(y^2+1+1/y^2\right) x_1^2 x_2 x_4 
   -\left(y+1/y\right) x_1 x_2^2 x_4 -(y^3+2y+1/y)x_1 x_2 x_3 x_4 +\dots 
\eea

For future reference, we also record the unframed stacky invariants for trivial stability condition:
$$
\begin{array}{|c|l|}
\hline
d & \cA(d,y)\cdot [G_d] \\ \hline
(1,0,1,1) & 
-{y}^{3} \left( -y^4 - y^2 + 1 \right)
\\
(0,1,1,1) & 
{y}^{2} \left( -y^4 - y^2 + 1 \right)
\\
(1,1,1,1) & 
{y}^{2} \left(y^8 + 2 y^6 - y^4 - 2 y^2 + 1
 \right)
\\
\hline
\end{array}
$$

\medskip

\subsection{Invariants of \tpdf{$\dP_2$}{dP2}}
\label{secdP2}
Consider del Pezzo surface $\dP_2$ which is the blow-up of $\bP^2$ at two points,
or equivalently the blow-up of $\IF_1$ at one point.
Its canonical bundle $K_{\dP_2}$ is a crepant resolution of a singular toric CY3 variety corresponding to the brane tiling shown in Figure \ref{quivdP2}  (in the so called phase II). The quiver $Q$ has 5 vertices, 11 arrows 
and potential 
\bea
W &=&
\Phi_{51}  \Phi_{35} \Phi_{23}^1 \Phi_{12}^2 
+\Phi_{41} \Phi_{34} \Phi_{23}^2   \Phi_{12}^1  
+\Phi_{52} \Phi_{45}  \Phi_{24} 
\nn\\&&
- \Phi_{41} \Phi_{24}   \Phi_{12}^2 
-  \Phi_{52}\Phi_{35}  \Phi_{23}^2 
-  \Phi_{51} \Phi_{45}  \Phi_{34} \Phi_{23}^1   \Phi_{12}^1 
\eea

\begin{figure}
\begin{tikzcd}[column sep=.3cm]
&&2\ar[rrd,->>]\ar[ddr]\\
1\ar[rru,->>]&&&&3\ar[ld],\ar[llld]\\
&5\ar[ul]\ar[uur]&&4\ar[ll]\ar[lllu]
\end{tikzcd}
\hspace*{2cm}
\raisebox{-2cm}{\includegraphics[height=5cm]{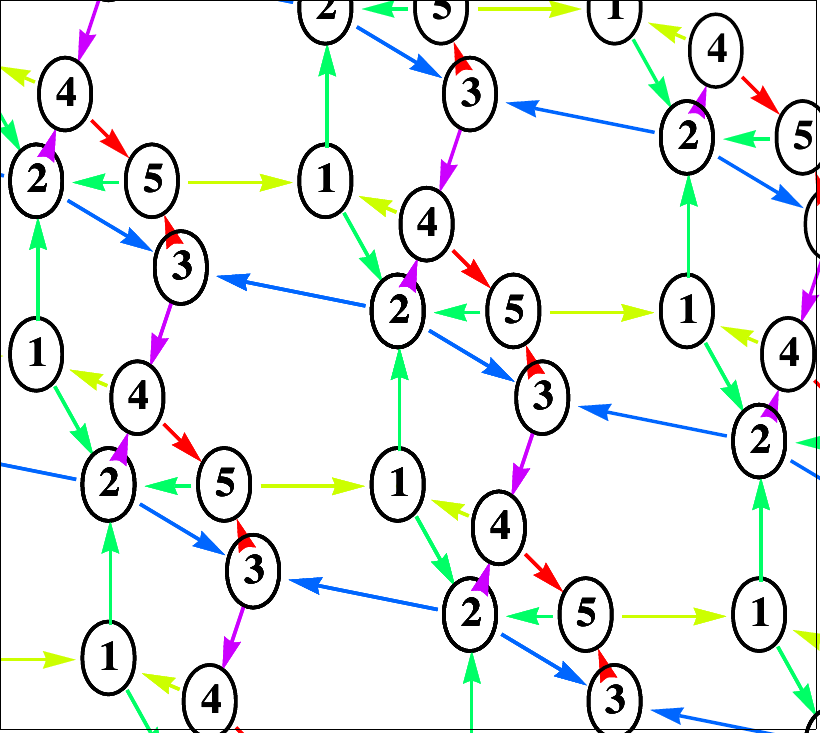}}
\caption{Quiver and tiling for $\dP_2$}
\label{quivdP2}
\end{figure}

We choose the disjoint cuts
\begin{equation}
I=\set{\Phi_{23}^1,\Phi_{41}, \Phi_{52}},\qquad
I'=\set{\Phi_{12}^1,\Phi_{24},\Phi_{35}}.
\end{equation}
Then 
\begin{equation}
Q''=Q\ms(I\cup I')=\set{\Phi_{12}^2,\Phi_{23}^1,\Phi_{34},\Phi_{45},
\Phi_{51}}
=\set{a_1,a_2,a_3,a_4,a_5}
\end{equation}
is a cyclic quiver with $5$ vertices.
The same analysis as before shows that an element of the fiber $\pi\inv(M)$, for $M\in\Rep Q''$, induces a commutative diagram
\begin{ctikzcd}
1\dar\rar["a_1"]&2\dar\rar["a_2"]
&3\dar\ar[ll,bend right,"a_5a_4a_3"']\\
2\rar["a_3a_2"]&4\rar["a_4"]&5\ar[ll,bend left,"a_1a_5"]
\end{ctikzcd}
The arrows at the top produce representation $F_{4,5}M\in\Rep C_3$.
The arrows at the bottom produce representation $F_{1,3}M\in\Rep C_3$.
The above analysis implies that the fiber of $\pi$ can be identified with
\begin{equation}
\pi\inv(M)\iso\Hom_{C_3}(F_{4,5}M,F_{1,3}M).
\end{equation}
Using notation from \S\ref{double reduction}, we obtain
\begin{equation}
\si(M,N)=2h(M,N)-2h(F_{4,5}M,F_{1,3}N)-\rho(M,N),
\end{equation}
As was discussed in \S\ref{double reduction}, this gives us all ingredients to compute the generating function $\cA(x)$ of unframed stacky invariants of the Jacobian algebra $J(Q,W)$.
Then we apply the Joyce-Reineke formula (see Theorem \ref{JR}) to find the attractor invariants $\Om_*(d,y)$.
The following conjecture was verified in small degrees.

\begin{conjecture}
We have
\begin{equation}
\Om_*(e_i)=1,\qquad
\Om_*(n\de)=-y\inv(y^4+3y^2+1),\qquad n\ge1.
\end{equation}
All other attractor invariants vanish.
\end{conjecture}
Note that this conjecture is compatible with 
$\Om_*(n\de)=(-y)^{-3}[K_{\dP_2}]$ 
\eqref{n-delta}
as 
$[K_{\dP_2}]=q[\dP_2]=q([\bP^2]+2q)=q(q^2+3q+1)$. Moreover 
 $\Om_*(n\de,y)=\OmS(n\de, y)-y-1/y$  \cite{Beaujard:2020sgs}.

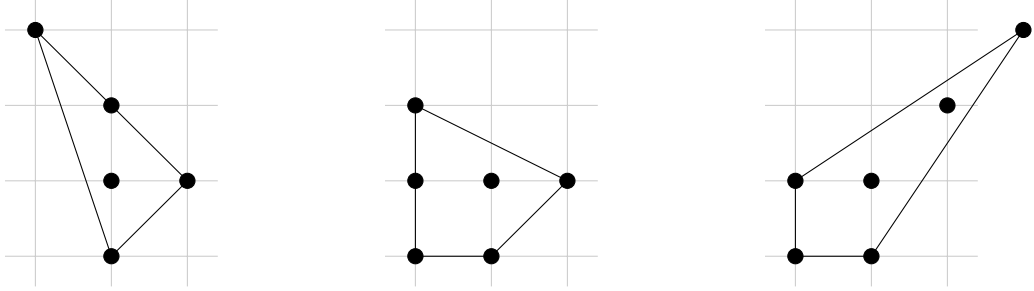
\begin{figure}[ht]
\begin{center}
\begin{tikzpicture}[inner sep=2mm,scale=1]
\begin{scope}[shift={(-5,0)}]  
\draw[step=1cm,black!20,thin] (-1.4,-1.4) grid (1.4,2.4); 
\draw (0,1) -- (-1,2) -- (0,-1) -- (1,0) -- (0,1);
\filldraw [black]  (0,0) ellipse (0.1 and 0.1);
\filldraw [black]  (1,0) ellipse (0.1 and 0.1);
\filldraw [black]  (0,1) ellipse (0.1 and 0.1);
\filldraw [black]  (-1,2) ellipse (0.1 and 0.1);
\filldraw [black]  (0,-1) ellipse (0.1 and 0.1);
\end{scope}
\begin{scope}[shift={(0,0)}]  
\draw[step=1cm,black!20,thin] (-1.4,-1.4) grid (1.4,2.4); 
\draw (1, 0) -- (-1, 1) -- (-1, 0) -- (-1, -1) -- (0, -1) -- (1,0);
\filldraw [black]  (0,0) ellipse (0.1 and 0.1);
\filldraw [black]  (1,0) ellipse (0.1 and 0.1);
\filldraw [black]  (-1,1) ellipse (0.1 and 0.1);
\filldraw [black]  (-1,0) ellipse (0.1 and 0.1);
\filldraw [black]  (-1,-1) ellipse (0.1 and 0.1);
\filldraw [black]  (0,-1) ellipse (0.1 and 0.1);
\end{scope}
\begin{scope}[shift={(5,0)}]  
\draw[step=1cm,black!20,thin] (-1.4,-1.4) grid (1.4,2.4); 
\draw (-1,0) -- (-1,-1) -- (0,-1) -- (2,2) -- (-1,0);
\filldraw [black]  (0,0) ellipse (0.1 and 0.1);
\filldraw [black]  (-1,0) ellipse (0.1 and 0.1);
\filldraw [black]  (-1,-1) ellipse (0.1 and 0.1);
\filldraw [black]  (0,-1) ellipse (0.1 and 0.1);
\filldraw [black]  (1,1) ellipse (0.1 and 0.1);
\filldraw [black]  (2,2) ellipse (0.1 and 0.1);
\end{scope}
\end{tikzpicture}
\end{center}
\caption{Toric diagram of $\IF_2$ (left), $\PdP_2$ (center) and $Y^{3,2}$ (right) \label{toricF2}}
\end{figure}

For future reference, we also record the unframed stacky invariants for trivial stability condition:
$$
\begin{array}{|c|l|}
\hline
d & \cA(d,y)\cdot [G_d] \\ \hline
(1,1,0,1,0) & 
-{y}^{3} \left(2y^{2}  - 1 \right)
\\
(0,1,1,0,1) & 
-{y}^{3} \left(2y^{2}  - 1 \right)
\\
(0,1,0,1,1) & 
{y}^{2} \left(2y^{2}  - 1 \right)
\\
(1,1,1,0,1) & 
-{y}^{3} \left(y^{6} + 2 y^4 + 3 y^2 - 1 \right)
\\
(1,1,1,1,0) & 
-{y}^{3} \left(y^{6} + 2 y^4 + 3 y^2 - 1 \right)
\\
(1,1,1,1,1) & 
{y}^{2} \left(y^{10} + 3 y^8 - 2 y^6 - 4 y^4 + 4 y^2 - 1 \right)
\\
\hline
\end{array}
$$

\medskip

\begin{remark}
For $\dP_3$, none of the four brane tilings corresponding to models I through IV in \cite{Beaujard:2020sgs} admit a double cut reduction to an oriented cyclic quiver $C_6$. Nonetheless, they can
be reduced to a cyclic quiver with some flipped arrows, whose representation theory is also
under control but more  complicated.
\end{remark}

\subsection{Invariants of \tpdf{$\IF_2$}{F2}}
\label{secF2}
Consider the Hirzebruch surface
$\IF_2=P(\cO_{\bP^1}\oplus\cO_{\bP^1}(2))$, whose toric diagram is shown on Figure \ref{toricF2}.
Its canonical bundle $\widetilde\cX=K_{\IF_2}$ is a crepant resolution
of the quotient singularity $\cX=\bC^3/\bZ_4$, where the action of $\bZ_4$ on $\bC^3$ is given by
\begin{equation}
1\mto\diag(\om,\om,\om^2).\qquad \om=e^{2\pi\bi/4}=\bi.
\end{equation}
The corresponding tiling and McKay quiver is shown on Figure \ref{quivF2}, and 
 the potential is given by \eqref{potential for Z_N}.
The toric diagram (see Figure \ref{toricF2}) has one internal lattice point and one lattice point on the boundary, in agreement with the fact that $\IF_2$ is an almost Fano surface.

\begin{figure}[ht]
\begin{tikzcd}[sep=2cm]
0\rar[->>]\ar[dr,<->]&1\dar[->>]\ar[dl,<->]\\
3\uar[->>]&2\lar[->>]
\end{tikzcd}
\hspace*{2cm}
\raisebox{-2cm}{\includegraphics[height=5cm]{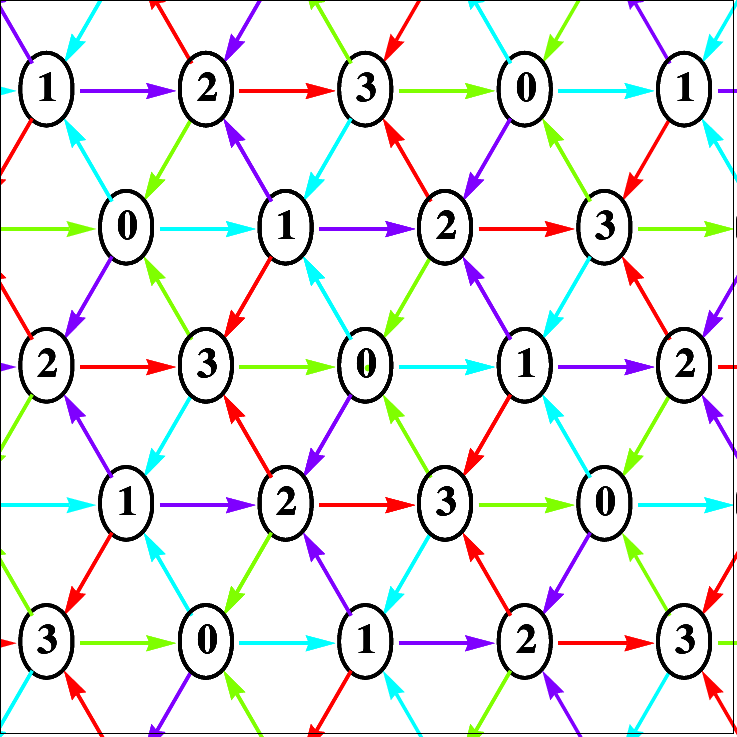}}
\caption{Quiver and tiling for $\IF_2$}
\label{quivF2}
\end{figure}

We apply the same approach as before to compute the attractor invariants of $J(Q,W)$.
The following conjecture was verified up to degrees $(3,3,3,3)$:

\begin{conjecture}
We have
\begin{gather}
\Om_*(e_i)=1,\qquad \Om_*(n\de)=-y\inv(y^2+1)^2,\qquad n\ge1, \nn\\
\Om_*(e_0+e_2+n\de)=\Om_*(e_1+e_3+n\de)=-y,\qquad n\ge0.
\end{gather}
All other attractor invariants vanish.
\end{conjecture}
Note that the vectors $e_i+e_{i+2}+n\de$ are contained in the kernel of the skew-symmetric form of the above quiver. The vanishing of the attractor invariants for dimension vectors  in the kernel of the 
skew-symetric form was conjectured in \cite{Beaujard:2020sgs}, but these invariants were
left undetermined. 
As in the previous section, the attractor invariant $\Om_*(n\de)$ counting D0-branes on $\wtl\cX$
is seen to be consistent with \eqref{n-delta},
upon computing the motive $[\wtl\cX]=q(q+1)^2$ (this follows from Lemma \ref{lm:Ptoric} or from the fact that $\wtl\cX=K_{\bF_2}$).
Moreover, we find $\Omstar(n\delta)=\OmS(n\delta)-y-1/y$, in agreement with Eq.~\eqref{OmSgen},
and $\Om_*(\gamma)=\OmS(\gamma)$ for $\gamma=e_0+e_2+n\de$ or $e_1+e_3+n\de$ 
with $n=1,2$.

\medskip

Using the prescriptions in \S\ref{sec:framed-unframed2}, or equivalently the attractor tree formula
for the framed quiver, we can compute the framed, refined DT invariants in the non-commutative
chamber:
\bea
Z_{0,\NCDT}(x)&=&1+x_0 -(y+1/y) x_0 x_1- y x_0 x_2 - y x_0^2 x_2 
+x_1^2 x_0+ (y^2+2+1/y^2)    x_0 x_1 x_2 \nn\\&&
   + y^2 x_0^2 x_2^2 +(1+y^2) x_0^2 x_1 x_2
   -(y+1/y) x_0 x_1 x_2^2 -(y^3+y+1/y+1/y^3) x_0 x_1^2 x_2 
   \nn\\&&
   -(y^3+2y+1/y) x_0 x_1 x_2 x_3 
   + y^2 x_0^3 x_2^2
    -(y^3+y+1/y+1/y^3)  x_0^2 x_1 x_2^2
   - y x_0^2 x_1^2 x_2
  \nn\\&&  -(y^3+2y+1/y) x_0^2 x_1 x_2 x_3
   +\left(y^4+y^2+2+1/y^2+1/y^4\right)  x_0 x_1^2 x_2^2 
 \nn\\&&  +\left(y^2+2+1/y^2\right)  \left( x_0 x_1 x_2^2 x_3 +  x_0 x_1^2 x_2 x_3 \right) 
   +\dots
\eea
One may check that this agrees with the counting of molten crystals 
in the unrefined limit $y\to 1$.

%

\subsection{Invariants of \tpdf{$\PdP_2$}{PdP2}}
The `pseudo del Pezzo surface' $\PdP_2$ is a  blow up of $\dP_1$ or $\IF_2$
at one point. Its
toric diagram, shown on Figure \ref{toricF2}, includes one lattice boundary point,
so $\PdP_2$ is an almost Fano surface.
The canonical bundle $K_{\PdP_2}$ is a crepant resolution of a singular toric CY3 variety corresponding to the brane tiling shown on Figure \ref{quivPdP2}.
The quiver $Q$ has 5 vertices, 13 arrows (including one bidirectional arrow)
and potential 
 \cite[(13.1)]{Hanany:2012hi} (relabelling nodes 1,2,3,4,5 into 3,1,4,5,2 in that reference)
\bea
W&=&
\Phi _{41}\Phi_{24}   \Phi _{12}^2 
+\Phi _{51}  \Phi _{35}  \Phi _{13}
+ \Phi _{31}  \Phi_{23}^2 \Phi _{12}^1
+\Phi _{52} \Phi_{45}  \Phi_{34}   \Phi _{23^1}\nn\\
&&
-\Phi _{41}  \Phi _{34}  \Phi _{13}
-\Phi _{31} \Phi _{23}^1 \Phi _{12}^2 
- \Phi _{52} \Phi _{35} \Phi _{23}^2 
- \Phi _{51} \Phi _{45}   \Phi _{24} \Phi _{12}^1 
\eea

\begin{figure}[ht]
\begin{tikzcd}[column sep=.3cm]
&&2\ar[rrd,->>]\ar[ddr]\\
1\ar[<->,rrrr]\ar[rru,->>]&&&&3\ar[ld],\ar[llld]\\
&5\ar[ul]\ar[uur]&&4\ar[ll]\ar[lllu]
\end{tikzcd}
\hspace*{2cm}
\raisebox{-2cm}{\includegraphics[height=5cm]{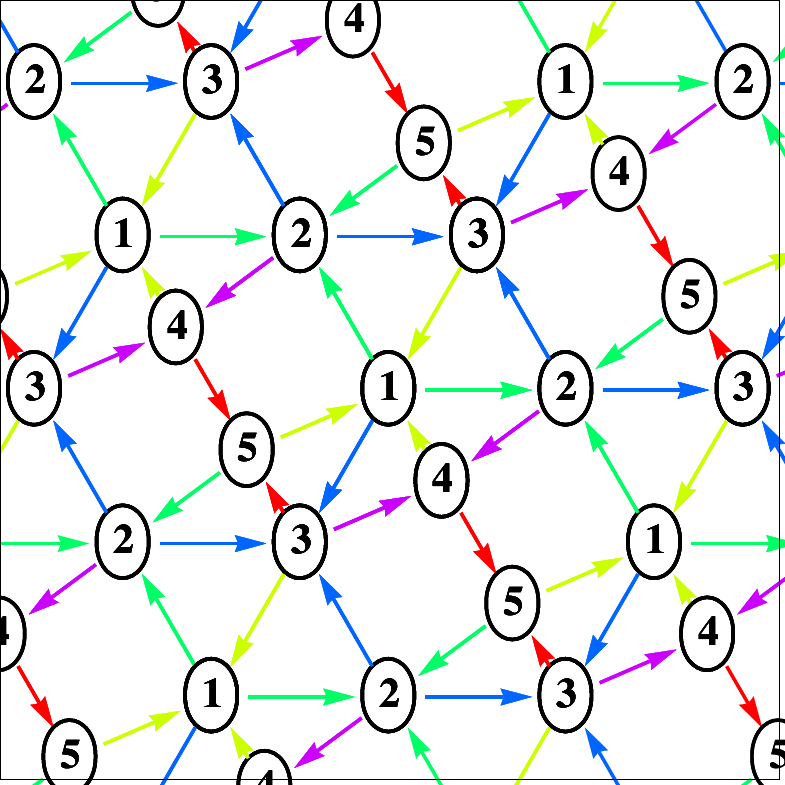}}
\caption{Quiver and tiling for $\PdP_2$}
\label{quivPdP2}
\end{figure}

We choose the disjoint cuts
\be
I=\set{\Phi_{13},\Phi_{24}, \Phi_{31},\Phi_{52}},\qquad
I'=\set{\Phi_{12}^1,\Phi_{23}^1,\Phi_{35},\Phi_{41}}
\ee
Then 
\begin{equation}
Q''=Q\ms(I\cup I')=\set{\Phi_{12}^2,\Phi_{23}^2,\Phi_{34},\Phi_{45},
\Phi_{51}}
=\set{a_1,a_2,a_3,a_4,a_5}
\end{equation}
is a cyclic quiver with $5$ vertices.
The same analysis as before shows that an element of the fiber $\pi\inv(M)$, for $M\in\Rep Q''$, induces a commutative diagram
\begin{ctikzcd}
1\dar\rar["a_1"]&2\dar\rar["a_2"]
&3\rar["a_3"]\dar&4\dar\ar[lll,bend right,"a_5a_4"']\\
2\rar["a_2"]&3\rar["a_4a_3"]&5\rar["a_5"]&1
\ar[lll,bend left,"a_1"']
\end{ctikzcd}

The arrows at the top produce representation $F_{5}M\in\Rep C_4$.
The arrows at the bottom produce representation $\Si F_{4}M\in\Rep C_4$.
The above analysis implies that the fiber of $\pi$ can be identified with
\begin{equation}
\pi\inv(M)\iso\Hom_{C_4}(F_{5}M,\Si F_{4}M).
\end{equation}
Using notation from \S\ref{double reduction}, we obtain
\begin{equation}
\si(M,N)=2h(M,N)-2h(F_{5}M,\Si F_{4}N)-\rho(M,N),
\end{equation}
As was discussed in \S\ref{double reduction}, this gives us all ingredients to compute the generating function $\cA(x)$ of unframed stacky invariants of the Jacobian algebra $J(Q,W)$.
Then we apply the Joyce-Reineke formula (see Theorem \ref{JR}) to find the attractor invariants $\Om_*(d,y)$.
The following conjecture was verified in small degrees.

\begin{conjecture}
We have
\bea
\Om_*(e_i) =1,\qquad
\Om_*(n\de)=-y\inv(y^4+3y^2+1),\qquad n\ge1 \nn\\
\Om_*(e_1+e_3+n\de)=\Om_*(e_2+e_4+e_5+n\de)=-y,\qquad n\ge0.
\eea
All other attractor invariants vanish.
\end{conjecture}
Note that this conjecture is compatible with
\eqref{n-delta}
as $[K_{\PdP_2}]=q(q^2+3q+1)$,
which is easily computed from the toric diagram in Figure \ref{toricF2}. Moreover we find
that $\Om_*(n\de)=\OmS(n\de)-y-1/y$ and $\Om_*(\gamma)=\OmS(\gamma)$ for 
$\gamma=e_1+e_3+n\de$ or $\gamma=e_2+e_4+e_5+n\de$ for low values of $n$.

\subsection{Invariants of \tpdf{$Y^{3,2}$}{Y32}}
The CY3 variety $Y^{3,2}$ is the simplest element in the $Y^{p,q}$  family of cones over smooth
 Sasaki-Einstein manifolds 
 constructed in \cite{Gauntlett:2004yd,Martelli:2004wu},
 after the conifold $Y^{1,0}$ and $Y^{2,1}=\dP_1$. Its
toric diagram, shown on Figure \ref{toricF2}, has two internal lattice points, but
no boundary lattice points except for the four corners.
 The brane tiling was determined in  \cite{Franco:2005rj}. The tiling
 and quiver are shown in Figure \ref{quivY32}.
The quiver has 6 vertices, 16 arrows and potential 
\bea
W&=&
\Phi_{42} \Phi_{34}^2  \Phi _{23}^1
+\Phi_{31} \Phi_{23}^2 \Phi_{12}^1
+\Phi_{56}^2 \Phi_{45}^1   \Phi_{64}
+ \Phi_{34}^1 \Phi_{53}  \Phi_{45}^2
+\Phi_{12}^2  \Phi_{61} \Phi_{56}^1 \Phi_{25} \nn\\
&&
- \Phi_{42}\Phi_{34}^1  \Phi _{23}^2
 - \Phi_{31}\Phi_{23}^1  \Phi_{12}^2
-\Phi_{56}^1 \Phi_{45}^2  \Phi_{64} 
- \Phi_{34}^2\Phi_{53}  \Phi_{45}^1
-\Phi_{12}^1\Phi_{61} \Phi_{56}^2 \Phi_{25}
\eea

\begin{figure}[ht]
\begin{tikzcd}[column sep=.7cm]
&1\dlar[->>]&6\lar\ar[dd]\\
2\drar[->>]\ar[rrr]&&&5\ular[->>]\ar[lld]\\
&3\rar[->>]\ar[uu]&4\urar[->>]\ar[llu]
\end{tikzcd}
\hspace*{2cm}
\raisebox{-2cm}{\includegraphics[height=5cm]{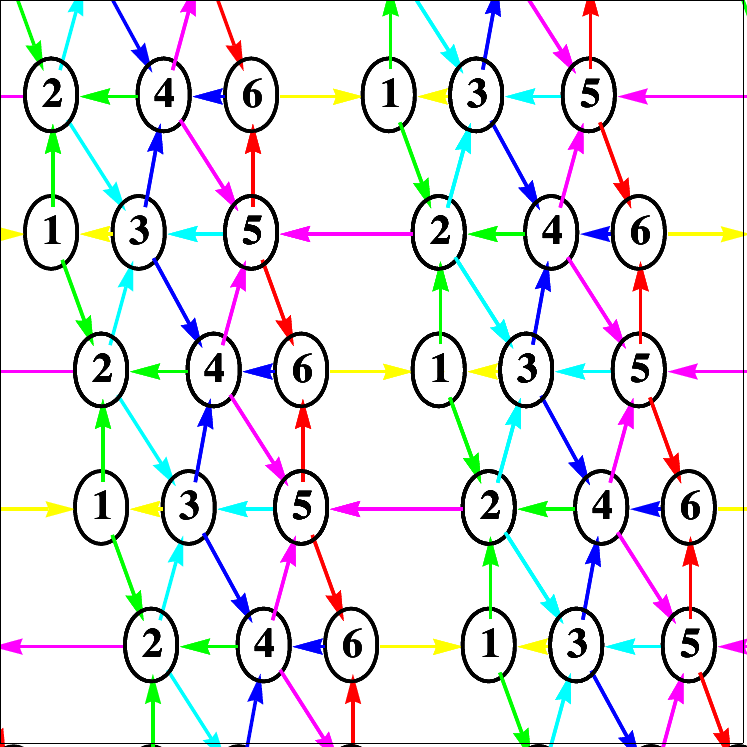}}
\caption{Quiver and tiling for $Y^{3,2}$}
\label{quivY32}
\end{figure}

We choose disjoint cuts
\begin{equation}
I=\set{\Phi_{64},\Phi_{42},\Phi_{25},\Phi_{53},\Phi_{31}}\ ,\quad 
I'=\set{\Phi_{12}^2,\Phi_{23}^2,\Phi_{34}^2,\Phi_{45}^2,\Phi_{56}^2}
\end{equation}
such that $I$ consists of all diagonals and $I'$ 
contains one arrow of each double arrow.
Then 
\begin{equation}
Q''=Q\ms(I\cup I')=\set{\Phi_{12}^1,\Phi_{23}^1,\Phi_{34}^1,\Phi_{45}^1,\Phi_{56}^1,\Phi_{61}}
=\set{a_1,a_2,a_3,a_4,a_5,a_6}
\end{equation}
is a cyclic quiver with $6$ vertices.
The same analysis as before shows that an element of the fiber $\pi\inv(M)$, for $M\in\Rep Q''$, induces a commutative diagram
\begin{ctikzcd}
1\rar\dar&2\rar\dar&3\rar\dar&4\dar\rar&5\dar
\ar[llll,bend right=20]\\
2\rar&3\rar&4\rar&5\rar&6
\ar[llll,bend left=20]
\end{ctikzcd}
The arrows at the top produce representation $F_{6}M\in\Rep C_5$.
The arrows at the bottom produce representation $F_{1}M\in\Rep C_5$.
The above analysis implies that the fiber of $\pi$ can be identified with
\begin{equation}
\pi\inv(M)\iso\Hom_{C_5}(F_{6}M,F_{1}M).
\end{equation}
Using notation from \S\ref{double reduction}, we obtain
\begin{equation}
\si(M,N)=2h(M,N)-2h(F_{6}M,F_{1}N)-\rho(M,N),
\end{equation}
As was discussed in \S\ref{double reduction}, this gives us all ingredients to compute the generating function $\cA(x)$ of unframed stacky invariants of the Jacobian algebra $J(Q,W)$.
Then we apply the Joyce-Reineke formula (see Theorem \ref{JR}) to find the attractor invariants $\Om_*(d,y)$.
The following conjecture was verified in small degrees.
\begin{conjecture}
We have
\begin{equation}
\Om_*(e_i)=1,\qquad
\Om_*(n\de)=-y\inv(y^4+3y^2+2),\qquad n\ge1.
\end{equation}
All other attractor invariants vanish.
\end{conjecture}
Note that this conjecture is compatible with 
$\Om_*(n\de)=(-y)^{-3}[Y^{3,2}]$
\eqref{n-delta}
as $[Y^{3,2}]= q(q^2+3q+2)$
which is easily computed from Figure \ref{toricF2}. Moreover we find that 
$\Om_*(n\de)=\Om_S(n\de)-2y-2/y$ for low values of $n$, in line with Eq.~\eqref{OmSgen}.

\subsection{Further orbifold examples}

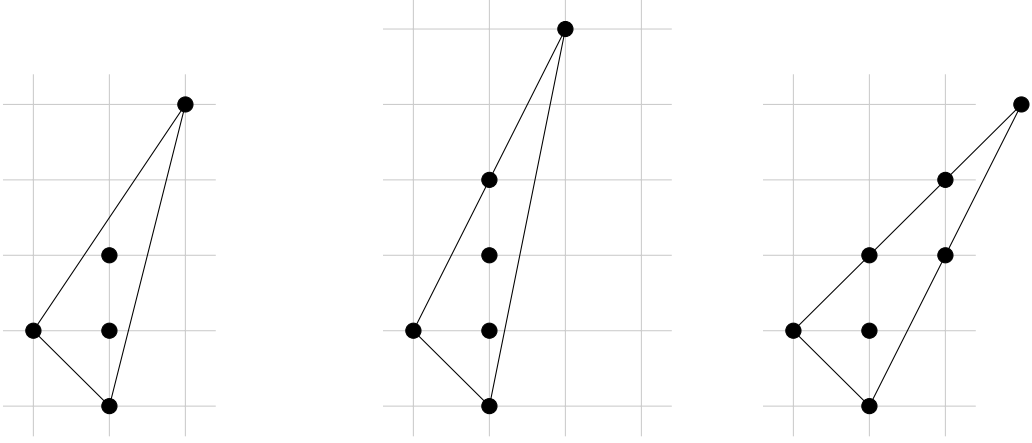
\begin{figure}[ht]
\begin{center}
\begin{tikzpicture}[inner sep=2mm,scale=1]
\begin{scope}[shift={(-5,0)}]  
\draw[step=1cm,black!20,thin] (-1.4,-1.4) grid (1.4,3.4); 
\draw (-1, 0) -- (0, -1) -- (1,3) -- (-1,0);
\filldraw [black]  (-1,0) ellipse (0.1 and 0.1);
\filldraw [black]  (0,-1) ellipse (0.1 and 0.1);
\filldraw [black]  (0,0) ellipse (0.1 and 0.1);
\filldraw [black]  (0,1) ellipse (0.1 and 0.1);
\filldraw [black]  (1,3) ellipse (0.1 and 0.1);
\end{scope}
\begin{scope}[shift={(0,0)}]  
\draw[step=1cm,black!20,thin] (-1.4,-1.4) grid (2.4,4.4); 
\draw (-1, 0) -- (0, -1) -- (1,4) -- (-1,0);
\filldraw [black]  (-1,0) ellipse (0.1 and 0.1);
\filldraw [black]  (0,-1) ellipse (0.1 and 0.1);
\filldraw [black]  (0,0) ellipse (0.1 and 0.1);
\filldraw [black]  (0,1) ellipse (0.1 and 0.1);
\filldraw [black]  (0,2) ellipse (0.1 and 0.1);
\filldraw [black]  (1,4) ellipse (0.1 and 0.1);
\end{scope}
\begin{scope}[shift={(5,0)}]  
\draw[step=1cm,black!20,thin] (-1.4,-1.4) grid (1.4,3.4); 
\draw (-1, 0) -- (0, -1) -- (2,3) -- (-1,0);
\filldraw [black]  (-1,0) ellipse (0.1 and 0.1);
\filldraw [black]  (0,-1) ellipse (0.1 and 0.1);
\filldraw [black]  (0,0) ellipse (0.1 and 0.1);
\filldraw [black]  (1,1) ellipse (0.1 and 0.1);
\filldraw [black]  (1,2) ellipse (0.1 and 0.1);
\filldraw [black]  (0,1) ellipse (0.1 and 0.1);
\filldraw [black]  (2,3) ellipse (0.1 and 0.1);
\end{scope}
\end{tikzpicture}
\end{center}
\caption{Toric diagram of $\IC^3/\IZ_5$  with action $(\omega,\omega,\omega^3)$ (left), $\IC^3/\IZ_6$  with action $(\omega,\omega,\omega^4)$ (center) and $\IC^3/\IZ_6$ with action $(\omega,\omega^2,\omega^3)$  (right). The first two have two internal lattice points, hence two compact divisors, while the third has one internal lattice point and corresponds to the almost Fano surface $PdP_{3a}$.  \label{toric56}}
\end{figure}


Let us consider the action of $\bZ_N$ on $\bC^3$ given by
\begin{equation}
1\mto\diag(\om,\om,\om^{-2}),\qquad \om=e^{2\pi\bi/N}.
\end{equation}
The case $N=2$ corresponds to a small crepant resolution $\IC^2/\IZ_2\times \IC$ already discussed in \S\ref{sec:inv of quotients}, while  $N=3$ and $N=4$ were considered in \S\ref{secP2} and \S\ref{secF2}, respectively.
For $N=5$, our computations for low dimension vectors indicate that
\begin{equation}
\Om_*(e_i)=1,\qquad \Om_*(n\de)=-y\inv(y^4+2y^2+2),
\end{equation}
while all other attractor invariants vanish. This is in agreement with the motive computed from 
the toric diagram in Figure \ref{toric56}.
Moreover the Coulomb branch formula gives
$\Om_*(\de)=\OmS(\de)-2y-2/y$, in agreement with \eqref{OmSgen}.

For $N=6$, we find instead 
\begin{gather}
\Om_*(e_i)=1,\qquad \Om_*(n\de)=-y\inv(y^2+1)(y^2+2),\qquad n\ge1,\nn\\
\Om_*(e_i+e_{i+2}+e_{i+4}+n\de)=-y,\qquad n\ge0,
\end{gather}
while all other attractor invariants vanish. The value of $\Om_*(n\de)$ is in agreement with the motive computed from  the toric diagram in Figure \ref{toric56}.
The Coulomb branch formula gives again
$\Om_*(\de)=\OmS(\de)-2y-2/y$, in agreement with \eqref{OmSgen}.

\begin{remark}
The toric diagram of the quotient $\cX_N=\bC^3/\bZ_N$ was described in \S\ref{sec:general action}. 
Applying Lemma \ref{lm:Ptoric}, we obtain the motive of a crepant resolution $\widetilde \cX_N$ of 
$\bC^3/\bZ_N$,
\begin{equation}
[\widetilde \cX_N]=\begin{cases}
q(q^2+kq+k-1)&N=2k,\\
q(q^2+kq+k)&N=2k+1.
\end{cases}
\end{equation}
Then $\Om_*(n\de)=(-y)^{-3}[\wtl\cX_N]$ as explained earlier.
\end{remark}

\begin{conjecture}
For $N\ge3$, consider the action of $\bZ_N$ on $\bC^3$ given by
$1\mto\diag(\om,\om,\om^{-2})$, $\om=e^{2\pi\bi/N}$.
Then the attractor invariants of the corresponding quiver with potential are
\bea
\Om_*(e_i) &=&1,\qquad \Om_*(n\de)=(-y)^{-3}[\wtl\cX_N],
\qquad n\ge1, \nn \\
\Om_*(e_i+e_{i+2}+\dots+e_{i-2})&=&-y,\qquad n\ge0,
\text{ even }N.
\eea
All other attractor invariants vanish.
\end{conjecture}


Finally, let  us consider the action of $\bZ_6$ on $\bC^3$ given by
\begin{equation}
1\mto\diag(\om,\om^2,\om^3),\qquad \om=e^{2\pi\bi/6}.
\end{equation}
A crepant resolution of $\cX=\IC^3/\IZ_6$ is the canonical bundle over the almost Fano surface $\PdP_{3a}$, in the notation of \cite{Beaujard:2020sgs}. 
Our computations indicate that 
\bea
\Om_*(e_i) &=&1,\qquad
\Om_*(n\de)=-y\inv(y^4+4y^2+1),
\\
\Om_*(e_i+e_{i+2}+e_{i+4}+n\de)
&=&\Om_*(e_i+e_{i+3}+n\de)
=\Om_*(e_i+e_{i+1}+e_{i+3}+e_{i+4}+n\de)=-y \nn
\eea
and all other attractor invariants vanish. The Coulomb branch formula gives
$\Om_*(\de)=\OmS(\de)-y-1/y$, in agreement with \eqref{OmSgen}.

\bibliography{attractor_combined}
\bibliographystyle{hamsplain}
\end{document}